\documentclass{llncs}
\usepackage[T1]{fontenc}
\usepackage[utf8]{inputenc}
\usepackage{amssymb,amsmath,graphicx,graphics,enumerate,tkz-graph}
\usepackage{microtype}
\usepackage{hyperref}

\hypersetup{hypertexnames=false}

\title{Clique-Width for Graph Classes Closed under Complementation\thanks{This paper received support from EPSRC (EP/K025090/1 and EP/L020408/1) and the Leverhulme Trust (RPG-2016-258).}}
\author{Alexandre Blanch\'e\inst{1} \and Konrad K. Dabrowski\inst{2} \and Matthew Johnson\inst{2} \and Vadim~V.~Lozin\inst{3} \and Dani\"el Paulusma\inst{2} \and Viktor Zamaraev\inst{3}}

\institute{
\'Ecole normale sup\'erieure de Rennes,\\
D\'epartement Informatique et T\'el\'ecommunications,\\
Campus de Ker Lann, Avenue Robert Schuman, 35170 Bruz, France\\
\texttt{alexandre.blanche@ens-rennes.fr}\\
\and
School of Engineering and Computing Sciences, Durham University
Science Laboratories, South Road, Durham DH1 3LE, United Kingdom\\
\texttt{\{konrad.dabrowski,matthew.johnson2,daniel.paulusma\}@durham.ac.uk}\\
\and
Mathematics Institute, University of Warwick, Coventry CV4 7AL, United Kingdom\\
\texttt{\{v.lozin,v.zamaraev\}@warwick.ac.uk}}

\newcommand{\ssi}{\subseteq_i}
\newcommand{\si}{\supseteq_i}

\newcommand{\NP}{{\sf NP}}

\DeclareMathOperator{\cw}{cw}

\newcounter{ctrclaim}[theorem]

\newcounter{ctrcase}[theorem]
\newcounter{ctrsubcase}[ctrcase]

\newtheorem{oproblem}{Open Problem}

\newcommand\displaycase[1]{{\bf #1}}
\newcommand{\qedllncs}{\qed}
\newcommand\faketheorem[1]{{\bf #1}}

\newcommand{\clm}[1]{\setcounter{ctrcase}{0}\medskip\phantomsection\refstepcounter{ctrclaim}\noindent\displaycase{Claim \thectrclaim. }{\em #1}\\}
\newcommand{\clmnonewline}[1]{\setcounter{ctrcase}{0}\medskip\phantomsection\refstepcounter{ctrclaim}\noindent\displaycase{Claim \thectrclaim. }{\em #1}}
\newcommand{\shortclm}[1]{\setcounter{ctrcase}{0}\phantomsection\refstepcounter{ctrclaim}\noindent\displaycase{Claim \thectrclaim. }{\em #1}}
\newcommand{\thmcase}[1]{\medskip\phantomsection\refstepcounter{ctrcase}\noindent\displaycase{Case \thectrcase. }{\em #1}\\}

\begin{document}
\maketitle
\begin{abstract}
Clique-width is an important graph parameter due to its algorithmic and structural properties. A graph class is hereditary if it can be characterized by a (not necessarily finite) set~${\cal H}$ of forbidden induced subgraphs. 
We initiate a systematic study into the boundedness of clique-width of hereditary graph classes closed under complementation.
First, we extend the known classification for the $|{\cal H}|=1$ case by classifying the boundedness of clique-width for every set~${\cal H}$ of self-complementary graphs.
We then completely settle the $|{\cal H}|=2$ case.
In particular, we determine one new class of $(H,\overline{H})$-free graphs of bounded clique-width (as a side effect, this leaves only six classes of $(H_1,H_2)$-free graphs, for which it is not known whether their clique-width is bounded).
Once we have obtained the classification of the $|{\cal H}|=2$ case, we research the effect of forbidding self-complementary graphs on the boundedness of clique-width.
Surprisingly, we show that for a set~${\cal F}$ of self-complementary graphs on at least five vertices, the classification of the boundedness of clique-width for $(\{H,\overline{H}\}\cup {\cal F})$-free graphs coincides with the one for the $|{\cal H}|=2$ case if and only if~${\cal F}$ does not include the bull (the only non-empty self-complementary graphs on fewer than five vertices are~$P_1$ and~$P_4$, and $P_4$-free graphs have clique-width at most~$2$).
Finally, we discuss the consequences of our results for the {\sc Colouring} problem.
\end{abstract}

\section{Introduction}\label{sec:into}
Many graph-theoretic problems that are computationally hard for general graphs may still be solvable in polynomial time if the input graph
can be decomposed into large parts of ``similarly behaving'' vertices.
Such decompositions may lead to an algorithmic speed up and are often defined via some type of graph construction.
One particular type is to use vertex labels and to allow certain graph operations, which ensure that vertices labelled alike will always keep the same label and thus behave identically.
The clique-width~$\cw(G)$ of a graph~$G$ is the minimum number of different labels needed to construct~$G$ using four such operations (see Section~\ref{sec:prelim} for details).
Clique-width has been studied extensively both in algorithmic and structural graph theory.
The main reason for its popularity is that, indeed, many well-known \NP-hard problems~\cite{CMR00,EGW01,KR03b,Ra07}, such as {\sc Graph Colouring} and {\sc Hamilton Cycle}, become polynomial-time solvable on any graph class~${\cal G}$ of {\it bounded} clique-width, that is, for which there exists a constant~$c$, such that every graph in~${\cal G}$ has clique-width at most~$c$.
{\sc Graph Isomorphism} is also polynomial-time solvable on such graph classes~\cite{GS15}.
Having bounded clique-width is equivalent to having bounded rank-width~\cite{OS06} and having bounded NLC-width~\cite{Johansson98}, two other well-known width-parameters.
However, despite these close relationships, clique-width is a notoriously difficult graph parameter, and our understanding of it is still very limited.
For instance, no polynomial-time algorithms are known for computing the clique-width of very restricted graph classes, such as unit interval graphs, or for deciding whether a graph has clique-width at most~$4$.\footnote{It is known that computing clique-width is \NP-hard in general~\cite{FRRS09} and that deciding whether a graph has clique-width at most~$3$ is polynomial-time solvable~\cite{CHLRR12}.}
In order to get a better understanding of clique-width and to identify new ``islands of tractability'' for central \NP-hard problems, many graph classes of bounded and unbounded clique-width have been identified; see, for instance, the Information System on Graph Classes and their Inclusions~\cite{isgci}, which keeps a record of such graph classes.
In this paper we study the following research question:

\smallskip
\noindent
{\it What kinds of properties of a graph class ensure that its clique-width is bounded?}

\smallskip
\noindent
We refer to the surveys~\cite{Gu07,KLM09} for examples of such properties.
Here, we consider graph complements.
The complement~$\overline{G}$ of a graph~$G$ is the graph with vertex set~$V_G$ and edge set $\{uv\; |\; uv\notin E(G)\}$ and has clique-width $\cw(\overline{G})\leq 2\cw(G)$~\cite{CO00}.
This result implies that a graph class~${\cal G}$ has bounded clique-width if and only if the class consisting of all complements of graphs in~${\cal G}$ has bounded clique-width.
Due to this, we initiate a {\it systematic} study of the boundedness of clique-width for graph classes~${\cal G}$ {\it closed under complementation}, that is, for every graph~$G\in {\cal G}$, its complement~$\overline{G}$ also belongs to~${\cal G}$.

To get a handle on graph classes closed under complementation, we restrict ourselves to graph classes that are not only closed under complementation but also under vertex deletion.
This is a natural assumption, as deleting a vertex does not increase the clique-width of a graph.
A graph class closed under vertex deletion is said to be {\it hereditary} and can be characterized by a (not necessarily finite) set~${\cal H}$ of forbidden induced subgraphs.
Over the years many results on the (un)boundedness of clique-width of hereditary graph classes have appeared.
We briefly survey some of these results below.

\begin{sloppypar}
A hereditary graph class of graphs is monogenic or $H$-free if it can be characterized by one forbidden induced subgraph~$H$, and bigenic or $(H_1,H_2)$-free if it can be characterized by two forbidden induced subgraphs~$H_1$ and~$H_2$.
It is well known (see~\cite{DP15}) that a class of $H$-free graphs has bounded clique-width if and only if~$H$ is an induced subgraph of~$P_4$.\footnote{We refer to Section~\ref{sec:prelim} for all the notation used in this section.}
By combining known results~\cite{BL02,BDHP15,BELL06,BKM06,BLM04b,BLM04,BM02,DGP14,DHP0,DLRR12,MR99} with new results for bigenic graph classes, Dabrowski and Paulusma~\cite{DP15} classified the (un)boundedness of clique-width of $(H_1,H_2)$-free graphs for all but 13 pairs $(H_1,H_2)$ (up to an equivalence relation).
Afterwards, five new classes of $(H_1,H_2)$-free graphs were identified by Dross et al.~\cite{DDP15} and recently, another one was identified by Dabrowski et al.~\cite{DLP}.
This means that only seven cases $(H_1,H_2)$ remained open.
Other systematic studies were performed for $H$-free weakly chordal graphs~\cite{BDHP15}, $H$-free chordal graphs~\cite{BDHP15} (two open cases), $H$-free triangle-free graphs~\cite{DLP} (two open cases), $H$-free bipartite graphs~\cite{DP14}, $H$-free split graphs~\cite{BDHP15b} (two open cases), and ${\cal H}$-free graphs where~${\cal H}$ is any set of $1$-vertex extensions of the~$P_4$~\cite{BDHR05} or any set of graphs on at most four vertices~\cite{BELL06}.
Clique-width results or techniques for these graph classes impacted upon each other and could also be used for obtaining new results for bigenic graph classes.
\end{sloppypar}

\medskip
\noindent
{\bf Our Contribution.}
Recall that we investigate the clique-width of hereditary graph classes closed under complementation.
A graph that contains no induced subgraph isomorphic to a graph in a set~${\cal H}$ is said to be {\em ${\cal H}$-free}.
We first consider the $|{\cal H}|=1$ case.
The class of $H$-free graphs is closed under complementation if and only if~$H$ is a self-complementary graph, that is, $H=\overline{H}$.
Self-complementary graphs have been extensively studied; see~\cite{Fa99} for a survey.
From the aforementioned result for $P_4$-free graphs, we find that the only self-complementary graphs~$H$ for which the class of $H$-free graphs has bounded clique-width are $H=P_1$ and $H=P_4$.
In Section~\ref{s-self} we prove the following generalization of this result.

\begin{theorem}\label{t-self}
Let~${\cal H}$ be a
set of non-empty self-complementary graphs.
Then the class of ${\cal H}$-free graphs has bounded clique-width if and only if either $P_1 \in {\cal H}$ or $P_4 \in {\cal H}$.
\end{theorem}

We now consider the $|{\cal H}|=2$ case.
Let ${\cal H}=\{H_1,H_2\}$.
Due to Theorem~\ref{t-self} we may assume $H_2=\overline{H_1}$ and~$H_1$ is not self-complementary.
The class of $(2P_1+\nobreak P_3,\overline{2P_1+P_3})$-free graphs was one of the seven remaining
bigenic graph classes, and the only bigenic graph class closed under complementation, for which boundedness of clique-width was open.
We settle this case by proving in Section~\ref{s-main} that the clique-width of this class is bounded.
In the same section we combine this new result with known results to prove the following theorem, which, together with Theorem~\ref{t-self}, shows to what extent the property of being closed under complementation helps with bounding the clique-width for bigenic graph classes (see also \figurename~\ref{fig:H-co-H-bdd-cw}).

\begin{theorem}\label{t-main}
For a graph~$H$, the class of $(H,\overline{H})$-free graphs has bounded clique-width if and only if~$H$ or~$\overline{H}$ is an induced subgraph of $K_{1,3}$, $P_1+P_4$, $2P_1+P_3$ or~$sP_1$ for some $s\geq 1$.
\end{theorem}

\begin{figure}
\begin{center}
\begin{tabular}{cccccccc}
\begin{minipage}{0.11\textwidth}
\centering
\scalebox{0.6}{
{\begin{tikzpicture}[scale=1,rotate=90]
\GraphInit[vstyle=Simple]
\SetVertexSimple[MinSize=6pt]
\Vertices{circle}{a,b,c}
\Vertex[x=0,y=0]{d}
\Edges(a,d,b)
\Edges(c,d)
\end{tikzpicture}}}
\end{minipage}
&
\begin{minipage}{0.11\textwidth}
\centering
\scalebox{0.6}{
{\begin{tikzpicture}[scale=1,rotate=90]
\GraphInit[vstyle=Simple]
\SetVertexSimple[MinSize=6pt]
\Vertices{circle}{a,b,c}
\Vertex[x=0,y=0]{d}
\Edges(a,b,c,a)
\end{tikzpicture}}}
\end{minipage}
&
\begin{minipage}{0.11\textwidth}
\centering
\scalebox{0.6}{
{\begin{tikzpicture}[scale=1,rotate=90]
\GraphInit[vstyle=Simple]
\SetVertexSimple[MinSize=6pt]
\Vertices{circle}{a,b,c,d,e}
\Edges(b,c,d,e)
\end{tikzpicture}}}
\end{minipage}
&
\begin{minipage}{0.11\textwidth}
\centering
\scalebox{0.6}{
{\begin{tikzpicture}[scale=1,rotate=90]
\GraphInit[vstyle=Simple]
\SetVertexSimple[MinSize=6pt]
\Vertices{circle}{a,b,c,d,e}
\Edges(b,c,d,e)
\Edges(b,a,c)
\Edges(d,a,e)
\end{tikzpicture}}}
\end{minipage}
&
\begin{minipage}{0.11\textwidth}
\centering
\scalebox{0.6}{
{\begin{tikzpicture}[scale=1,rotate=90]
\GraphInit[vstyle=Simple]
\SetVertexSimple[MinSize=6pt]
\Vertices{circle}{a,b,c,d,e}
\Edges(e,a,b)
\end{tikzpicture}}}
\end{minipage}
&
\begin{minipage}{0.11\textwidth}
\centering
\scalebox{0.6}{
{\begin{tikzpicture}[scale=1,rotate=90]
\GraphInit[vstyle=Simple]
\SetVertexSimple[MinSize=6pt]
\Vertices{circle}{a,b,c,d,e}
\Edges(e,a,b)
\Edges(b,c,d,e,b)
\Edges(b,d)
\Edges(c,e)
\end{tikzpicture}}}
\end{minipage}
&
\begin{minipage}{0.11\textwidth}
\centering
\scalebox{0.6}{
{\begin{tikzpicture}[scale=1,rotate=90]
\GraphInit[vstyle=Simple]
\SetVertexSimple[MinSize=6pt]
\Vertices{circle}{a,b,c,d,e}
\end{tikzpicture}}}
\end{minipage}
&
\begin{minipage}{0.11\textwidth}
\centering
\scalebox{0.6}{
{\begin{tikzpicture}[scale=1,rotate=90]
\GraphInit[vstyle=Simple]
\SetVertexSimple[MinSize=6pt]
\Vertices{circle}{a,b,c,d,e}
\Edges(a,b,c,d,e,a,c,e,b,d,a)
\end{tikzpicture}}}
\end{minipage}\\
\\
$K_{1,3}$ &
$\overline{K_{1,3}}$ &
$P_1+\nobreak P_4$ &
$\overline{P_1+\nobreak P_4}$ &
$2P_1+\nobreak P_3$ &
$\overline{2P_1+\nobreak P_3}$ &
$sP_1$ &
$\overline{sP_1}$\\
\end{tabular}
\end{center}
\caption{\label{fig:H-co-H-bdd-cw}Graphs~$H$ for which the clique-width of $(H,\overline{H})$-free graphs is bounded.
For~$sP_1$ and~$\overline{sP_1}$ the $s=5$ case is shown.}
\end{figure}
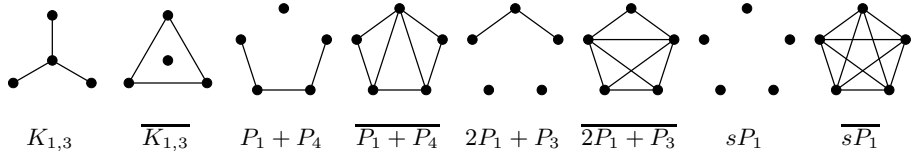

For the $|{\cal H}|=3$ case, where $\{H_1,H_2,H_3\}={\cal H}$, we observe that a class of $(H_1,H_2,H_3)$-free graphs is closed under complementation if and only if either every~$H_i$ is self-complementary, or one~$H_i$ is self-complementary and the other two graphs~$H_j$ and~$H_k$ are complements of each other.
By Theorem~\ref{t-self}, we only need to consider $(H_1,\overline{H_1},H_2)$-free graphs, where~$H_1$ is not self-complementary, $H_2$ is self-complementary, and neither~$H_1$ nor~$H_2$ is an induced subgraph of~$P_4$.
The next two smallest self-complementary graphs~$H_2$ are the~$C_5$ and the bull (see also \figurename~\ref{fig:small-self-comp}).
Observe that any self-complementary graph on~$n$ vertices must contain $\frac{1}{2}\binom{n}{2}$ edges and this number must be an integer, so $n=4q$ or $n=4q+1$ for some integer $q\geq 0$.
There are exactly ten non-isomorphic self-complementary graphs on eight vertices~\cite{Re63} and we depict these in Figure~\ref{fig:self-complementary-8-vertices}.

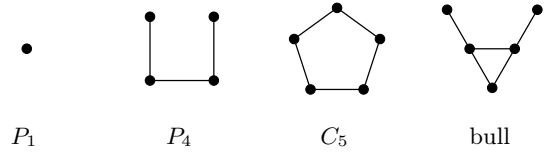
\begin{figure}
\begin{center}
\begin{tabular}{cccc}
\begin{minipage}{0.15\textwidth}
\centering
\scalebox{0.6}{
{\begin{tikzpicture}[scale=1,rotate=90]
\GraphInit[vstyle=Simple]
\SetVertexSimple[MinSize=6pt]
\Vertex[x=0,y=0]{a}
\end{tikzpicture}}}
\end{minipage}
&
\begin{minipage}{0.15\textwidth}
\centering
\scalebox{0.6}{
{\begin{tikzpicture}[scale=1,rotate=135]
\GraphInit[vstyle=Simple]
\SetVertexSimple[MinSize=6pt]
\Vertices{circle}{a,b,c,d}
\Edges(a,b,c,d)
\end{tikzpicture}}}
\end{minipage}
&
\begin{minipage}{0.15\textwidth}
\centering
\scalebox{0.6}{
{\begin{tikzpicture}[scale=1,rotate=90]
\GraphInit[vstyle=Simple]
\SetVertexSimple[MinSize=6pt]
\Vertices{circle}{a,b,c,d,e}
\Edges(a,b,c,d,e,a)
\end{tikzpicture}}}
\end{minipage}
&
\begin{minipage}{0.15\textwidth}
\centering
\scalebox{0.6}{
{\begin{tikzpicture}[scale=1,rotate=90]
\GraphInit[vstyle=Simple]
\SetVertexSimple[MinSize=6pt]
\Vertex[x=0,y=0]{a}
\Vertex[a=30,d=1]{b}
\Vertex[a=30,d=2]{c}
\Vertex[a=-30,d=1]{d}
\Vertex[a=-30,d=2]{e}
\Edges(c,b,a,d,e)
\Edges(b,d)
\end{tikzpicture}}}
\end{minipage}\\
\\
$P_1$ & $P_4$ & $C_5$ & bull
\end{tabular}
\end{center}
\caption{\label{fig:small-self-comp}The four non-empty self-complementary graphs on less than eight vertices~\cite{Re63}.}
\end{figure}

\begin{figure}
\begin{center}
{
\tikzstyle{vertex}=[circle,fill=black!100,text=white,inner sep=0.4mm,draw]

\begin{tabular}{ccccc}
\begin{minipage}{0.19\textwidth}
\begin{center}
\begin{tikzpicture}[scale=0.35]
\node[vertex] (1) at (0,0) {};
\node[vertex] (2) at (3,0) {};
\node[vertex] (3) at (0,3) {};
\node[vertex] (4) at (3,3) {};
\node[vertex] (5) at (1.5,-0.75) {};
\node[vertex] (6) at (-0.75,1.5) {};
\node[vertex] (7) at (3.75,1.5) {};
\node[vertex] (8) at (1.5,3.75) {};

\foreach \from/\to in {1/2,1/3,3/4,2/4,1/5,1/6,2/5,2/7,4/7,4/8,3/6,3/8,1/4,2/3}
\draw (\from) -- (\to);
\end{tikzpicture}
\end{center}
\end{minipage}
&
\begin{minipage}{0.19\textwidth}
\begin{center}
\begin{tikzpicture}[scale=0.75]
\node[vertex] (1) at (0,0) {};
\node[vertex] (2) at (1,0) {};
\node[vertex] (3) at (2,0) {};
\node[vertex] (4) at (3,0) {};
\node[vertex] (5) at (0,1) {};
\node[vertex] (6) at (1,1) {};
\node[vertex] (7) at (2,1) {};
\node[vertex] (8) at (3,1) {};

\foreach \from/\to in {1/2,2/3,3/4,5/1,5/2,5/3,6/2,7/3,8/2,8/3,8/4}
\draw (\from) -- (\to);

\draw (1.south) to [bend right] (4.south);
\draw (1.south) to [bend right] (3.south);
\draw (2.south) to [bend right] (4.south);
\end{tikzpicture}
\end{center}
\end{minipage}
&
\begin{minipage}{0.19\textwidth}
\begin{center}
\begin{tikzpicture}[scale=0.75]
\node[vertex] (1) at (0,0) {};
\node[vertex] (2) at (1,0) {};
\node[vertex] (3) at (2,0) {};
\node[vertex] (4) at (3,0) {};
\node[vertex] (5) at (0,1) {};
\node[vertex] (6) at (1,1) {};
\node[vertex] (7) at (2,1) {};
\node[vertex] (8) at (3,1) {};

\foreach \from/\to in {1/2,2/3,3/4,5/1,5/2,6/1,6/2,7/3,7/4,8/3,8/4}
\draw (\from) -- (\to);

\draw (1.south) to [bend right] (4.south);
\draw (1.south) to [bend right] (3.south);
\draw (2.south) to [bend right] (4.south);
\end{tikzpicture}
\end{center}
\end{minipage}
&
\begin{minipage}{0.19\textwidth}
\begin{center}
\begin{tikzpicture}[scale=0.65]
\node[vertex] (1) at (0,0) {};
\node[vertex] (2) at (2,0) {};
\node[vertex] (3) at (-0.5,-1) {};
\node[vertex] (4) at (1.5,-1) {};
\node[vertex] (5) at (0.2,-2) {};
\node[vertex] (6) at (2.2,-2) {};
\node[vertex] (7) at (-0.3,-3) {};
\node[vertex] (8) at (1.7,-3) {};

\foreach \from/\to in {1/2,1/3,2/4,3/4,5/6,7/8,5/7,6/8,7/8,3/5,4/6,1/5,2/6,3/7,4/8}
\draw (\from) -- (\to);
\end{tikzpicture}
\end{center}
\end{minipage}
&
\begin{minipage}{0.19\textwidth}
\begin{center}
\begin{tikzpicture}[scale=0.65]
\node[vertex] (1) at (0,0) {};
\node[vertex] (2) at (2,0) {};
\node[vertex] (3) at (0,-3) {};
\node[vertex] (4) at (2,-3) {};
\node[vertex] (5) at (-0.3,-1.5) {};
\node[vertex] (6) at (0.3,-1.5) {};
\node[vertex] (7) at (1.7,-1.5) {};
\node[vertex] (8) at (2.3,-1.5) {};

\foreach \from/\to in {1/2, 3/4, 1/5, 1/6, 3/5, 3/6, 2/7, 2/8, 4/7, 4/8, 1/4, 2/3, 5/6, 7/8}
\draw (\from) -- (\to);
\end{tikzpicture}
\end{center}
\end{minipage}\\
$X_1$ & $X_2$ & $X_3$ & $X_4$ & $X_5$ \\
\\

\begin{minipage}{0.19\textwidth}
\begin{center}
\begin{tikzpicture}[scale=0.65]
\node[vertex] (1) at (0,0) {};
\node[vertex] (2) at (2,0) {};
\node[vertex] (3) at (0,-3) {};
\node[vertex] (4) at (2,-3) {};
\node[vertex] (5) at (-0.3,-1.5) {};
\node[vertex] (6) at (0.3,-1.5) {};
\node[vertex] (7) at (1.7,-1.5) {};
\node[vertex] (8) at (2.3,-1.5) {};

\foreach \from/\to in {1/2, 1/5, 1/6, 3/5, 3/6, 2/7, 2/8, 4/7, 4/8, 1/4, 2/3, 5/6, 7/8, 6/7}
\draw (\from) -- (\to);
\end{tikzpicture}
\end{center}
\end{minipage}
&
\begin{minipage}{0.19\textwidth}
\begin{center}
\begin{tikzpicture}[scale=0.65]
\node[vertex] (1) at (0,0) {};
\node[vertex] (2) at (2,0) {};
\node[vertex] (3) at (0,-3) {};
\node[vertex] (4) at (2,-3) {};
\node[vertex] (5) at (-0.3,-1.5) {};
\node[vertex] (6) at (0.3,-1.5) {};
\node[vertex] (7) at (1.7,-1.5) {};
\node[vertex] (8) at (2.3,-1.5) {};

\foreach \from/\to in {1/2, 1/5, 1/6, 3/5, 3/6, 2/7, 2/8, 4/7, 4/8, 1/4, 2/3, 2/3, 1/3, 2/4, 6/7}
\draw (\from) -- (\to);
\end{tikzpicture}
\end{center}
\end{minipage}
&
\begin{minipage}{0.19\textwidth}
\begin{center}
\begin{tikzpicture}[scale=0.5]
\node[vertex] (1) at (0,0) {};
\node[vertex] (2) at (3,0) {};
\node[vertex] (3) at (0,3) {};
\node[vertex] (4) at (3,3) {};
\node[vertex] (5) at (1.5,0.75) {};
\node[vertex] (6) at (0.75,1.5) {};
\node[vertex] (7) at (2.25,1.5) {};
\node[vertex] (8) at (1.5,2.25) {};

\foreach \from/\to in {1/2,1/3,3/4,2/4,1/5,1/6,2/5,2/7,4/7,4/8,3/6,3/8,5/8,6/7}
\draw (\from) -- (\to);
\end{tikzpicture}
\end{center}
\end{minipage}
&
\begin{minipage}{0.19\textwidth}
\begin{center}
\begin{tikzpicture}[scale=0.5]
\node[vertex] (1) at (0,0) {};
\node[vertex] (2) at (3,0) {};
\node[vertex] (3) at (0,3) {};
\node[vertex] (4) at (3,3) {};
\node[vertex] (5) at (1.5,0.75) {};
\node[vertex] (6) at (0.75,1.5) {};
\node[vertex] (7) at (2.25,1.5) {};
\node[vertex] (8) at (1.5,2.25) {};

\foreach \from/\to in {1/2,1/3,3/4,2/4,1/5,1/6,2/5,2/7,4/7,4/8,3/6,3/8,6/7,2/3}
\draw (\from) -- (\to);
\end{tikzpicture}
\end{center}
\end{minipage}
&
\begin{minipage}{0.19\textwidth}
\begin{center}
\begin{tikzpicture}[scale=0.5]
\node[vertex] (1) at (0,0) {};
\node[vertex] (2) at (1,0) {};
\node[vertex] (3) at (2,0) {};
\node[vertex] (4) at (3,0) {};
\node[vertex] (5) at (4,1.5) {};
\node[vertex] (6) at (4,0.5) {};
\node[vertex] (7) at (4,-0.5) {};
\node[vertex] (8) at (4,-1.5) {};

\foreach \from/\to in {1/2,2/3,3/4,5/1,5/2,5/3,5/4,5/6,6/7,7/8,8/1,8/2,8/3,8/4}
\draw (\from) -- (\to);
\end{tikzpicture}
\end{center}
\end{minipage}\\
$X_6$ & $X_7$ & $X_8$ & $X_9$ & $X_{10}$
\end{tabular}
}
\end{center}
\caption{\label{fig:self-complementary-8-vertices}The ten self-complementary graphs on eight vertices~\cite{Re63}.}
\end{figure}
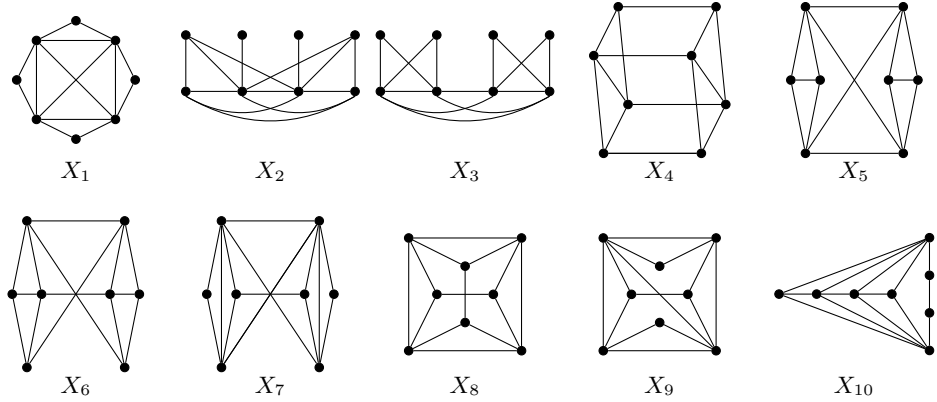

It is known that split graphs, or equivalently, $(2P_2,\overline{2P_2},C_5)$-free graphs have unbounded clique-width~\cite{MR99}.
In Section~\ref{s-main2} we determine three new hereditary graph classes of unbounded clique-width, which
imply that the class of $(H,\overline{H},C_5)$-free graphs has unbounded clique-width if
$H \in \{K_{1,3}+\nobreak P_1,2P_2,3P_1+\nobreak P_2,S_{1,1,2}\}$.
By combining this with known results, we discovered that the classification of boundedness of clique-width for $(H,\overline{H},C_5)$-free graphs coincides with the one of Theorem~\ref{t-main}.
This raised the question of whether the same is true for other sets of self-complementary graphs ${\cal F}\neq \{C_5\}$.
If ${\cal F}$ contains the bull, then the answer is negative: by Theorem~\ref{t-main},
both the class of $(S_{1,1,2},\overline{S_{1,1,2}})$-free graphs and the class of $(2P_2,C_4)$-free graphs have unbounded clique-width, but both the class of $(S_{1,1,2},\overline{S_{1,1,2}},\mbox{bull})$-free graphs and even the class of $(P_5,\overline{P_5},\mbox{bull})$-free graphs have bounded clique-width~\cite{BDHR05}.
However, also in Section~\ref{s-main2}, we prove that the bull is the {\em only} exception (apart from the trivial cases when $H'\in \{P_1,P_4\}$ which yield bounded clique-width of $(H,\overline{H},H')$-free graphs for any graph~$H$).

\begin{theorem}\label{t-main2}
Let~${\cal F}$ be a set of self-complementary graphs on at least five vertices not equal to the bull.
For a graph~$H$, the class of $(\{H,\overline{H}\} \cup {\cal F})$-free graphs has bounded clique-width if and only if~$H$ or~$\overline{H}$ is an induced subgraph of $K_{1,3}$, $P_1+P_4$, $2P_1+P_3$ or~$sP_1$ for some $s\geq 1$.
\end{theorem}

\medskip
\noindent
{\bf Structural Consequences.}
Due to our result for $(2P_1+\nobreak P_3,\overline{2P_1+P_3})$-free graphs, we can update the summary of~\cite{DLP} for the clique-width of bigenic graph classes.

\begin{theorem}\label{thm:classification2}
Let~${\cal G}$ be a class of graphs defined by two forbidden induced subgraphs.
Then:
\begin{enumerate}
\item ${\cal G}$ has bounded clique-width if it is equivalent\footnote{Given four graphs $H_1,H_2,H_3,H_4$, the class of $(H_1,H_2)$-free graphs and the class of $(H_3,H_4)$-free graphs are {\em equivalent} if the unordered pair $H_3,H_4$ can be obtained from the unordered pair $H_1,H_2$ by some combination of the operations (i) complementing both graphs in the pair and (ii) if one of the graphs in the pair is~$K_3$, replacing it with $\overline{P_1+P_3}$ or vice versa.
If two classes are equivalent, then one of them has bounded clique-width if and only if the other one does (see~\cite{DP15}).}
to a class of $(H_1,H_2)$-free graphs such that one of the following holds:\\
\begin{enumerate}[(i)]
\item \label{thm:classification2:bdd:P4} $H_1$ or $H_2 \ssi P_4$;
\item \label{thm:classification2:bdd:ramsey} $H_1=sP_1$ and $H_2=K_t$ for some $s,t$;
\item \label{thm:classification2:bdd:P_1+P_3} $H_1 \ssi P_1+\nobreak P_3$ and $\overline{H_2} \ssi K_{1,3}+\nobreak 3P_1,\; K_{1,3}+\nobreak P_2,\;\allowbreak P_1+\nobreak P_2+\nobreak P_3,\;\allowbreak P_1+\nobreak P_5,\;\allowbreak P_1+\nobreak S_{1,1,2},\;\allowbreak P_2+\nobreak P_4,\;\allowbreak P_6,\; \allowbreak S_{1,1,3}$ or~$S_{1,2,2}$;
\item \label{thm:classification2:bdd:2P_1+P_2} $H_1 \ssi 2P_1+\nobreak P_2$ and $\overline{H_2}\ssi P_1+\nobreak 2P_2,\; 3P_1+\nobreak P_2$ or~$P_2+\nobreak P_3$;
\item \label{thm:classification2:bdd:P_1+P_4} $H_1 \subseteq_i P_1+\nobreak P_4$ and $\overline{H_2} \ssi P_1+\nobreak P_4$ or~$P_5$;
\item \label{thm:classification2:bdd:K_13} $H_1,\overline{H_2} \ssi K_{1,3}$;
\item \label{thm:classification2:bdd:2P1_P3} $H_1,\overline{H_2} \ssi 2P_1+\nobreak P_3$.\\
\end{enumerate}
\item ${\cal G}$ has unbounded clique-width if it is equivalent to a class of $(H_1,H_2)$-free graphs such that one of the following holds:
\begin{enumerate}[(i)]
\item \label{thm:classification2:unbdd:not-in-S} $H_1\not\in {\cal S}$ and $H_2 \not \in {\cal S}$;
\item \label{thm:classification2:unbdd:not-in-co-S} $\overline{H_1}\notin {\cal S}$ and $\overline{H_2} \not \in {\cal S}$;
\item \label{thm:classification2:unbdd:K_13or2P_2} $H_1 \si K_{1,3}$ or~$2P_2$ and $\overline{H_2} \si 4P_1$ or~$2P_2$;
\item \label{thm:classification2:unbdd:2P_1+P_2} $H_1 \si 2P_1+\nobreak P_2$ and $\overline{H_2} \si K_{1,3},\; 5P_1,\; P_2+\nobreak P_4$ or~$P_6$;
\item \label{thm:classification2:unbdd:3P_1} $H_1 \si 3P_1$ and $\overline{H_2} \si 2P_1+\nobreak 2P_2,\; 2P_1+\nobreak P_4,\; 4P_1+\nobreak P_2,\; 3P_2$ or~$2P_3$;
\item \label{thm:classification2:unbdd:4P_1} $H_1 \si 4P_1$ and $\overline{H_2} \si P_1 +\nobreak P_4$ or~$3P_1+\nobreak P_2$.
\end{enumerate}
\end{enumerate}
\end{theorem}
Theorem~\ref{thm:classification2} leaves us with the following six
open cases.
\begin{oproblem}\label{oprob:twographs}
Does the class of $(H_1,H_2)$-free graphs have bounded or unbounded clique-width when:\\
\begin{enumerate}[(i)]
\item \label{oprob:twographs:3P_1} $H_1=3P_1$ and $\overline{H_2} \in \{P_1+\nobreak S_{1,1,3},
\allowbreak S_{1,2,3}\}$;
\item\label{oprob:twographs:2P_1+P_2} $H_1=2P_1+\nobreak P_2$ and $\overline{H_2} \in \{P_1+\nobreak P_2+\nobreak P_3,\allowbreak P_1+\nobreak P_5\}$;
\item \label{oprob:twographs:P_1+P_4} $H_1=P_1+\nobreak P_4$ and $\overline{H_2} \in \{P_1+\nobreak 2P_2,\allowbreak P_2+\nobreak P_3\}$.
\end{enumerate}
\end{oproblem}

\medskip
\noindent
{\bf Algorithmic Consequences.} 
We show that the class of $(2P_1+P_3,\overline{2P_1+P_3})$-free graphs has bounded clique-width.
As mentioned, this implies that {\sc Colouring} is polynomial-time solvable for this graph class.
This result was used by Blanch\'e et al. to prove the following theorem.

\begin{sloppypar}
\begin{theorem}[\cite{BDJP16}]\label{t-mainmain}
Let $H,\overline{H} \notin \{(s+\nobreak 1)P_1+\nobreak P_3, sP_1+\nobreak P_4\; |\; s\geq 2\}$.
Then {\sc Colouring} is polynomial-time solvable for $(H,\overline{H})$-free graphs if~$H$ or~$\overline{H}$ is an induced subgraph of $K_{1,3}$,
$P_1+\nobreak P_4$,
$2P_1+\nobreak P_3$,
$P_2+\nobreak P_3$,
$P_5$, or $sP_1+\nobreak P_2$ for some $s\geq 0$
and it is \NP-complete otherwise.
\end{theorem}
\end{sloppypar}
Comparing Theorems~\ref{t-main} and~\ref{t-mainmain} shows that there are graph classes of unbounded clique-width that are closed under complementation, but for which {\sc Colouring} is still polynomial-time solvable.
Nevertheless, on many graph classes, polynomial-time solvability of \NP-hard problems stems from the underlying property of having bounded clique-width.
This is also illustrated by our paper, and for {\sc Colouring} we identify, by updating the summary of~\cite{DDP15} (see also~\cite{GJPS}) or by a direct reduction from Open Problem~\ref{oprob:twographs}, the following 12 classes of $(H_1,H_2)$-free graphs, for which {\sc Colouring} could still potentially be solved in polynomial time by showing that their clique-width is bounded.

\begin{oproblem}\label{oprob:twographs2}
Is {\sc Colouring} polynomial-time solvable for $(H_1,H_2)$-free graphs when:\\
\begin{enumerate}[(i)]
\item $\overline{H_1}\in \{3P_1,P_1+P_3\}$ and $H_2\in \{P_1+S_{1,1,3},S_{1,2,3}\}$;\\[-10pt]
\item $H_1=2P_1+\nobreak P_2$ and $\overline{H_2} \in \{P_1+\nobreak P_2+\nobreak P_3, \allowbreak 
P_1+\nobreak P_5\}$;\\[-10pt]
\item $H_1=\overline{2P_1+\nobreak P_2}$ and $H_2 \in \{P_1+\nobreak P_2+\nobreak P_3,\allowbreak 
P_1+\nobreak P_5\}$;\\[-10pt]
\item $H_1=P_1+P_4$ and $\overline{H_2} \in \{P_1+2P_2,P_2+P_3\}$;\\[-10pt]
\item $\overline{H_1}=P_1+P_4$ and ${H_2} \in \{P_1+2P_2,P_2+P_3\}$.
\end{enumerate}
\end{oproblem}

\medskip
\noindent
{\bf Future Work.} Apart from settling the classification of boundedness of clique-width for $(H_1,H_2)$-free graphs by addressing Open Problem~\ref{oprob:twographs}, we aim to continue our study of boundedness of clique-width for graph classes closed under complementation.
In particular, to complete the classification for 
${\cal H}$-free graphs when 
$|{\cal H}|=3$, we still need to determine those graphs~$H$ for which $(H,\overline{H},\mbox{bull})$-free graphs have bounded clique-width (there are several cases left).
It may also be worthwhile to consider the consequences of our research for the boundedness of variants of clique-width, such as linear clique-width~\cite{HMP12} and power-bounded clique-width~\cite{BGMS14}.

\section{Preliminaries}\label{sec:prelim}
Throughout our paper we consider only finite, undirected graphs without multiple edges or self-loops.
Below we define further graph terminology.

Given two graphs~$G$ and~$H$, an {\em isomorphism} from~$G$ to~$H$ is a bijection $f:V(G) \rightarrow V(H)$ such that $uv \in E(G)$ if and only if $f(u)f(v) \in E(H)$.
If such an isomorphism exists, we say that~$G$ and~$H$ are {\em isomorphic}.

The {\em disjoint union} $(V(G)\cup V(H), E(G)\cup E(H))$ of two vertex-disjoint graphs~$G$ and~$H$ is denoted by~$G+\nobreak H$ and the disjoint union of~$r$ copies of a graph~$G$ is denoted by~$rG$.
The {\em complement} of a graph~$G$, denoted by~$\overline{G}$, has vertex set $V(\overline{G})=\nobreak V(G)$ and an edge between two distinct vertices if and only if these two vertices are not adjacent in~$G$.
A graph~$G$ is {\em self-complementary} if~$G$ is isomorphic to~$\overline{G}$.
For a subset $S\subseteq V(G)$, we let~$G[S]$ denote the subgraph of~$G$ {\em induced} by~$S$, which has vertex set~$S$ and edge set $\{uv\; |\; u,v\in S, uv\in E(G)\}$.
If $S=\{s_1,\ldots,s_r\}$ then, to simplify notation, we may also write $G[s_1,\ldots,s_r]$ instead of $G[\{s_1,\ldots,s_r\}]$.
We use $G \setminus S$ to denote the graph obtained from~$G$ by deleting every vertex in~$S$, that is, $G \setminus S = G[V(G)\setminus S]$; if $S=\{v\}$, we may write $G \setminus v$ instead.
We write $G'\ssi G$ to indicate that~$G'$ is an induced subgraph of~$G$.

A graph $G=(V,G)$ is {\em empty} if $V=E=\emptyset$, otherwise it is non-empty.
The graphs $C_r$, $K_r$, $K_{1,r-1}$ and~$P_r$ denote the cycle, complete graph, star and path on~$r$ vertices, respectively.
The graphs~$K_3$ and~$K_{1,3}$ are also called the {\em triangle} and {\em claw}, respectively.
A graph~$G$ is a {\em linear forest} if every component of~$G$ is a path (on at least one vertex).
The graph~$S_{h,i,j}$, for $1\leq h\leq i\leq j$, denotes the {\em subdivided claw}, that is, the tree that has only one vertex~$x$ of degree~$3$ and exactly three leaves, which are of distance~$h$,~$i$ and~$j$ from~$x$, respectively.
Observe that $S_{1,1,1}=K_{1,3}$.
We let ${\cal S}$ be the class of graphs each connected component of which is either a subdivided claw or a path on at least one vertex.

For a set of graphs~${\cal H}$, a graph~$G$ is {\em ${\cal H}$-free} (or $({\cal H})$-free) if it has no induced subgraph isomorphic to a graph in~${\cal H}$.
If ${\cal H}=\{H_1,\ldots,H_p\}$ for some integer~$p$, then we may write $(H_1,\ldots,H_p)$-free instead of $(\{H_1,\ldots,H_p\})$-free, or, if $p=1$, we may simply write $H_1$-free.

For a graph $G=(V,E)$, the set $N(u)=\{v\in V\; |\; uv\in E\}$ denotes the {\em neighbourhood} of $u\in V$.
A component in~$G$ is {\em trivial} if it contains exactly one vertex, otherwise it is non-trivial.
A graph is {\em bipartite} if its vertex set can be partitioned into two (possibly empty) independent sets.
A graph is {\em split} if its vertex set can be partitioned into a (possibly empty) independent set and a (possibly empty) clique.
Split graphs have been characterized as follows.

\begin{lemma}[\cite{FH77}]\label{lem:split-forb-graphs}
A graph~$G$ is split if and only if it is $(2P_2,C_4,C_5)$-free.
\end{lemma}

Let~$X$ be a set of vertices in a graph $G=(V,E)$.
A vertex $y\in V\setminus X$ is {\em complete} to~$X$ if it is adjacent to every vertex of~$X$ and {\em anti-complete} to~$X$ if it is non-adjacent to every vertex of~$X$.
Similarly, a set of vertices $Y\subseteq V\setminus X$ is {\em complete} (resp. {\em anti-complete}) to~$X$ if every vertex in~$Y$ is complete (resp. anti-complete) to~$X$.
We say that the edges between two disjoint sets of vertices~$X$ and~$Y$ form a {\em matching} (resp. {\em co-matching}) if each vertex in~$X$ has at most one neighbour (resp. non-neighbour) in~$Y$ and vice versa.
A vertex $y\in V\setminus X$ {\em distinguishes}~$X$ if~$y$ has both a neighbour and a non-neighbour in~$X$.
The set~$X$ is a {\em module} of~$G$ if no vertex in $V\setminus X$ distinguishes~$X$.
A module~$X$ is {\em non-trivial} if $1<|X|<|V|$, otherwise it is {\em trivial}.
A graph is {\em prime} if it has only trivial modules.

To help reduce the amount of case analysis needed to prove Theorems~\ref{t-main} and~\ref{t-main2}, we prove the following lemma.

\begin{lemma}\label{lem:useful}
Let~$H \in {\cal S}$.
Then~$H$ is $(K_{1,3}+\nobreak P_1,2P_2,3P_1+\nobreak P_2,S_{1,1,2})$-free if and only if~$H$ is an induced subgraph of $K_{1,3}$, $P_1+P_4$, $2P_1+P_3$ or~$sP_1$ for some $s\geq 1$.
\end{lemma}
\begin{proof}
Let~$H \in {\cal S}$.
First suppose~$H$ is an induced subgraph of $K_{1,3}$, $P_1+P_4$, $2P_1+P_3$ or~$sP_1$ for some $s\geq 1$.
It is readily seen that~$H$ is $(K_{1,3}+\nobreak P_1,2P_2,3P_1+\nobreak P_2,S_{1,1,2})$-free.

Now suppose that~$H$ is $(K_{1,3}+\nobreak P_1,2P_2,3P_1+\nobreak P_2,S_{1,1,2})$-free.
If~$H$ is not a linear forest then since $H \in {\cal S}$, it contains an induced subgraph isomorphic to~$K_{1,3}$.
We may assume that~$H$ is not an induced subgraph of~$K_{1,3}$, otherwise we are done.
In this case~$H$ contains an induced subgraph that is a one-vertex extension of~$K_{1,3}$.
Since~$H \in {\cal S}$, this means that~$H$ contains $K_{1,3}+\nobreak P_1$ or~$S_{1,1,2}$ as an induced subgraph, a contradiction.
We may therefore assume that~$H$ is a linear forest.

Since~$H$ is a linear forest, it is isomorphic to $P_{i_1}+\nobreak P_{i_2}+\nobreak\cdots+\nobreak P_{i_k}$ for some positive integers $i_1 \geq i_2 \geq \cdots \geq i_k$.
We may assume that $i_1 \geq 2$, otherwise $H=sP_1$ for some~$s \geq 1$.
Since~$H$ is $2P_2$-free, it follows that $i_1 \leq 4$ and, if $k \geq 2$, then $i_2 \leq 1$, so~$H$ has exactly one non-trivial component and that component is isomorphic to $P_2$, $P_3$ or~$P_4$.
So $H=P_s+tP_1$ for some $s \in \{2,3,4\}$ and $t \geq 0$.
If $s=4$ then $t \leq 1$, since~$H$ is $(3P_1+\nobreak P_2)$-free, in which case~$H$ is an induced subgraph of $P_1+\nobreak P_4$ and we are done.
If $s\in \{2,3\}$ then $t \leq 2$, since~$H$ is $(3P_1+\nobreak P_2)$-free, in which case~$H$ is an induced subgraph of $2P_1+\nobreak P_3$ and we are done.
This completes the proof.\qedllncs
\end{proof}

\subsection{Clique-width}\label{sec:clique-width}
The {\em clique-width} of a graph~$G$, denoted by~$\cw(G)$, is the minimum number of labels needed to construct~$G$ by using the following four operations:
\begin{enumerate}
\item creating a new graph consisting of a single vertex~$v$ with label~$i$;
\item taking the disjoint union of two labelled graphs~$G_1$ and~$G_2$;
\item joining each vertex with label~$i$ to each vertex with label~$j$ ($i\neq j$);
\item renaming label~$i$ to~$j$.
\end{enumerate}

\noindent
A class of graphs~${\cal G}$ has bounded clique-width if there is a constant~$c$ such that the clique-width of every graph in~${\cal G}$ is at most~$c$; otherwise the clique-width of~${\cal G}$ is {\em unbounded}.

Let~$G$ be a graph.
We define the following operations.
For an induced subgraph $G'\ssi G$, the {\em subgraph complementation} operation (acting on~$G$ with respect to~$G'$) replaces every edge present in~$G'$ by a non-edge, and vice versa.
Similarly, for two disjoint vertex subsets~$S$ and~$T$ in~$G$, the {\em bipartite complementation} operation with respect to~$S$ and~$T$ acts on~$G$ by replacing every edge with one end-vertex in~$S$ and the other one in~$T$ by a non-edge and vice versa.
We now state some useful facts about how the above operations (and some other ones) influence the clique-width of a graph.
We will use these facts throughout the paper.
Let $k\geq 0$ be a constant and let~$\gamma$ be some graph operation.
We say that a graph class~${\cal G'}$ is {\em $(k,\gamma)$-obtained} from a graph class~${\cal G}$ if the following two conditions hold:
\begin{enumerate}
\item every graph in~${\cal G'}$ is obtained from a graph in~${\cal G}$ by performing~$\gamma$ at most~$k$ times, and
\item for every $G\in {\cal G}$ there exists at least one graph in~${\cal G'}$ obtained from~$G$ by performing~$\gamma$ at most~$k$ times.
\end{enumerate}

\noindent
We say that~$\gamma$ {\em preserves} boundedness of clique-width if for any finite constant~$k$ and any graph class~${\cal G}$, any graph class~${\cal G}'$ that is $(k,\gamma)$-obtained from~${\cal G}$ has bounded clique-width if and only if~${\cal G}$ has bounded clique-width.

\begin{enumerate}[\bf F{a}ct 1.]
\item \label{fact:del-vert}Vertex deletion preserves boundedness of clique-width~\cite{LR04}.\\[-1em]

\item \label{fact:comp}Subgraph complementation preserves boundedness of clique-width~\cite{KLM09}.\\[-1em]

\item \label{fact:bip}Bipartite complementation preserves boundedness of clique-width~\cite{KLM09}.\\[-1em]

\end{enumerate}
We need the following lemmas on clique-width, the first one of which is easy to show.

\begin{lemma}\label{lem:atmost-2}
The clique-width of a graph of maximum degree at most~$2$ is at most~$4$.
\end{lemma}

\begin{lemma}[\cite{DP15}]\label{lem:P4-free-bdd-cw}
Let~$H$ be a graph.
The class of $H$-free graphs has bounded clique-width if and only~if $H \ssi P_4$.
\end{lemma}

\begin{lemma}[\cite{LR06}]\label{lem:classS}
Let $\{H_1,\ldots,H_p\}$ be a finite set of graphs.
If $H_i\notin {\cal S}$ for all $i \in \{1,\ldots,p\}$ then the class of $(H_1,\ldots,H_p)$-free
graphs has unbounded~clique-width.
\end{lemma}

\begin{lemma}[\cite{CO00}]\label{lem:prime}
Let~$G$ be a graph and let~${\cal P}$ be the set of all induced subgraphs of~$G$ that are prime.
Then $\cw(G)=\max_{H \in {\cal P}}\cw(H)$.
\end{lemma}

\section{The Proof of Theorem~\ref{t-self}}\label{s-self}

We first prove the following lemma, which we will also use in the proof of Theorem~\ref{t-main2}.
\begin{lemma}\label{lem:ramsey-for-self-comp}
If~$G$ is a $(C_4,C_5,K_4)$-free self-complementary graph then~$G$ is an induced subgraph of the bull.
\end{lemma}
\begin{proof}
Suppose, for contradiction, that~$G$ is a $(C_4,C_5,K_4)$-free self-complementary graph on~$n$ vertices that is not an induced subgraph of the bull.
Since~$G$ is~$C_5$-free and is not an induced subgraph of the bull, it is not equal to~$P_1$, $P_4$, $C_5$ or the bull.
As these are the only non-empty self-complementary graph on less than eight vertices (see \figurename~\ref{fig:small-self-comp}), $G$ must have at least eight vertices.
Since~$G$ is $C_4$-free and self-complementary, it is also $2P_2$-free, so it is $(C_4,C_5,2P_2)$-free.
Then, by Lemma~\ref{lem:split-forb-graphs}, $G$ must be a split graph, so its vertex set can be partitioned into a clique~$C$ and an independent set~$I$.
Since~$G$ is $K_4$-free and self-complementary, it is also $4P_1$-free.
Therefore $|C|,|I| \leq 3$, so~$G$ has at most six vertices, a contradiction.
This completes the proof.\qedllncs
\end{proof}

We are now ready to prove Theorem~\ref{t-self}. 
Note that this theorem holds even if~${\cal H}$ is infinite.

\medskip
\newpage
\noindent
\faketheorem{Theorem~\ref{t-self} (restated).}
{\it Let~${\cal H}$ be a
set of non-empty self-complementary graphs.
Then the class of ${\cal H}$-free graphs has bounded clique-width if and only if either $P_1 \in {\cal H}$ or $P_4 \in {\cal H}$.}

\begin{proof}
Suppose there is a graph $H \in {\cal F} \cap \{P_1,P_4\}$.
Then the class of ${\cal F}$-free graphs is a subclass of the class of $P_4$-free graphs, which have bounded clique-width by Lemma~\ref{lem:P4-free-bdd-cw}.
Now suppose that ${\cal F} \cap \{P_1,P_4\}=\emptyset$.
The only non-empty self-complementary graphs on at most five vertices that are not equal to~$P_1$ and~$P_4$ are the bull and the~$C_5$ (see \figurename~\ref{fig:small-self-comp}).
By Lemma~\ref{lem:ramsey-for-self-comp}, it follows that every graph in~${\cal F}$ contains an induced subgraph isomorphic to the bull, $C_4$, $C_5$ or~$K_4$.
Therefore the class of ${\cal F}$-free graphs contains the class of $(\mbox{bull},C_4,C_5,K_4)$-free graphs, which has unbounded clique-width by Lemma~\ref{lem:classS}.\qedllncs
\end{proof}

\section{The Proof of Theorem~\ref{t-main}}\label{s-main}

In this section we prove Theorem~\ref{t-main}.
As stated in Section~\ref{sec:into}, to do this, we need to prove that $(2P_1+\nobreak P_3,\allowbreak\overline{2P_1+P_3})$-free graphs have bounded clique-width.
Open Problem~\ref{oprob:twographs} tells us that no other cases remain to be solved.
We will prove that $(2P_1+\nobreak P_3,\allowbreak\overline{2P_1+P_3})$-free graphs have bounded clique-width in the following way.
We first prove three useful structural lemmas, namely Lemmas~\ref{lem:matching-comatching}--\ref{l-victor}; we will use these lemmas repeatedly throughout the proof.
Next, we prove Lemmas~\ref{lem:2P1+P_3-co-2P_1+P_3-free-c5-non-free} and~\ref{lem:2P1+P_3-co-2P_1+P_3-free-c6-non-free}, which state that if a $(2P_1+\nobreak P_3,\allowbreak\overline{2P_1+P_3})$-free graph~$G$ contains an induced~$C_5$ or~$C_6$, respectively, then~$G$ has bounded clique-width.
We do this by partitioning the vertices outside this cycle into sets, depending on their neighbourhood in the cycle.
We then analyse the edges within these sets and between pairs of such sets.
After a lengthy case analysis, we find that~$G$ has bounded clique-width in both these cases.
By Fact~\ref{fact:comp} it only remains to analyse $(2P_1+\nobreak P_3,\allowbreak\overline{2P_1+P_3})$-free graphs that are also $(C_5,C_6,\overline{C_6})$-free.
Next, in Lemma~\ref{lem:2P1P3-co-2P1P3-C6-co-C6-prime-graphs}, we show that if such graphs are prime, then they are either $K_7$-free or $\overline{K_7}$-free.
In Lemma~\ref{lem:2P1P3-co-2P1P3-Kk-free-bdd-cw} we use the fact that $(2P_1+\nobreak P_3,\allowbreak\overline{2P_1+P_3})$-free graphs are $\chi$-bounded to deal with the case where a graph in the class is $K_7$-free.
Finally, we combine all these results together in the proof of Theorem~\ref{t-main}.

\medskip
\noindent
We start by proving the aforementioned structural lemmas.
Recall that if~$X$ and~$Y$ are disjoint sets of vertices in a graph, we say that the edges between these two sets form a matching if each vertex in~$X$ has at most one neighbour in~$Y$ and vice versa
(if each vertex has exactly one such neighbour, we say that the matching is {\em perfect}).
Similarly, the edges between these sets form a co-matching if each vertex in~$X$ has at most one non-neighbour in~$Y$ and vice versa.
Also note that when describing a set as being a clique or an independent set, we allow the case where this set is empty.

\begin{lemma}\label{lem:matching-comatching}
Let~$G$ be a $(2P_1+\nobreak P_3,\allowbreak\overline{2P_1+P_3})$-free graph whose vertex set can be partitioned into two sets~$X$ and~$Y$, each of which is a clique or an independent set.
Then by deleting at most one vertex from each of~$X$ and~$Y$, it is possible to obtain subsets such that the edges between them form a matching or a co-matching.
\end{lemma}

\begin{proof}
\setcounter{ctrclaim}{0}
Given two disjoint sets of vertices, we say that with respect to these sets, a vertex is \emph{full} if it is adjacent to all but at most one vertex in the other set, and it is \emph{empty} if it is adjacent to at most one vertex in the other set.
If every vertex in the two sets is full, then the edges between the two sets form a co-matching, and if every vertex in the two sets is empty, then the edges between them form a matching.

\clm{\label{clm:2p1p3-full-or-empty}Each vertex in~$X$ and~$Y$ is either full or empty.}
If a vertex in, say, $X$ is neither full nor empty, then it has two neighbours and two non-neighbours in~$Y$, and these five vertices induce a $2P_1+\nobreak P_3$ if~$Y$ is an independent set, or a $\overline{2P_1+P_3}$ if~$Y$ is a clique.
This completes the proof of the claim.

\medskip
\noindent
To prove the lemma, we must show that, after discarding at most one vertex from each of~$X$ and~$Y$, we have a pair of sets such that every vertex is full or every vertex is empty with respect to this pair.
We note that if a vertex is full (or empty) with respect to~$X$ and~$Y$, then is also full (or empty) with respect to any pair of subsets of~$X$ and~$Y$ respectively, so if we establish or assume fullness (or emptiness) before discarding a vertex, then it still holds afterwards.

We consider a number of cases.

\thmcase{\label{case:2p1p3-trivial}Neither~$X$ nor~$Y$ contains two full vertices, or neither~$X$ nor~$Y$ contains two empty vertices.}
By deleting at most one vertex from each of~$X$ and~$Y$, we can obtain a pair of sets where either every vertex is full or every vertex is empty.
This completes the proof of Case~\ref{case:2p1p3-trivial}.

\thmcase{\label{case:2p1p3-small}$|X| \leq 2$ or $|Y| \leq 2$.}
By symmetry we may assume that $|X| \leq 2$.
If~$X$ is empty or contains exactly one vertex, the lemma is immediate, so we may assume that~$X$ contains exactly two vertices, say~$x$ and~$x'$.
Consider the pair of sets~$\{x\}$ and~$Y$.
Every vertex in~$Y$ is both full and empty with respect to~$\{x\}$ and~$Y$, and, by Claim~\ref{clm:2p1p3-full-or-empty}, $x$ is either full or empty with respect to~$\{x\}$ and~$Y$.
This completes the proof of Case~\ref{case:2p1p3-small}.

\thmcase{\label{case:2p1p3-big}There are vertices $x_1, x_2 \in X$ and $y_1 \in Y$ such that~$x_1$ and~$x_2$ are complete to~$Y \setminus \{y_1\}$.}
In this case, every vertex in $Y \setminus \{y_1\}$ is adjacent to both~$x_1$ and~$x_2$, so it cannot be empty with respect to~$X$ and~$Y$.
By Claim~\ref{clm:2p1p3-full-or-empty}, it follows that every vertex in $Y \setminus \{y_1\}$ is full.
We may assume that $|Y| \geq 3$ (otherwise we apply Case~\ref{case:2p1p3-small}).
Let~$y_2$ and~$y_3$ be vertices in $Y \setminus \{y_1\}$.
As~$y_2$ and~$y_3$ are both full with respect to~$X$ and $Y \setminus \{y_1\}$, all but at most two vertices of~$X$ are adjacent to both~$y_2$ and~$y_3$.
Note that if a vertex~$x$ is adjacent to both~$y_2$ and~$y_3$ then it must be full with respect to~$X$ and $Y \setminus \{y_1\}$.
If at most one vertex of~$X$ is empty with respect to~$X$ and $Y \setminus \{y_1\}$ then by discarding this vertex (if it exists) from~$X$ and discarding~$y_1$ from~$Y$, we are done.

So we may assume that~$X$ contains exactly two vertices~$x_3$ and~$x_4$ that are not full with respect to~$X$ and $Y \setminus \{y_1\}$ and thus are empty.
Suppose that $|Y| \geq 4$.
Then there are three full vertices in $Y \setminus \{y_1\}$ that must each be adjacent to at least one of~$x_3$ and~$x_4$.
Thus at least one of~$x_3$ and~$x_4$ has at least two neighbours in $Y \setminus \{y_1\}$ contradicting the fact that they are both empty with respect to~$X$ and $Y \setminus \{y_1\}$.

Thus we may now assume that $|Y|=3$, so $Y \setminus \{y_1\} = \{y_2, y_3\}$.
By assumption, $x_3$ and~$x_4$ are not full with respect to~$X$ and~$\{y_2,y_3\}$, so they must have two non-neighbours in $\{y_2, y_3\}$ i.e. they must be anti-complete to $\{y_2,y_3\}$.
Thus~$y_2$ has two non-neighbours in~$X$, so it is empty with respect to~$X$ and~$Y$.
Since $|X| \ge 4$, this means that~$y_2$ is not full with respect to~$X$ and~$Y$, a contradiction.
This completes the proof of Case~\ref{case:2p1p3-big}.

\medskip
\noindent
We note that if, in Case~\ref{case:2p1p3-big}, we swap~$X$ and~$Y$, or write anti-complete instead of complete, we obtain further cases with essentially the same proof.
We now assume that neither these cases, nor Cases~\ref{case:2p1p3-trivial} and~\ref{case:2p1p3-small}, hold.

\clm{\label{clm:2p1p3-distinct}If there are two full vertices $x_1, x_2 \in X$, then they have distinct non-neighbours in~$Y$.
If there are two empty vertices $x_1, x_2 \in X$, then they have distinct neighbours in~$Y$.}
We prove the first statement (the second follows by symmetry).
If~$x_1$ and~$x_2$ are both complete to~$Y$, then Case~\ref{case:2p1p3-big} would apply with any vertex in~$Y$ chosen as~$y_1$.
Suppose instead that~$y_1$ is the unique non-neighbour of~$x_1$.
Then~$x_2$ must have a non-neighbour in~$Y$ that is different from~$y_1$, otherwise Case~\ref{case:2p1p3-big} would apply.
This completes the proof of the claim.

\clm{\label{clm:2p1p3-two-and-two}There are at least two empty vertices in~$X$ and at least two full vertices in~$Y$ or vice versa.}
As Case~\ref{case:2p1p3-trivial} does not apply, we know that one of~$X$ and~$Y$ contains two empty vertices, and one of~$X$ and~$Y$ contains two full vertices.
We are done unless these two properties belong to the same set.
So let us suppose that, without loss of generality, it is~$X$ that contains two empty vertices and two full vertices, which we may assume are distinct (if a vertex in~$X$ is both full and empty, then $|Y| \leq 2$ and Case~\ref{case:2p1p3-small} applies).
By Claim~\ref{clm:2p1p3-distinct}, the two empty vertices of~$X$ have distinct neighbours~$y_1$ and~$y_2$ in~$Y$.
If~$y_1$ and~$y_2$ are both full, we are done.
If, say, $y_1$ is empty, then, as it is adjacent to one of the empty vertices in~$X$, it cannot be adjacent to either of the full vertices in~$X$, contradicting Claim~\ref{clm:2p1p3-distinct}.
This completes the proof of Claim~\ref{clm:2p1p3-two-and-two}.

\medskip
\noindent
We immediately use Claim~\ref{clm:2p1p3-two-and-two}.
Let us assume, without loss of generality, that $x_1, x_2 \in X$ are empty and $y_1, y_2 \in Y$ are full with respect to~$X$ and~$Y$.
Moreover, by Claim~\ref{clm:2p1p3-distinct}, we may assume that~$y_1$ is the unique neighbour of~$x_1$ and~$y_2$ is the unique neighbour of~$x_2$ (so~$x_1$ is the unique non-neighbour of~$y_2$ and~$x_2$ is the unique non-neighbour of~$y_1$).
Thus every vertex of $X \setminus \{x_1, x_2\}$ is complete to $\{y_1, y_2\}$, and therefore, by Claim~\ref{clm:2p1p3-full-or-empty}, full with respect to~$X$ and~$Y$.
Similarly, every vertex of $Y \setminus \{y_1, y_2\}$ is anti-complete to $\{x_1, x_2\}$, and therefore empty with respect to~$X$ and~$Y$.

If $|X|=3$, then every vertex in $\{x_1, x_2\}$ and~$Y$ is empty with respect to $\{x_1, x_2\}$ and~$Y$.
Otherwise we can find distinct vertices $x_3, x_4$ in $X \setminus \{x_1, x_2\}$ which we know are both full and both complete to $\{y_1, y_2\}$.
Hence, by Claim~\ref{clm:2p1p3-distinct}, there are distinct vertices $y_3, y_4$ in $Y \setminus \{y_1, y_2\}$ such that~$y_3$ is the unique non-neighbour of~$x_3$ and~$y_4$ is the unique non-neighbour of~$x_4$.
If~$X$ and~$Y$ are independent sets then $G[x_1,y_4,y_2,x_4,y_3]$ is a $2P_1+\nobreak P_3$.
If~$X$ is an independent set and~$Y$ is a clique then $G[y_1,y_2,y_3,x_3,x_4]$ is a $\overline{2P_1+P_3}$.
If~$X$ is a clique and~$Y$ is an independent set then $G[x_3,x_4,x_2,y_1,y_2]$ is a $\overline{2P_1+P_3}$.
Finally if~$X$ and~$Y$ are cliques then $G[x_3,y_1,y_2,x_1,y_4]$ is a $\overline{2P_1+\nobreak P_3}$.
This contradiction completes the proof.\qedllncs
\end{proof}

\begin{lemma}\label{lem:comp-anti}
Let~$G$ be a $(2P_1+\nobreak P_3,\allowbreak\overline{2P_1+P_3})$-free graph whose vertex set can be partitioned into a clique~$X$ and an independent set~$Y$.
Then by deleting at most three vertices from each of~$X$ and~$Y$, it is possible to obtain subsets that are either complete or anti-complete to each other.
\end{lemma}

\begin{proof}
Let~$G$ be such a graph.
By Lemma~\ref{lem:matching-comatching}, by deleting at most one vertex from each of~$X$ and~$Y$, we may reduce to the case where the edges between~$X$ and~$Y$ form a matching or a co-matching.
(Note that after this we may delete at most two further vertices from each of~$X$ and~$Y$.)
Complementing the graph if necessary (in which case we also swap~$X$ and~$Y$), we may assume that the edges between~$X$ and~$Y$ form a matching.
Let $x_1y_1,\ldots,x_iy_i$ be the edges between~$X$ and~$Y$, with $x_1,\ldots,x_i \in X$ and $y_1,\ldots,y_i \in Y$.
If $i \geq 4$ then $G[y_1,y_2,y_3,x_3,x_4]$ is a $2P_1+\nobreak P_3$, a contradiction.
We may therefore assume that $i \leq 3$.
Deleting the vertices $x_1,x_2,y_3$ (if they are present) completes the proof.\qedllncs
\end{proof}

\begin{lemma}\label{l-victor}
The class of those $(2P_1+\nobreak P_3,\overline{2P_1+P_3})$-free graphs whose vertex set can be partitioned into at most three cliques and at most three independent sets has bounded clique-width.
\end{lemma}

\begin{proof}
Let~$G$ be a $(2P_1+\nobreak P_3,\overline{2P_1+P_3})$-free graph whose vertex set can be partitioned into at three (possibly empty) cliques~$K^1$, $K^2$, $K^3$ and three (possibly empty) independent sets~$I^1$, $I^2$, $I^3$.
By Lemma~\ref{lem:comp-anti} and Fact~\ref{fact:del-vert}, we may delete at most $2\times 3\times 3\times 3=54$ vertices, after which every~$K^i$ is either complete or anti-complete to every~$I^j$.
If two such sets are complete to each other, then by Fact~\ref{fact:bip}, we may apply a bipartite complementation between them.
Now, by Lemma~\ref{lem:matching-comatching} and Fact~\ref{fact:del-vert}, we may delete at most $2 \times 2 \times \binom{3}{2} = 12$ vertices, after which the edges between any two cliques $K^i$, $K^j$ and the edges between any two independent sets $I^i$, $I^j$ either form a matching or a co-matching.
If the edges form a co-matching, then by Fact~\ref{fact:bip} we may apply a bipartite complementation between these sets.
Finally, by Fact~\ref{fact:comp}, we may complement every clique~$K^i$.
The resulting graph has maximum degree at most~$2$, and therefore has clique-width at most~$4$ by Lemma~\ref{lem:atmost-2}.
It follows that~$G$ has bounded clique-width.\qedllncs
\end{proof}

\begin{lemma}\label{lem:2P1+P_3-co-2P_1+P_3-free-c5-non-free}
The class of $(2P_1+\nobreak P_3,\allowbreak \overline{2P_1+P_3})$-free graphs containing an induced~$C_5$ has bounded clique-width.
\end{lemma}

\begin{proof}
\setcounter{ctrclaim}{0}
Suppose~$G$ is a $(2P_1+\nobreak P_3,\allowbreak \overline{2P_1+P_3})$-free graph containing an induced cycle~$C$ on five vertices, say $v_1,\ldots,v_5$ in that order.
For $S \subseteq \{1,\ldots,5\}$, let~$V_S$ be the set of vertices $x \in V(G) \setminus V(C)$ such that $N(x)\cap V(C)=\{v_i \;|\; i \in S\}$.
We say that a set~$V_S$ is {\em large} if it contains at least five vertices, otherwise it is {\em small}.

To ease notation, in the following claims, subscripts on vertex sets should be interpreted modulo~$5$ 
and whenever possible we will write~$V_i$ instead of~$V_{\{i\}}$ and~$V_{i,j}$ instead of~$V_{\{i,j\}}$ and so on.

\clm{\label{clm2:V_S-large-or-empty}We may assume that for $S \subseteq \{1,2,3,4,5\}$, the set~$V_S$ is either large or empty.}
If a set~$V_S$ is small, but not empty, then by Fact~\ref{fact:del-vert}, we may delete all vertices of this set.
If later in our proof we delete vertices in some set~$V_S$ and in doing so make a large set~$V_S$ become small, we may immediately delete the remaining vertices in~$V_S$.
The above arguments involve deleting a total of at most $2^5 \times 4$ vertices.
By Fact~\ref{fact:del-vert}, the claim follows.

\clm{\label{clm2:V0V1V12-clique}For $i \in \{1,2,3,4,5\}$, $V_\emptyset \cup V_i \cup V_{i+1} \cup V_{i,i+1}$ is a clique.}
Indeed, if $x,y \in V_\emptyset \cup V_1 \cup V_2 \cup V_{1,2}$ are non-adjacent then $G[x,y,v_3,v_4,v_5]$ is a $2P_1+\nobreak P_3$, a contradiction.
The claim follows by symmetry.

\clm{\label{clm2:V13-P3-free}For $i \in \{1,2,3,4,5\}$, $G[V_{i,i+2}]$ is $P_3$-free.}
Indeed, if $G[V_{1,3}]$ contains an induced~$P_3$, say on vertices $x,y,z$, then $G[v_2,v_4,x,y,z]$ is a $2P_1+\nobreak P_3$, a contradiction.
The claim follows by symmetry.

\medskip
\noindent
Note that since~$G$ is a $(2P_1+\nobreak P_3,\allowbreak \overline{2P_1+P_3})$-free graph containing a~$C_5$, it follows that~$\overline{G}$ is also a $(2P_1+\nobreak P_3,\allowbreak \overline{2P_1+P_3})$-free graph containing a~$C_5$, namely on the vertices $v_1,v_3,v_5,v_2,v_4$, in that order.
Let $w_1=v_1$, $w_2=v_3$, $w_3=v_5$, $w_4=v_2$ and $w_5=v_4$.
For $S \subseteq \{1,2,3,4,5\}$, we say that a vertex~$x$ not in~$C$ belongs to~$W_S$ if $N(x) \cap V(C) = \{w_i \; | \; i \in S\}$ {\em in the graph~$\overline{G}$}.
We define the function $\sigma: \{1,2,3,4,5\} \rightarrow \{1,2,3,4,5\}$ as follows: $\sigma(1)=1$, $\sigma(3)=2$, $\sigma(5)=3$, $\sigma(2)=4$ and $\sigma(4)=5$.
Now for $S,T \subseteq \{1,2,3,4,5\}$, $x \in V_S$ if and only if $x \in W_T$ where $T=\{1,2,3,4,5\} \setminus \{\sigma(i) \; | \; i \in S\}$.
Therefore, we may assume that any claims proved for a set~$V_S$ in~$G$ also hold for the set~$W_S$ in~$\overline{G}$.

For convenience we provide Table~\ref{tbl:WT-VS-corr}, which lists the correspondence between the sets~$W_T$ and the sets~$V_S$.
\begin{table}
\begin{center}
\begin{tabular}{c|c|c|c|c}
$V_1 = W_{2,3,4,5}$ &
$V_2 = W_{1,2,3,5}$ &
$V_3 = W_{1,3,4,5}$ &
$V_4 = W_{1,2,3,4}$ &
$V_5 = W_{1,2,4,5}$\\

$V_{1,2} = W_{2,3,5}$ &
$V_{2,3} = W_{1,3,5}$ &
$V_{3,4} = W_{1,3,4}$ &
$V_{4,5} = W_{1,2,4}$ &
$V_{1,5} = W_{2,4,5}$\\

$V_{1,3} = W_{3,4,5}$ &
$V_{2,4} = W_{1,2,3}$ &
$V_{3,5} = W_{1,4,5}$ &
$V_{1,4} = W_{2,3,4}$ &
$V_{2,5} = W_{1,2,5}$\\

$V_{1,2,3} = W_{3,5}$ &
$V_{2,3,4} = W_{1,3}$ &
$V_{3,4,5} = W_{1,4}$ &
$V_{1,4,5} = W_{2,4}$ &
$V_{1,2,5} = W_{2,5}$\\

$V_{1,2,4} = W_{2,3}$ &
$V_{2,3,5} = W_{1,5}$ &
$V_{1,3,4} = W_{3,4}$ &
$V_{2,4,5} = W_{1,2}$ &
$V_{1,3,5} = W_{4,5}$\\

$V_{1,2,3,4} = W_{3}$ &
$V_{1,2,3,5} = W_{5}$ &
$V_{1,2,4,5} = W_{2}$ &
$V_{1,3,4,5} = W_{4}$ &
$V_{2,3,4,5} = W_{1}$\\

$V_\emptyset = W_{1,2,3,4,5}$ & $V_{1,2,3,4,5} = W_\emptyset$
\end{tabular}
\end{center}
\caption{\label{tbl:WT-VS-corr}The correspondence between the sets~$W_T$ and the sets~$V_S$.}
\end{table}

We therefore get the following two corollaries of Claims~\ref{clm2:V0V1V12-clique} and~\ref{clm2:V13-P3-free}, respectively.
We include the argument for the first corollary to demonstrate how this ``casting to the complement'' argument works.

\clm{\label{clm2:V124V1234V12345-indep}For $i \in \{1,2,3,4,5\}$, $V_{i,i+1,i+3} \cup V_{i,i+1,i+2,i+3} \cup V_{i,i+1,i+3,i+4} \cup V_{1,2,3,4,5}$ is an independent set.}
Indeed, for $i=1$, $V_{i,i+1,i+3} \cup V_{i,i+1,i+2,i+3} \cup V_{i,i+1,i+3,i+4} \cup V_{1,2,3,4,5}$ is $V_{1,2,4} \cup V_{1,2,3,4} \cup V_{1,2,4,5} \cup V_{1,2,3,4,5}$, which is equal to $W_{2,3} \cup W_3 \cup W_2 \cup W_\emptyset$.
By Claim~\ref{clm2:V0V1V12-clique}, $W_{2,3} \cup W_3 \cup W_2 \cup W_\emptyset$ is an independent set.
The claim follows by complementing and symmetry.

\medskip
\shortclm{\label{clm2:V123-P1+P2-free}For $i \in \{1,2,3,4,5\}$, $G[V_{i,i+1,i+2}]$ is $(P_1+\nobreak P_2)$-free.}

\clm{\label{clm2:clique-indep-trivial}We may assume that for distinct $S,T \subseteq \{1,2,3,4,5\}$ if~$V_S$ is an independent set and~$V_T$ is a clique then~$V_S$ is either complete or anti-complete to~$V_T$.}
Let $S,T \subseteq \{1,\ldots,5\}$ be distinct.
If~$V_S$ is an independent set and~$V_T$ is a clique, then by Lemma~\ref{lem:comp-anti}, we may delete at most three vertices from each of these sets, such that in the resulting graph, $V_S$ will be complete or anti-complete to~$V_T$.
Doing this for every pair of independent set~$V_S$ and a clique~$V_T$ we delete at most $\binom{2^5}{2}\times 2 \times 3$ vertices from~$G$.
The claim follows by Fact~\ref{fact:del-vert}.

\clm{\label{clm2:indep-comatching}We may assume that for distinct $S,T \subseteq \{1,2,3,4,5\}$, if~$V_S$ and~$V_T$ are both independent sets then the edges between~$V_S$ and~$V_T$ form a co-matching.}
Let $S,T \subseteq \{1,\ldots,5\}$ be distinct.
We may assume that~$V_S$ and~$V_T$ are not empty, in which case they must both be large, i.e. $|V_S|,|V_T| \geq 5$.
If~$V_S$ and~$V_T$ are both independent sets, then by Lemma~\ref{lem:matching-comatching}, we may delete at most one vertex from each of these sets, such that in the resulting graph, the edges between~$V_S$ and~$V_T$ form a matching or a co-matching.
Note that after this modification we only have the weaker bound $|V_S|,|V_T| \geq 4$ in the resulting graph.
Suppose, for contradiction, that the edges between~$V_S$ and~$V_T$ form a matching.
Without loss of generality assume there is an $i \in T \setminus S$.
Since $|V_S| \geq 4$, there must be vertices $x,x'\in V_S$.
Since each vertex in~$V_S$ has at most one neighbour in~$V_T$ and $|V_T| \geq 4$, there must be vertices $y,y' \in V_T$ that are non-adjacent to both~$x$ and~$x'$.
Then $G[x,x',y,v_i,y']$ is a $2P_1+\nobreak P_3$, a contradiction.
Therefore the edges between~$V_S$ and~$V_T$ must indeed form a co-matching.
The claim follows by Fact~\ref{fact:del-vert}.

\medskip
\noindent
In many cases, we can prove a stronger claim:

\clm{\label{clm2:indeps-no-i-comp}For distinct $S,T \subseteq \{1,2,3,4,5\}$, if~$V_S$ and~$V_T$ are both independent and there is an $i \in \{1,2,3,4,5\}$ with $i \notin V_S$ and $i \notin V_T$ then~$V_S$ is complete to~$V_T$.}
Let $S,T \subseteq \{1,\ldots,5\}$ be distinct and suppose there is an $i \in \{1,2,3,4,5\}$ with $i \notin V_S$ and $i \notin V_T$.
We may assume~$V_S$ and~$V_T$ are not empty, so they must be large.
By Claim~\ref{clm2:indep-comatching}, we may assume that the edges between~$V_S$ and~$V_T$ form a co-matching.
Suppose, for contradiction that $x \in V_S$ is non-adjacent to $y \in V_T$.
Since~$V_S$ is large, there must be vertices $x',x'' \in V_S \setminus \{x\}$ and these vertices must be adjacent to~$y$.
Now $G[v_i,x,x',y,x'']$ is a $2P_1+\nobreak P_3$, a contradiction.
Therefore~$V_S$ must be complete to~$V_T$.
The claim follows.

\medskip
\noindent
Casting to the complement we get the following as a corollary to the above two claims.

\medskip
\shortclm{\label{clm2:clique-matching}We may assume that for distinct $S,T \subseteq \{1,2,3,4,5\}$, if~$V_S$ and~$V_T$ are both cliques then the edges between~$V_S$ and~$V_T$ form a matching.}\\
\shortclm{\label{clm2:cliques-i-anti}For distinct $S,T \subseteq \{1,2,3,4,5\}$, if~$V_S$ and~$V_T$ are both cliques and there is an $i \in \{1,2,3,4,5\}$ with $i \in V_S$ and $i \in V_T$ then~$V_S$ is anti-complete to~$V_T$.}

\clm{\label{clm2:V_S-V_T-contain12-empty}For $i \in \{1,\ldots,5\}$, if $\{i,i+1\} \subseteq S \cap T$ and $T \neq S$ then either~$V_S$ or~$V_T$ is empty.}
Suppose~$S$ and~$T$ are as described above, but~$V_S$ and~$V_T$ are both non-empty.
By Claim~\ref{clm2:V_S-large-or-empty}, $V_S$ and~$V_T$ must be large.
Without loss of generality, we may assume that $1,2 \in S \cap T$ and $3 \in T \setminus S$ or $S=\{1,2\}$, $T=\{1,2,4\}$.

First consider the case where $1,2 \in S \cap T$ and $3 \in T \setminus S$.
If $x \in V_S$ and $y \in V_T$ are adjacent then $G[y,v_2,v_1,v_3,x]$ is a $\overline{2P_1+P_3}$, a contradiction.
Therefore~$V_S$ is anti-complete to~$V_T$.
Suppose $x,x' \in V_S$ and $y,y' \in V_T$.
If~$x$ is adjacent to~$x'$ then $G[v_1,v_2,x,y,x']$ is a $\overline{2P_1+P_3}$, a contradiction.
If~$y$ is adjacent to~$y'$ then $G[v_1,v_2,y,x,y']$ is a $\overline{2P_1+P_3}$, a contradiction.
Therefore~$x$ must be non-adjacent to~$x'$ and~$y$ must be non-adjacent to~$y'$.
This means that $G[x,x',y,v_3,y']$ is a $2P_1+\nobreak P_3$, a contradiction.

Now consider the case where $S=\{1,2\}$, $T=\{1,2,4\}$.
Then $V_{1,2}$ is a clique and~$V_{1,2,4}$ is an independent set, by Claims~\ref{clm2:V0V1V12-clique} and~\ref{clm2:V124V1234V12345-indep}, respectively.
By Claim~\ref{clm2:clique-indep-trivial}, $V_{1,2}$ must be complete or anti-complete to~$V_{1,2,4}$.
Suppose $x,x' \in V_{1,2}$ and $y,y' \in V_{1,2,4}$.
If~$V_{1,2}$ is anti-complete to~$V_{1,2,4}$ then $G[v_1,v_2,x,y,x']$ is a $\overline{2P_1+P_3}$.
If~$V_{1,2}$ is complete to~$V_{1,2,4}$ then $G[v_3,v_5,y,x,y']$ is a $2P_1+\nobreak P_3$.
This is a contradiction.
The claim follows by symmetry.

\medskip
\noindent
Casting Claim~\ref{clm2:V_S-V_T-contain12-empty} to the complement, we obtain the following corollary.

\medskip
\shortclm{\label{clm2:V_S-V_T-not-contain13-empty}For $i \in \{1,\ldots,5\}$, if $\{i,i+2\} \cap (S \cup T)=\emptyset$ and $T \neq S$ then~$V_S$ or~$V_T$ is empty.}

\medskip
\noindent
We now give a brief outline of the remainder of the proof.
First, in Claims~\ref{clm2:V13-V2-V13-indep-clique}-\ref{clm2:V13-V1345-V13-indep-comp} we will analyse the edges between different sets~$V_S$ and sets of the form~$V_{i,i+2}$.
Next, in Claim~\ref{clm2:V13-indep-or-clique} we will consider the case where a set~$V_{i,i+2}$ is neither a clique nor an independent set.
We will then assume that this case does not hold, in which case every set of the form~$V_{i,i+2}$ is either a clique or an independent set.
Casting to the complement, we will get the same conclusion for all sets of the form~$V_{i,i+1,i+2}$.
Combined with Claims~\ref{clm2:V0V1V12-clique} and~\ref{clm2:V124V1234V12345-indep}, this means that every set~$V_S$ is either a clique or an independent set.
By Fact~\ref{fact:del-vert}, we may delete the vertices $v_1,\ldots,v_5$ of the original cycle.
By Claim~\ref{clm2:clique-indep-trivial}, if~$V_S$ is a clique and~$V_T$ is an independent set, then applying at most one bipartite complementation (which we may do by Fact~\ref{fact:bip}) we can remove all edges between~$V_S$ and~$V_T$.
It is therefore sufficient to consider the case where all sets~$V_S$ are cliques or all sets~$V_T$ are independent.
If there are at most three large cliques and at most three large independent sets, then by Lemma~\ref{l-victor} we can bound the clique-width of the graph induced on these sets.
In the proof of Claim~\ref{clm2:V13-indep} we consider the situation where a set of the form~$V_{i,i+2}$ is a large clique.
Having dealt with this case, we may assume that every set of the form~$V_{i,i+2}$ is an independent set (so, casting to the complement, every set of the form~$V_{i,i+1,i+2}$ is a clique) and we deal with this case in Claim~\ref{clm2:V13-empty}.
Finally, we deal with the case where all sets of the form~$V_{i,i+2}$ and~$V_{i,i+1,i+2}$ are empty.

\begin{sloppypar}
\clm{\label{clm2:V13-V2-V13-indep-clique}For $i \in \{1,2,3,4,5\}$, if~$V_{i,i+2}$ and~$V_{i+1}$ are large then~$V_{i,i+2}$ is either an independent set or a clique.}
Suppose, that both~$V_{1,3}$ and~$V_2$ are non-empty.
Then, by Claim~\ref{clm2:V_S-large-or-empty}, they must both be large.
Suppose that~$V_{1,3}$ is not a clique.
Then there are $y,y' \in V_{1,3}$ that are non-adjacent.
Suppose $x \in V_2$ is non-adjacent to~$y'$.
Then $G[v_4,y',v_2,x,y]$ or $G[x,v_4,y,v_1,y']$ is a $2P_1+\nobreak P_3$ if~$x$ is adjacent or non-adjacent to~$y$, respectively.
Therefore~$x$ must be complete to $\{y,y'\}$.
By Claim~\ref{clm2:V13-P3-free}, $G[V_{1,3}]$ is $P_3$-free, so it is a disjoint union of cliques.
Since~$y$ and~$y'$ were chosen arbitrarily and~$V_{1,3}$ is not a clique, it follows that~$x$ must be complete to~$V_{1,3}$.
Therefore~$V_2$ is complete to~$V_{1,3}$.
By Claim~\ref{clm2:V0V1V12-clique}, $V_2$ is a clique.
If $x,x' \in V_2$, $z,z' \in V_{1,3}$ with~$z$ adjacent to~$z'$ then $G[x,x',z,v_2,z']$ is a $\overline{2P_1+P_3}$, a contradiction.
Therefore if~$V_{1,3}$ is not a clique then it must be an independent set.
The claim follows by symmetry.
\end{sloppypar}

\clmnonewline{\label{clm2:V13-V24-both-cliques-or-one-indep-and-comp}For $i \in \{1,2,3,4,5\}$, if~$V_{i,i+2}$ and~$V_{i+1,i+3}$ are both large then either:
\begin{enumerate}[(i)]
\item both~$V_{i,i+2}$ and~$V_{i+1,i+3}$ are cliques or
\item at least one of them is an independent set and the two sets are complete to each other.
\end{enumerate}}
By Claim~\ref{clm2:V13-P3-free}, $G[V_{1,3}]$ and $G[V_{2,4}]$ are $P_3$-free, so every component in these graphs is a clique.
Suppose $G[V_{1,3}]$ is not a clique, so there are non-adjacent vertices $x,x' \in V_{1,3}$.
Suppose $y \in V_{2,4}$ is non-adjacent to~$x'$.
Then $G[x',v_5,v_2,y,x]$ or $G[y,v_5,x,v_3,x']$ is a $2P_1+\nobreak P_3$ if~$y$ is adjacent or non-adjacent to~$x$, respectively.
This contradiction implies that~$y$ is complete to $\{x,x'\}$.
Since we assumed that~$V_{1,3}$ was not a clique and~$x$ and~$x'$ were chosen to be arbitrary non-adjacent vertices in~$V_{1,3}$, it follows that~$y$ must be complete to~$V_{1,3}$.
Therefore if~$V_{1,3}$ is not a clique then~$V_{2,4}$ is complete to~$V_{1,3}$.
Similarly, if~$V_{2,4}$ is not a clique then~$V_{2,4}$ is complete to~$V_{1,3}$.

Now suppose that neither~$V_{1,3}$ nor~$V_{2,4}$ is an independent set.
If they are both cliques, then we are done, so assume for contradiction that at least one of them is not a clique.
Then~$V_{1,3}$ is complete to~$V_{2,4}$.
We can find $x,x' \in V_{1,3}$ that are adjacent and $y,y' \in V_{2,4}$ that are adjacent.
However, this means that $G[x,x',y,v_1,y']$ is a $\overline{2P_1+P_3}$, a contradiction.
The claim follows by symmetry.

\clm{\label{clm2:V13-V12V23-V13-clique-anti-or-indep-comp}For $i \in \{1,2,3,4,5\}$ and $S=\{i,i+1\}$ or $S=\{i+1,i+2\}$, if~$V_{i,i+2}$ and~$V_S$ are large then either~$V_{i,i+2}$ is a clique that is anti-complete to~$V_S$ or an independent set that is complete to~$V_S$.}
By Claim~\ref{clm2:V0V1V12-clique}, $V_{1,2}$ is a clique.
By Claim~\ref{clm2:V13-P3-free}, $G[V_{1,3}]$ is $P_3$-free, so it is a disjoint union of cliques.
If $y \in V_{1,3}$ is adjacent to $x \in V_{1,2}$, but non-adjacent to $x' \in V_{1,2}$ then $G[v_1,x,x',y,v_2]$ is a $\overline{2P_1+P_3}$, a contradiction.
Therefore every vertex of~$V_{1,3}$ is either complete or anti-complete to~$V_{1,2}$.

Suppose $y,y' \in V_{1,3}$ are adjacent and suppose $x,x' \in V_{1,2}$.
Suppose~$y$ is complete to~$V_{1,2}$.
Then $G[x,x',y,v_2,y']$ or $G[v_1,y,x,y',x']$ is a $\overline{2P_1+P_3}$, if~$y'$ is complete or anti-complete to~$V_{1,2}$, respectively.
This contradiction implies that if $y,y' \in V_{1,3}$ are adjacent then they are both anti-complete to~$V_{1,2}$.
It follows that every non-trivial component of~$G[V_{1,3}]$ is anti-complete to~$V_{1,2}$.
In particular, if~$V_{1,3}$ is a clique, then it is anti-complete to~$V_{1,2}$.

Now suppose that~$V_{1,3}$ is not a clique, so there are non-adjacent vertices $y,y' \in V_{1,3}$.
Choose a vertex $x \in V_{1,2}$.
Suppose~$y'$ is anti-complete to~$V_{1,2}$.
Then $G[v_4,y',y,x,v_2]$ or $G[x,v_5,y,v_3,y']$ is a $2P_1+\nobreak P_3$ if~$y$ is complete or anti-complete to~$V_{1,2}$, respectively.
Therefore both~$y$ and~$y'$ must be complete to~$V_{1,2}$.
Since every non-trivial component of~$G[V_{1,3}]$ is anti-complete to~$V_{1,2}$ and~$G[V_{1,3}]$ is a disjoint union of cliques, it follows that~$y$ and~$y'$ must belong to trivial components of~$G[V_{1,3}]$.
Since~$y$ and~$y'$ were arbitrary non-adjacent vertices in~$V_{1,3}$, it follows that every component of~$G[V_{1,3}]$ must be trivial.
Therefore if~$V_{1,3}$ is not a clique then it is an independent set and it is complete to~$V_{1,2}$.
The claim follows by symmetry.

\clm{\label{clm2:V13-V45-V13-clique-or-anti}We may assume that for $i \in \{1,2,3,4,5\}$ if~$V_{i,i+2}$ and~$V_{i+3,i+4}$ are large then either~$V_{i,i+2}$ is a clique or~$V_{i,i+2}$ is anti-complete to~$V_{i+3,i+4}$.}
If $y \in V_{1,3}$ has two neighbours $x,x' \in V_{4,5}$, then~$x$ is adjacent to~$x'$ by Claim~\ref{clm2:V0V1V12-clique}, so $G[x,x',v_4,y,v_5]$ is a $\overline{2P_1+P_3}$, a contradiction.
Therefore every vertex of~$V_{1,3}$ has at most one neighbour in~$V_{4,5}$.

By Claim~\ref{clm2:V13-P3-free}, $G[V_{1,3}]$ is $P_3$-free, so it is a disjoint union of cliques.
Suppose~$V_{1,3}$ is not a clique, so there are non-adjacent vertices $y,y' \in V_{1,3}$.
If $x \in V_{4,5}$ is adjacent to~$y$, but non-adjacent to~$y'$ then $G[v_2,y',y,x,v_4]$ is a $2P_1+\nobreak P_3$, a contradiction.
Therefore every vertex of~$V_{4,5}$ is complete or anti-complete to $\{y,y'\}$.
Since~$y$ and~$y'$ were arbitrary non-adjacent vertices in~$V_{1,3}$ and~$V_{1,3}$ is a disjoint union of (at least two) cliques, it follows that every vertex of~$V_{4,5}$ is complete or anti-complete to~$V_{1,3}$.
Since every vertex of~$V_{1,3}$ has at most one neighbour in~$V_{4,5}$, at most one vertex in~$V_{4,5}$ is complete to~$V_{1,3}$.
If such a vertex exists then by Fact~\ref{fact:del-vert}, we may delete it.
Therefore we may assume that either~$V_{1,3}$ is a clique or~$V_{1,3}$ is anti-complete to~$V_{4,5}$.
The claim follows by symmetry.

\clmnonewline{\label{clm2:V13-V123-V13-indep-or-V123-clique-etc}We may assume the following: for $i \in \{1,2,3,4,5\}$, if~$V_{i,i+2}$ and~$V_{i,i+1,i+2}$ are both large then all of the following statements hold:
\begin{enumerate}[(i)]
\renewcommand{\theenumi}{(\roman{enumi})}
\renewcommand{\labelenumi}{(\roman{enumi})}
\item \label{clm2:V13-V123-V13-indep-or-V123-clique-etc-i}Either~$V_{i,i+2}$ is an independent set or~$V_{i,i+1,i+2}$ is a clique.
\item \label{clm2:V13-V123-V13-indep-or-V123-clique-etc-ii}If~$V_{i,i+2}$ is not an independent set then it is anti-complete to~$V_{i,i+1,i+2}$.
\item \label{clm2:V13-V123-V13-indep-or-V123-clique-etc-iii}If~$V_{i,i+1,i+2}$ is not a clique then it is complete to~$V_{i,i+2}$.
\end{enumerate}}
Suppose~$V_{1,3}$ and~$V_{1,2,3}$ are large.
By Claim~\ref{clm2:V13-P3-free}, $G[V_{1,3}]$ is $P_3$-free, so it is a disjoint union of cliques.
By Claim~\ref{clm2:V123-P1+P2-free}, $G[V_{1,2,3}]$ is $(P_1+\nobreak P_2)$-free, so its complement is a disjoint union of cliques.
We consider three cases.

\thmcase{$V_{1,2,3}$ is independent.}
We will show that in this case~$V_{1,3}$ must be an independent set which is complete to~$V_{1,2,3}$.
If $x \in V_{1,3}$ is non-adjacent to $y,y' \in V_{1,2,3}$ then $G[x,v_4,y,v_2,y']$ is a $2P_1+\nobreak P_3$, a contradiction.
Therefore every vertex in~$V_{1,3}$ has at most one non-neighbour in~$V_{1,2,3}$.
Suppose $x,x'\in V_{1,3}$ are adjacent.
Since~$V_{1,2,3}$ is large, there must be a vertex $y \in V_{1,2,3}$ that is adjacent to both~$x$ and~$x'$.
Then $G[y,v_3,x,v_2,x']$ is a $\overline{2P_1+P_3}$.
Therefore~$V_{1,3}$ must be independent.
If $x,x' \in V_{1,3}$ are non-adjacent and $y \in V_{1,2,3}$ is adjacent to~$x$, but not to~$x'$ then $G[x',v_4,x,y,v_2]$ is a $2P_1+\nobreak P_3$, a contradiction.
Therefore every vertex of~$V_{1,2,3}$ is either complete or anti-complete to~$V_{1,3}$.
Suppose there is a vertex $y \in V_{1,2,3}$ that is anti-complete to~$V_{1,3}$.
Since every vertex of~$V_{1,3}$ has at most one non-neighbour in~$V_{1,2,3}$, there must be a vertex $y' \in V_{1,2,3}$ that is complete to~$V_{1,3}$.
Now $G[v_4,y',x,y,x']$ is a $2P_1+\nobreak P_3$, a contradiction.
Therefore~$V_{1,2,3}$ is complete to~$V_{1,3}$.
We conclude that if~$V_{1,2,3}$ is an independent set then~$V_{1,3}$ must also be an independent set and furthermore~$V_{1,3}$ must be complete to~$V_{1,2,3}$.
By symmetry, if~$V_{i,i+1,i+2}$ is an independent set then Statements~\ref{clm2:V13-V123-V13-indep-or-V123-clique-etc-i}--\ref{clm2:V13-V123-V13-indep-or-V123-clique-etc-iii} of the claim hold.

\thmcase{$V_{1,3}$ is a clique.}
Casting to the complement as before, the clique~$V_{1,3}$ in~$G$ becomes the independent set~$W_{3,4,5}$ in~$\overline{G}$ and the set~$V_{1,2,3}$ in~$G$ becomes the set~$W_{3,5}$ in~$\overline{G}$.
By the above argument, this means that in~$\overline{G}$, $W_{3,5}$ must be an independent set and it must be complete to~$W_{3,4,5}$.
Therefore in~$G$ the set~$V_{1,2,3}$ must be a clique and it must be anti-complete to~$V_{1,3}$.
By symmetry, if~$V_{i,i+2}$ is a clique then Statements~\ref{clm2:V13-V123-V13-indep-or-V123-clique-etc-i}--\ref{clm2:V13-V123-V13-indep-or-V123-clique-etc-iii} of the claim hold.

\thmcase{$V_{1,2,3}$ is not independent and~$V_{1,3}$ is not a clique.}
If $x,x' \in V_{1,3}$ are non-adjacent and $y \in V_{1,2,3}$ is adjacent to~$x$, but not~$x'$ then $G[v_4,x',x,y,v_2]$ is a $2P_1+\nobreak P_3$, a contradiction.
Since~$G[V_{1,3}]$ is a disjoint union of (at least two) cliques, it follows that every vertex of~$V_{1,2,3}$ is complete or anti-complete to~$V_{1,3}$.
If $y,y' \in V_{1,2,3}$ are adjacent and $x \in V_{1,3}$ is adjacent to~$y$, but not~$y'$ then $G[y,v_1,y',x,v_2]$ is a $\overline{2P_1+P_3}$, a contradiction.
Since~$G[V_{1,2,3}]$ is the complement of a disjoint union of (at least two) cliques, it follows that every vertex of~$V_{1,3}$ is complete or anti-complete to~$V_{1,2,3}$.
We conclude that~$V_{1,3}$ is complete or anti-complete to~$V_{1,2,3}$.

Suppose for contradiction that~$V_{1,3}$ is not an independent set and~$V_{1,3}$ is complete to~$V_{1,2,3}$.
Choose adjacent vertices $x,x' \in V_{1,3}$ and adjacent vertices $y,y' \in V_{1,2,3}$.
Then $G[y,y',x,v_2,x']$ is a $\overline{2P_1+P_3}$, a contradiction.
Therefore either~$V_{1,3}$ is independent or it is anti-complete to~$V_{1,2,3}$.
By symmetry Statement~\ref{clm2:V13-V123-V13-indep-or-V123-clique-etc-ii} of the claim holds.

Suppose for contradiction that~$V_{1,2,3}$ is not a clique and~$V_{1,2,3}$ is anti-complete to~$V_{1,3}$.
Choose non-adjacent vertices $y,y' \in V_{1,2,3}$ and non-adjacent vertices $x,x' \in V_{1,3}$.
Then $G[x,x',y,v_2,y']$ is a $2P_1+\nobreak P_3$, a contradiction.
Therefore either~$V_{1,2,3}$ is a clique or it is complete to~$V_{1,2,3}$.
By symmetry Statement~\ref{clm2:V13-V123-V13-indep-or-V123-clique-etc-iii} of the claim holds.

Note that if~$V_{1,3}$ is not independent then it is anti-complete to~$V_{1,2,3}$ and that if~$V_{1,2,3}$ is not a clique then it is complete to~$V_{1,3}$.
Since~$V_{1,3}$ and~$V_{1,2,3}$ are large, it follows that either~$V_{1,2,3}$ is an independent set or~$V_{1,3}$ is a clique.
By symmetry Statement~\ref{clm2:V13-V123-V13-indep-or-V123-clique-etc-i} of the claim holds.
This completes the proof of Claim~\ref{clm2:V13-V123-V13-indep-or-V123-clique-etc}.

\clmnonewline{\label{clm2:V13-V234V125-either-all-cliques-V13-anti-or-V13-indep-comp}We may assume that: for $i \in \{1,2,3,4,5\}$ and $S \in \{\{i+1,i+2,i+3\},\allowbreak\{i,i+\nobreak 1,i+\nobreak 4\}\}$, if~$V_{i,i+2}$ and~$V_S$ are large then one of the following cases holds:
\begin{enumerate}[(i)]
\renewcommand{\theenumi}{(\roman{enumi})}
\renewcommand{\labelenumi}{(\roman{enumi})}
\item \label{clm2:V13-V234V125-either-all-cliques-V13-anti-or-V13-indep-comp-i}$V_{i,i+2}$ and~$V_S$ are cliques and~$V_{i,i+2}$ is anti-complete to~$V_S$.
\item \label{clm2:V13-V234V125-either-all-cliques-V13-anti-or-V13-indep-comp-ii}$V_{i,i+2}$ is independent and complete to~$V_S$.
\end{enumerate}}
Suppose~$V_{1,3}$ and~$V_{2,3,4}$ are large.
By Claim~\ref{clm2:V13-P3-free}, $G[V_{1,3}]$ is $P_3$-free, so it is a disjoint union of cliques.
By Claim~\ref{clm2:V123-P1+P2-free}, $G[V_{1,2,3}]$ is $(P_1+\nobreak P_2)$-free, so its complement is a disjoint union of cliques.

First suppose that~$V_{1,3}$ is not a clique.
Let $x,x' \in V_{1,3}$ be non-adjacent and suppose $y \in V_{2,3,4}$ is non-adjacent to~$x'$.
Then $G[x',v_5,x,y,v_2]$ or $G[x,x',v_2,y,v_4]$ is a $2P_1+\nobreak P_3$ if~$x$ is adjacent or non-adjacent to~$y$, respectively.
Since~$G[V_{1,3}]$ is a disjoint union of (at least two) cliques, this contradiction implies that~$V_{2,3,4}$ is complete to~$V_{1,3}$.
If $x,x' \in V_{1,3}$ are adjacent and $y \in V_{2,3,4}$ then $G[y,v_3,x,v_2,x']$ is a $\overline{2P_1+P_3}$, a contradiction.
Therefore~$V_{1,3}$ must be independent, so Statement~\ref{clm2:V13-V234V125-either-all-cliques-V13-anti-or-V13-indep-comp-ii} of the claim holds.

Now suppose that~$V_{1,3}$ is a clique.
Again, if $y \in V_{2,3,4}$ is adjacent to $x,x' \in V_{1,3}$ then $G[y,v_3,x,v_2,x']$ is a $\overline{2P_1+P_3}$, a contradiction.
Therefore every vertex of~$V_{2,3,4}$ has at most one neighbour in~$V_{1,3}$.
Suppose $y,y' \in V_{2,3,4}$ are non-adjacent.
Since~$V_{1,3}$ is large, there must be a vertex $x \in V_{1,3}$ that is non-adjacent to both~$y$ and~$y'$.
Now $G[x,v_5,y,v_2,y']$ is a $2P_1+\nobreak P_3$, a contradiction.
Therefore~$V_{2,3,4}$ must be a clique.
Suppose $x \in V_{1,3}$ and $y,y' \in V_{2,3,4}$ with~$x$ adjacent to~$y$, but not to~$y'$.
Then $G[y,v_3,y',x,v_2]$ is a $\overline{2P_1+P_3}$, a contradiction.
Therefore every vertex of~$V_{1,3}$ is either complete or anti-complete to~$V_{2,3,4}$.
Since every vertex of~$V_{2,3,4}$ has at most one neighbour in~$V_{1,3}$, at most one vertex of~$V_{1,3}$ is complete to~$V_{2,3,4}$.
If such a vertex exists then by Fact~\ref{fact:del-vert}, we may delete it.
Therefore we may assume that~$V_{1,3}$ is anti-complete to~$V_{2,3,4}$, so Statement~\ref{clm2:V13-V234V125-either-all-cliques-V13-anti-or-V13-indep-comp-i} of the claim holds.
The claim follows by symmetry.

\clmnonewline{\label{clm2:V123-V14V35-either-all-indep-V123-comp-or-V123-clique-anti}We may assume that: for $i \in \{1,2,3,4,5\}$ and $S \in \{\{i+2,i+3,i+4\},\allowbreak \{i,i+\nobreak 3,i+\nobreak 4,\}\}$, if~$V_{i,i+2}$ and~$V_S$ are large then one of the following cases holds:
\begin{enumerate}[(i)]
\item $V_{i,i+2}$ and~$V_S$ are independent and~$V_{i,i+2}$ is complete to~$V_S$.
\item $V_S$ is a clique and anti-complete to~$V_{i,i+2}$.
\end{enumerate}}
By symmetry we need only prove the claim for the case where $i=1$ and $S=\{3,4,5\}$.
In this case the sets~$V_{i,i+2}$ and~$V_S$ are~$V_{1,3}$ and~$V_{3,4,5}$, respectively, which are equal to~$W_{3,4,5}$ and~$W_{1,4}$, respectively (see also Table~\ref{tbl:WT-VS-corr}).
The claim follows by casting to the complement and applying Claim~\ref{clm2:V13-V234V125-either-all-cliques-V13-anti-or-V13-indep-comp}.

\clm{\label{clm2:V13-V124V235-V13-indep-comp}For $i \in \{1,2,3,4,5\}$ and $S \in \{\{i,i+1,i+3\},\{i+1,i+2,i+4\}\}$, if~$V_{i,i+2}$ and~$V_S$ are large then~$V_{i,i+2}$ is independent and it is complete to~$V_S$.}
By Claim~\ref{clm2:V124V1234V12345-indep}, $V_{1,2,4}$ is independent.
Suppose $x \in V_{1,3}$ and $y,y' \in V_{1,2,4}$ with~$x$ non-adjacent to~$y'$.
Then $G[y',v_5,v_3,x,y]$ or $G[v_5,x,y,v_2,y']$ is a $2P_1+\nobreak P_3$ if~$x$ is adjacent or non-adjacent to~$y$, respectively.
Therefore~$V_{1,3}$ is complete to~$V_{1,2,4}$.
If $x,x' \in V_{1,3}$ are adjacent and $y \in V_{1,2,4}$ then $G[v_1,y,x,v_2,x']$ is a $\overline{2P_1+P_3}$, a contradiction.
Therefore~$V_{1,3}$ is independent.
The claim follows by symmetry.

\clm{\label{clm2:V13-V245-bdd-cw}For $i \in \{1,2,3,4,5\}$, if~$V_{i,i+2}$ and~$V_{i+1,i+3,i+4}$ are large then $G[V_{i,i+2} \cup V_{i+1,i+3,i+4}]$ has bounded clique-width.}
Suppose~$V_{1,3}$ and~$V_{2,4,5}$ are large.
By Claim~\ref{clm2:V13-P3-free}, $G[V_{1,3}]$ is $P_3$-free, so it is a disjoint union of cliques.
By Claim~\ref{clm2:V124V1234V12345-indep}, $V_{2,4,5}$ is independent.
If $x \in V_{1,3}$ is non-adjacent to $y,y' \in V_{2,4,5}$ then $G[y,y',v_1,x,v_3]$ is a $2P_1+\nobreak P_3$, a contradiction.
Therefore every vertex of~$V_{1,3}$ has at most one non-neighbour in~$V_{2,4,5}$.
If $x,x' \in V_{1,3}$ are non-adjacent and $y \in V_{2,4,5}$ is non-adjacent to~$x$ and~$x'$ then $G[x,x',v_2,y,v_4]$ is a $2P_1+\nobreak P_3$, a contradiction.
Therefore every vertex of~$V_{2,4,5}$ is complete to all but at most one component of~$V_{1,3}$.
Let~$G'$ be the graph obtained from $G[V_{1,3} \cup V_{2,4,5}]$ by applying a bipartite complementation between~$V_{1,3}$ and~$V_{2,4,5}$.
By Fact~\ref{fact:bip}, $G[V_{1,3} \cup V_{2,4,5}]$ has bounded clique-width if and only if every component of~$G'$ has bounded clique-width.
Every component~$C^{G'}$ of~$G'$ consists of either a single vertex (in which case it has clique-width~$1$) or a clique in~$V_{1,3}$ together with an independent set in~$V_{2,4,5}$, no two vertices of which have a common neighbour in the clique.
By Fact~\ref{fact:comp}, we may complement the clique on $V(C^{G'}) \cap V_{1,3}$.
The obtained graph will be a disjoint union of stars, which have clique-width at most~$2$.

\clm{\label{clm2:V13-V1234V1235-V13-indep-comp}For $i \in \{1,2,3,4,5\}$ and $S \in \{\{i,i+1,i+2,i+3\},\{i,i+1,i+2,i+4\}\}$, if~$V_{i,i+2}$ and~$V_S$ are large then~$V_{i,i+2}$ is independent and it is complete to~$V_S$.}
Suppose~$V_{1,3}$ and~$V_{1,2,3,4}$ are large.
By Claim~\ref{clm2:V13-P3-free}, $G[V_{1,3}]$ is $P_3$-free, so it is a disjoint union of cliques.
By Claim~\ref{clm2:V124V1234V12345-indep}, $V_{1,2,3,4}$ is independent.

Suppose $x \in V_{1,3}$ has two non-neighbours $y,y' \in V_{1,2,3,4}$.
Then $G[x,v_5,y,v_2,y']$ is a $2P_1+\nobreak P_3$, a contradiction.
It follows that every vertex of~$V_{1,3}$ has at most one non-neighbour in~$V_{1,2,3,4}$.

Suppose, for contradiction that~$V_{1,3}$ is not an independent set.
Let $x,x' \in V_{1,3}$ be adjacent vertices.
If $y \in V_{1,2,3,4}$ is adjacent to both~$x$ and~$x'$ then $G[y,v_3,x,v_2,x']$ is a $\overline{2P_1+P_3}$, a contradiction.
Therefore every vertex of~$V_{1,2,3,4}$ has at most one neighbour in~$\{x,x'\}$.
Since~$V_{1,2,3,4}$ is large, there must be two vertices $y',y'' \in V_{1,2,3,4}$ that are non-adjacent to the same vertex in~$\{x,x'\}$.
This is a contradiction since every vertex of~$V_{1,3}$ has at most one non-neighbour in~$V_{1,2,3,4}$.
It follows that~$V_{1,3}$ is an independent set.
Since $5 \notin \{1,3\} \cup \{1,2,3,4\}$, Claim~\ref{clm2:indeps-no-i-comp} implies that~$V_{1,3}$ is complete to~$V_{1,2,3,4}$.
The claim follows by symmetry.

\clmnonewline{\label{clm2:V13-V1245V2345V12345-either-V13-indep-or-V13-disjUn-indep-comp-clique-anti}We may assume that: for $i \in \{1,2,3,4,5\}$ and $S \in \{\{i,i+1,i+3,i+4\},\allowbreak \{i+1,i+2,i+3,i+4\},\allowbreak \{1,2,3,4,5\}\}$, if~$V_{i,i+2}$ and~$V_S$ are large then one of the following holds:
\begin{enumerate}[(i)]
\item $V_{i,i+2}$ is an independent set or
\item $V_{i,i+2}$ is the disjoint union of a (possibly empty) clique that is anti-complete to~$V_S$ and an (possibly empty) independent set that is complete to~$V_S$.
\end{enumerate}
}
Suppose~$V_{1,3}$ and~$V_S$ are large for $S \in \{\{1,2,4,5\},\allowbreak \{1,2,3,4,5\}\}$ (the $S=\{2,3,4,5\}$ case is symmetric).
By Claim~\ref{clm2:V13-P3-free}, $G[V_{1,3}]$ is $P_3$-free, so it is a disjoint union of cliques.
By Claim~\ref{clm2:V124V1234V12345-indep}, $V_S$ is independent.

If $x,x' \in V_{1,3}$ are non-adjacent and $y \in V_S$ is anti-complete to $\{x,x'\}$ then $G[x,x',v_2,y,v_4]$ is a $2P_1+\nobreak P_3$, a contradiction.
Since $G[V_{1,3}]$ is a disjoint union of cliques, it follows that every vertex of~$V_S$ is complete to all but at most one component of~$G[V_{1,3}]$.
If $x,x' \in V_{1,3}$ are adjacent and $y \in V_S$ is complete to $\{x,x'\}$ then $G[v_1,y,x,v_2,x']$ is a $\overline{2P_1+P_3}$, a contradiction.
Therefore no vertex of~$V_S$ has two neighbours in the same component of~$G[V_{1,3}]$.
It follows that~$G[V_{1,3}]$ contains at most one non-trivial component.
In other words, either~$V_{1,3}$ is an independent set or the disjoint union of a clique and an independent set.

Suppose that~$V_{1,3}$ is not an independent set.
Then~$G[V_{1,3}]$ contains a non-trivial component~$C'$.
We may assume~$C'$ contains at least three vertices, otherwise we may delete it by Fact~\ref{fact:del-vert}.
No vertex of~$V_S$ can have two neighbours in~$C'$ and every vertex of~$V_S$ is complete to all but at most one component of~$G[V_{1,3}]$.
Therefore every vertex of~$V_S$ is complete to the independent set $V_{1,3}\setminus V(C')$.
Suppose~$x \in V_S$ has a neighbour $y \in V(C')$.
Since~$V(C')$ contains at least three vertices, and every vertex of~$V_S$ has at most one neighbour in~$V(C')$, we can find vertices $y',y'' \in V(C')$ that are non-adjacent to~$x$.
Now $G[v_1,y,y',x,y'']$ is a $\overline{2P_1+P_3}$, a contradiction.
Therefore~$V_S$ is anti-complete to~$V(C')$.

We conclude that either~$V_{1,3}$ is an independent set or it is the disjoint union of an independent set that is complete to~$V_S$ and a clique that is anti-complete to~$V_S$.
The claim follows by symmetry.

\clm{\label{clm2:V13-V1345-V13-indep-comp}For $i \in \{1,2,3,4,5\}$, if~$V_{i,i+2}$ and~$V_{i,i+2,i+3,i+4}$ are large then~$V_{i,i+2}$ is an independent set that is complete to~$V_{i,i+2,i+3,i+4}$.}
Suppose~$V_{1,3}$ and~$V_{1,3,4,5}$ are large.
By Claim~\ref{clm2:V13-P3-free}, $G[V_{1,3}]$ is $P_3$-free, so it is a disjoint union of cliques.
By Claim~\ref{clm2:V124V1234V12345-indep}, $V_{1,3,4,5}$ is independent.

Suppose $x \in V_{1,3}$ is non-adjacent to $y,y' \in V_{1,3,4,5}$.
Then $G[v_2,x,y,v_4,y']$ is a $2P_1+\nobreak P_3$, a contradiction.
Therefore every vertex of~$V_{1,3}$ has at most one non-neighbour in~$V_{1,3,4,5}$.
Suppose $x,x' \in V_{1,3}$ are adjacent.
Since~$V_{1,3,4,5}$ is large, there must be a vertex $y \in V_{1,3,4,5}$ that is adjacent to both~$x$ and~$x'$.
Now $G[y,v_1,x,v_5,x']$ is an $\overline{2P_1+P_3}$, a contradiction.
Therefore~$V_{1,3}$ must be an independent set.
If $x,x' \in V_{1,3}$ and $y \in V_{1,3,4,5}$ is adjacent to~$x$, but not to~$x'$ then $G[x',v_2,x,y,v_4]$ is a $2P_1+\nobreak P_3$, a contradiction.
Therefore every vertex of~$V_{1,3,4,5}$ must be either complete or anti-complete to~$V_{1,3}$.
Since every vertex of~$V_{1,3}$ has at most one non-neighbour in~$V_{1,3,4,5}$, it follows that at most one vertex of~$V_{1,3,4,5}$ may be anti-complete to~$V_{1,3}$.
Suppose $x,x' \in V_{1,3}$.
If $y \in V_{1,3,4,5}$ is anti-complete to~$V_{1,3}$ and~$y' \in V_{1,3,4,5}$ is complete to~$V_{1,3}$ then $G[y,v_2,x,y',x']$ is a $2P_1+\nobreak P_3$.
We conclude that~$V_{1,3}$ is complete to~$V_{1,3,4,5}$.
The claim follows by symmetry.

\medskip
\noindent
The next two claims will allow us to assume that every set~$V_S$ is either a clique or an independent set.

\clm{\label{clm2:V13-indep-or-clique}For $i \in \{1,2,3,4,5\}$, if~$V_{i,i+2}$ is large then we may assume it is an independent set or a clique.}
Suppose~$V_{1,3}$ is large and that it is not a clique or an independent set.

By Claim~\ref{clm2:V_S-V_T-not-contain13-empty}, if~$V_S$ is large for some $S \subseteq \{1,2,3,4,5\}$ with $S \neq \{1,3\}$ then $S \cap \{2,4\} \neq \emptyset$ and $S \cap \{2,5\} \neq \emptyset$.
It follows that $V_\emptyset$, $V_4$, $V_5$, $V_1$, $V_{1,4}$, $V_{1,5}$, $V_3$, $V_{3,4}$, $V_{3,5}$, $V_{1,3,4}$ and~$V_{1,3,5}$ are empty.
$V_2$ is empty by Claim~\ref{clm2:V13-V2-V13-indep-clique}.
$V_{1,2}$ and~$V_{2,3}$ are empty by Claim~\ref{clm2:V13-V12V23-V13-clique-anti-or-indep-comp}.
$V_{2,3,4}$ and~$V_{1,2,5}$ are empty by Claim~\ref{clm2:V13-V234V125-either-all-cliques-V13-anti-or-V13-indep-comp}.
$V_{1,2,3,4}$ and~$V_{1,2,3,5}$ are empty by Claim~\ref{clm2:V13-V1234V1235-V13-indep-comp}.
$V_{1,3,4,5}$ is empty by Claim~\ref{clm2:V13-V1345-V13-indep-comp}.
$V_{1,2,4}$ and~$V_{2,3,5}$ are empty by Claim~\ref{clm2:V13-V124V235-V13-indep-comp}.
Therefore in addition to~$V_{1,3}$, only the following sets may be non-empty:
$V_{4,5}$, $V_{2,4}$, $V_{2,5}$, $V_{1,2,3}$, $V_{3,4,5}$, $V_{1,4,5}$, $V_{2,4,5}$, $V_{1,2,4,5}$, $V_{2,3,4,5}$ and $V_{1,2,3,4,5}$.

Suppose $V_{2,4,5}=W_{1,2}$ is large.
By Claim~\ref{clm2:V_S-V_T-not-contain13-empty} if a set~$V_S$ is large for some $S \subseteq \{1,2,3,4,5\}$ with $S \neq \{2,4,5\}$ then $S \cap \{1,3\} \neq \emptyset$.
It follows that~$V_{4,5}$, $V_{2,4}$ and~$V_{2,5}$ are empty.
By Claim~\ref{clm2:V_S-V_T-contain12-empty} if~$V_S$ is large for some $S \subseteq \{1,2,3,4,5\}$ with $S \neq \{2,4,5\}$ then $\{4,5\} \not \subseteq S$.
It follows that~$V_{3,4,5}$, $V_{1,4,5}$, $V_{1,2,4,5}$, $V_{2,3,4,5}$ and~$V_{1,2,3,4,5}$ are empty.
Therefore, apart from~$V_{1,3}$ and~$V_{2,4,5}$, only the set~$V_{1,2,3}$ can be large.
By Claim~\ref{clm2:V13-V123-V13-indep-or-V123-clique-etc}, if $V_{1,2,3}$ is large then it is a clique that is anti-complete to~$V_{1,3}$.
Casting to the complement (see also Table~\ref{tbl:WT-VS-corr}), since~$V_{1,2,3}$ is a clique in~$G$, it follows that $W_{3,5}=V_{1,2,3}$ is an independent set in~$\overline{G}$, so by Claim~\ref{clm2:V13-V45-V13-clique-or-anti}, $W_{3,5}$ is anti-complete to~$W_{1,2}=V_{2,4,5}$ in~$\overline{G}$.
Therefore in the graph~$G$, $V_{1,2,3}=W_{3,5}$ is a clique that is complete to~$V_{2,4,5}=W_{1,2}$ and anti-complete to~$V_{1,3}$.
By Fact~\ref{fact:del-vert}, we may delete the five vertices in the original cycle~$C$.
By Fact~\ref{fact:bip}, we may apply a bipartite complementation between~$V_{2,4,5}$ and~$V_{1,2,3}$.
This separates the graph into two parts: $G[V_{1,3} \cup V_{2,4,5}]$, which has bounded clique-width by Claim~\ref{clm2:V13-V245-bdd-cw} and $G[V_{1,2,3}]$, which is a clique and so has clique-width at most~$2$.
Therefore if~$V_{2,4,5}$ is large then~$G$ has bounded clique-width.
Thus we may assume that $V_{2,4,5}= \emptyset$.

We will now show how to disconnect~$V_{1,3}$ from the rest of the graph.
Note that~$V_{4,5}$ is anti-complete to~$V_{1,3}$ by Claim~\ref{clm2:V13-V45-V13-clique-or-anti}.
$V_{1,2,3}$ is anti-complete to~$V_{1,3}$ by Claim~\ref{clm2:V13-V123-V13-indep-or-V123-clique-etc}.
$V_{3,4,5}$ and~$V_{1,4,5}$ are anti-complete to~$V_{1,3}$ by Claim~\ref{clm2:V123-V14V35-either-all-indep-V123-comp-or-V123-clique-anti}.
$V_{2,4}$ and~$V_{2,5}$ are complete to~$V_{1,3}$ by Claim~\ref{clm2:V13-V24-both-cliques-or-one-indep-and-comp}.
By Fact~\ref{fact:bip}, we may apply a bipartite complementation between~$V_{1,3}$ and $\{v_1,v_3\} \cup V_{2,4} \cup V_{2,5}$.
By Claim~\ref{clm2:V13-V1245V2345V12345-either-V13-indep-or-V13-disjUn-indep-comp-clique-anti}, for $S \in \{\{1,2,4,5\}, \{2,3,4,5\}, \{1,2,3,4,5\}\}$, either~$V_S$ is empty or~$V_{1,3}$ is the disjoint union of a clique~$C'$ that is anti-complete to~$V_S$ and an independent set~$I$ that is complete to~$V_S$.
If~$V_{1,3}$ does have this form, then by Fact~\ref{fact:bip}, we may apply a bipartite complementation between~$I$ and $V_{1,2,4,5} \cup V_{2,3,4,5} \cup V_{1,2,3,4,5}$.
Doing this removes all edges from~$V_{1,3}$ to vertices not in~$V_{1,3}$.
By Claim~\ref{clm2:V13-P3-free}, $G[V_{1,3}]$ is a $P_3$-free graph, so it is a disjoint union of cliques and thus has clique-width at most~$2$.

We conclude that if~$V_{1,3}$ is large, but is neither a clique nor an independent set, then~$G[V_{1,3}]$ has bounded clique-width and we can remove all edges from~$V_{1,3}$ to vertices not in~$V_{1,3}$.
We may therefore remove all vertices in~$V_{1,3}$ from the graph.
The claim follows by symmetry.

\clm{\label{clm2:V123-indep-or-clique}For $i \in \{1,2,3,4,5\}$, if~$V_{i,i+1,i+2}$ is large then we may assume it is an independent set or a clique.}
This follows from Claim~\ref{clm2:V13-indep-or-clique} by casting to the complement (see also Table~\ref{tbl:WT-VS-corr}).

\medskip
\noindent
Note that by Claims~\ref{clm2:V0V1V12-clique}, \ref{clm2:V124V1234V12345-indep}, \ref{clm2:V13-indep-or-clique} and~\ref{clm2:V123-indep-or-clique}, we may assume that every large set~$V_S$ is either a clique or an independent set.

\clm{\label{clm2:V13-indep}For $i \in \{1,2,3,4,5\}$, if~$V_{i,i+2}$ is large then we may assume it is an independent set.}
Suppose that~$V_{1,3}$ is large, but not an independent set.
By Claim~\ref{clm2:V13-indep-or-clique}, we may assume that it is a clique.
We will show how to disconnect~$V_{1,3}$ (or a part of the graph that contains~$V_{1,3}$ and has bounded clique-width) from the rest of the graph.
First, by Fact~\ref{fact:del-vert}, we may delete the five vertices of the original cycle~$C$.
Let $G' =G[\bigcup V_S \; | \; S \mbox{ is a clique}]$ and let $G'' =G[\bigcup V_S \; | \; S \mbox{ is an independent set}]$.
By Claim~\ref{clm2:clique-indep-trivial}, if~$V_S$ is a clique and~$V_T$ is an independent set, then~$V_S$ is either complete or anti-complete to~$V_T$.
If~$V_S$ is complete to~$V_T$, by Fact~\ref{fact:bip}, we may apply a bipartite complementation between these sets.
Doing so for every pair of a clique~$V_S$ and an independent set~$V_T$ that are complete to each other, we disconnect~$G'$ from~$G''$.
Since our aim is to show how to remove the clique~$V_{1,3}$ from~$G$, it is therefore sufficient to show how to remove it from~$G'$.
In other words, we may assume that if~$V_T$ is an independent set then $V_T=\emptyset$.
That is, we may assume that every set~$V_S$ is a (possibly empty) clique.

By Claim~\ref{clm2:V124V1234V12345-indep}, $V_{1,2,4}$, $V_{2,3,5}$, $V_{1,3,4}$, $V_{2,4,5}$, $V_{1,3,5}$, $V_{1,2,3,4}$, $V_{1,2,3,5}$, $V_{1,2,4,5}$, $V_{1,3,4,5}$, $V_{2,3,4,5}$ and~$V_{1,2,3,4,5}$ are independent sets, so we may assume that they are empty.
Since~$V_{1,3}$ is a large, by Claim~\ref{clm2:V_S-V_T-not-contain13-empty} if~$V_S$ is large for some $S \subseteq \{1,2,3,4,5\}$ with $S \neq \{1,3\}$ then $S \cap \{2,4\} \neq \emptyset$ and $S \cap \{2,5\} \neq \emptyset$.
It follows that $V_\emptyset$, $V_4$, $V_5$, $V_1$, $V_{1,4}$, $V_{1,5}$, $V_3$, $V_{3,4}$, $V_{3,5}$ are empty.
This means that apart from~$V_{1,3}$, only the following sets can be large: $V_{1,2}$, $V_{2,3}$, $V_{1,2,3}$, $V_{2,3,4}$, $V_{3,4,5}$, $V_{1,4,5}$, $V_{1,2,5}$, $V_2$, $V_{4,5}$, $V_{2,4}$, $V_{2,5}$ and recall that all these sets are (possibly empty) cliques by assumption (see also \figurename~\ref{fig:V13-clique}).
For two of these sets, if there is an $i \in S \cap T$ then~$V_S$ is anti-complete to~$V_T$ by Claim~\ref{clm2:cliques-i-anti}.
Since $\{1,3\} \cap (\{2\} \cup \{4,5\} \cup \{2,4\} \cup \{2,5\}) = \emptyset$, at most one of the sets~$V_2$, $V_{4,5}$, $V_{2,4}$ and~$V_{2,5}$ is large by Claim~\ref{clm2:V_S-V_T-not-contain13-empty}.
We consider several cases.

\begin{figure}
\begin{center}
\begin{tikzpicture}[scale=0.65]
\coordinate (v13)  at (360*0/12+90:8) ;
\coordinate (v23)  at (360*1/12+90:8) ;
\coordinate (v12)  at (360*2/12+90:8) ;
\coordinate (v145) at (360*3/12+90:8) ;
\coordinate (v345) at (360*4/12+90:8) ;
\coordinate (v125) at (360*5/12+90:8) ;
\coordinate (v234) at (360*6/12+90:8) ;
\coordinate (v123) at (360*7/12+90:8) ;
\coordinate (v25)  at (360*8/12+90:8) ;
\coordinate (v24)  at (360*9/12+90:8) ;
\coordinate (v45)  at (360*10/12+90:8) ;
\coordinate (v2)   at (360*11/12+90:8) ;

\draw [fill=black] (v13)  circle (1.5pt) ;
\draw [fill=black] (v23)  circle (1.5pt) ;
\draw [fill=black] (v12)  circle (1.5pt) ;
\draw [fill=black] (v145) circle (1.5pt) ;
\draw [fill=black] (v345) circle (1.5pt) ;
\draw [fill=black] (v125) circle (1.5pt) ;
\draw [fill=black] (v234) circle (1.5pt) ;
\draw [fill=black] (v123) circle (1.5pt) ;
\draw [fill=black] (v25)  circle (1.5pt) ;
\draw [fill=black] (v24)  circle (1.5pt) ;
\draw [fill=black] (v45)  circle (1.5pt) ;
\draw [fill=black] (v2)   circle (1.5pt) ;

\draw (v13)  node [label={[label distance=-0pt]+90:$V_{1,3}$}] {};
\draw (v23)  node [label={[label distance=-0pt]+90:$V_{2,3}$}] {};
\draw (v12)  node [label={[label distance=-0pt]+180:$V_{1,2}$}] {};
\draw (v145) node [label={[label distance=-0pt]+180:$V_{1,4,5}$}] {};
\draw (v345) node [label={[label distance=-0pt]+180:$V_{3,4,5}$}] {};
\draw (v125) node [label={[label distance=-0pt]+270:$V_{1,2,5}$}] {};
\draw (v234) node [label={[label distance=-0pt]+270:$V_{2,3,4}$}] {};
\draw (v123) node [label={[label distance=-0pt]+270:$V_{1,2,3}$}] {};
\draw (v25)  node [label={[label distance=-0pt]+0:$V_{2,5}$}] {};
\draw (v24)  node [label={[label distance=-0pt]+0:$V_{2,4}$}] {};
\draw (v45)  node [label={[label distance=-0pt]+0:$V_{4,5}$}] {};
\draw (v2)   node [label={[label distance=-0pt]+90:$V_2$}] {};

\draw (v13) -- (v2);
\draw (v13) -- (v45);
\draw (v13) -- (v24);
\draw (v13) -- (v25);

\draw[dashed] (v2) -- (v45);
\draw[dashed] (v2) -- (v24);
\draw[dashed] (v2) -- (v25);
\draw (v2) -- (v345);
\draw (v2) -- (v145);
\draw[dashed] (v2) -- (v12);
\draw[dashed] (v2) -- (v23);

\draw[dashed] (v45) -- (v24);
\draw[dashed] (v45) -- (v25);
\draw (v45) -- (v123);
\draw[dashed] (v45) -- (v345);
\draw[dashed] (v45) -- (v145);
\draw (v45) -- (v12);
\draw (v45) -- (v23);

\draw[dashed] (v24) -- (v25);
\draw[dashed] (v24) -- (v12);

\draw[dashed] (v25) -- (v23);

\draw[dashed] (v123) -- (v234);
\draw[dashed] (v123) -- (v125);
\draw[dashed] (v123) -- (v12);
\draw[dashed] (v123) -- (v23);

\draw[dashed] (v234) -- (v345);
\draw[dashed] (v234) -- (v23);

\draw[dashed] (v125) -- (v145);
\draw[dashed] (v125) -- (v12);

\draw[dashed] (v345) -- (v145);
\draw (v345) -- (v12);

\draw (v145) -- (v23);
\end{tikzpicture}
\end{center}
\caption{\label{fig:V13-clique}The set of possible cliques when~$V_{1,3}$ is a clique.
Two sets are joined by a line if the edges between them form a matching (recall that a matching may contain no edges, in which case the two sets are anti-complete to each other).
Two sets are joined by a dashed line if at most one of them is large and the other is empty.
Two sets are not joined by a line if they are anti-complete to each other.
These properties follow from Claims~\ref{clm2:clique-matching}, \ref{clm2:cliques-i-anti}, \ref{clm2:V_S-V_T-contain12-empty}, and~\ref{clm2:V_S-V_T-not-contain13-empty}.}
\end{figure}
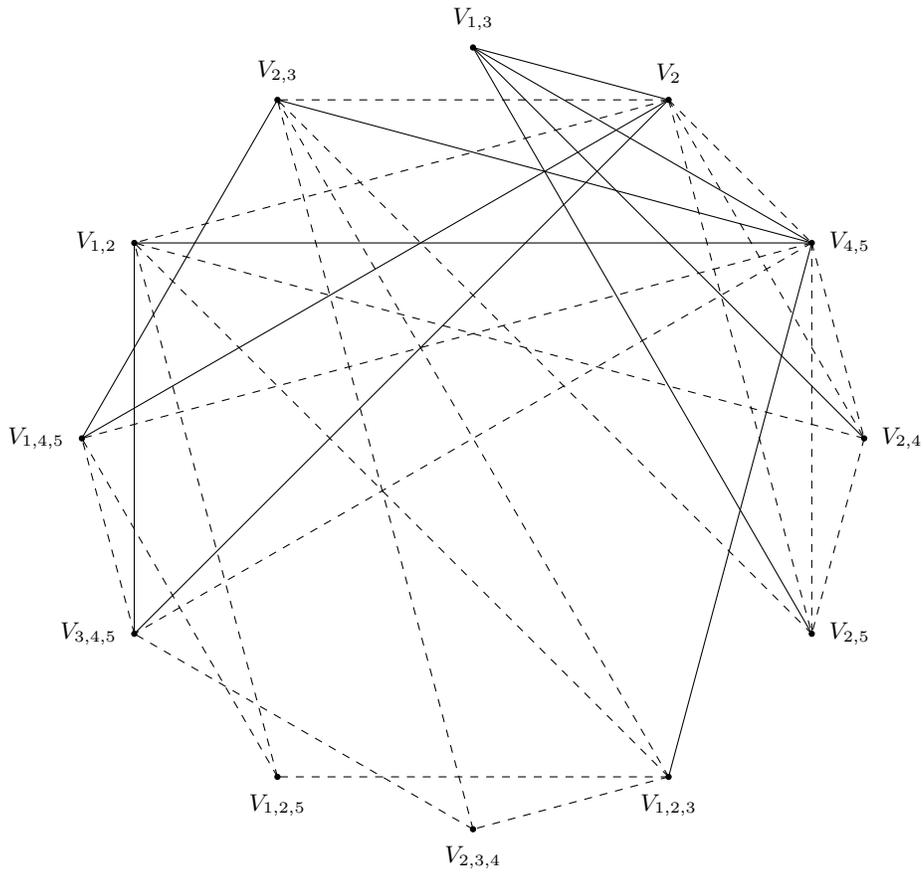

\thmcase{$V_{2,4}$ or~$V_{2,5}$ is large.}
By symmetry, we may assume~$V_{2,4}$ is large.
Then~$V_2$, $V_{4,5}$ and~$V_{2,5}$ are empty, as stated above.
Also, $V_{1,3}$ and~$V_{2,4}$ are anti-complete to~$V_{1,2}$, $V_{2,3}$, $V_{1,2,3}$, $V_{2,3,4}$, $V_{3,4,5}$, $V_{1,4,5}$ and~$V_{1,2,5}$ by Claim~\ref{clm2:cliques-i-anti}.
This means that $G[V_{1,3} \cup V_{2,4}]$ is disconnected from the rest of~$G'$.
By Lemma~\ref{l-victor}, $G[V_{1,3} \cup V_{2,4}]$ has bounded clique-width.
This completes the case.

\thmcase{$V_{2}$ is large.}
Then~$V_{4,5}$, $V_{2,4}$ and~$V_{2,5}$ are empty, as stated above.
Since $\{3,5\} \notin \{2\} \cup \{1,2\}$ and $\{1,4\} \notin \{2\} \cup \{2,3\}$, Claim~\ref{clm2:V_S-V_T-not-contain13-empty} implies that $V_{1,2}$ and~$V_{2,3}$ are empty.
Now $V_{1,3}$, $V_{1,2,3}$, $V_{2,3,4}$, $V_{3,4,5}$, $V_{1,4,5}$ and~$V_{1,2,5}$ are pair-wise anti-complete by Claim~\ref{clm2:cliques-i-anti}.
By Claim~\ref{clm2:clique-matching}, the edges between~$V_2$ and each of $V_{1,3}$, $V_{1,2,3}$, $V_{2,3,4}$, $V_{3,4,5}$, $V_{1,4,5}$ and~$V_{1,2,5}$ form matchings.
By Fact~\ref{fact:comp}, we can complement all of the large sets.
We obtain a graph which is a disjoint union of stars, which have clique-width at most~$2$.
It follows that $G'=G[V_{1,3} \cup V_{1,2,3} \cup V_{2,3,4} \cup V_{3,4,5} \cup V_{1,4,5} \cup V_{1,2,5}]$ has bounded clique-width.
This completes the case.

\thmcase{$V_{4,5}$ is large.}
Then~$V_2$, $V_{2,4}$ and~$V_{2,5}$ are empty, as stated above.
Since~$V_{4,5}$ is large and $\{4,5\} \subseteq \{3,4,5\}, \{1,4,5\}$, Claim~\ref{clm2:V_S-V_T-contain12-empty} implies that~$V_{3,4,5}$ and~$V_{1,4,5}$ are empty.
Now $V_{1,3}$, $V_{1,2}$, $V_{2,3}$, $V_{1,2,3}$, $V_{2,3,4}$ and~$V_{1,2,5}$ are pair-wise anti-complete by Claim~\ref{clm2:cliques-i-anti}.
By Claim~\ref{clm2:clique-matching}, the edges between~$V_{4,5}$ and each of $V_{1,3}$, $V_{1,2}$, $V_{2,3}$, $V_{1,2,3}$, $V_{2,3,4}$ and~$V_{1,2,5}$ form matchings.
By Fact~\ref{fact:comp}, we can complement all of the large sets.
We obtain a graph which is a disjoint union of stars, which have clique-width at most~$2$.
It follows that $G'=G[V_{1,3} \cup V_{4,5} \cup V_{1,2} \cup V_{2,3} \cup V_{1,2,3} \cup V_{2,3,4} \cup V_{1,2,5}]$ has bounded clique-width.
This completes the case.

\thmcase{$V_2$, $V_{4,5}$, $V_{2,4}$ and~$V_{2,5}$ are empty.}
$V_{1,3}$ is anti-complete to $V_{1,2}$, $V_{2,3}$, $V_{1,2,3}$, $V_{2,3,4}$, $V_{3,4,5}$, $V_{1,4,5}$ and~$V_{1,2,5}$ by Claim~\ref{clm2:cliques-i-anti}.
Therefore $G'[V_{1,3}]=G[V_{1,3}]$ is disconnected from rest of~$G'$.
Since~$V_{1,3}$ is a clique, $G[V_{1,3}]$ has clique-width at most~$2$.
We may therefore remove~$V_{1,3}$ from the graph.
This completes the case.

\medskip
\noindent
Since one of the above cases must hold by Claim~\ref{clm2:V_S-large-or-empty}, this completes the proof of the claim when $i=1$.
The claim follows by symmetry.

\clm{\label{clm2:V123-clique}For $i \in \{1,2,3,4,5\}$, if~$V_{i,i+1,i+2}$ is large then we may assume it is a clique.}
This follows from Claim~\ref{clm2:V13-indep} by casting to the complement (see also Table~\ref{tbl:WT-VS-corr}).

\clm{\label{clm2:V13-empty}For $i \in \{1,2,3,4,5\}$, we may assume~$V_{i,i+2}$ is empty.}
Suppose that~$V_{1,3}$ is large.
By Claim~\ref{clm2:V13-indep}, we may assume that it is an independent set.
We will show how to disconnect~$V_{1,3}$ (or a part of the graph that contains~$V_{1,3}$ and has bounded clique-width) from the rest of the graph.
First, by Fact~\ref{fact:del-vert}, we may delete the five vertices of the original cycle~$C$.
Let $G' =G[\bigcup V_S \; | \; S \mbox{ is a clique}]$ and let $G'' =G[\bigcup V_S \; | \; S \mbox{ is an independent set}]$.
By Claim~\ref{clm2:clique-indep-trivial}, if~$V_S$ is a clique and~$V_T$ is an independent set, then~$V_S$ is either complete or anti-complete to~$V_T$.
If~$V_S$ is complete to~$V_T$, by Fact~\ref{fact:bip}, we may apply a bipartite complementation between these sets.
Doing so for every pair of a clique~$V_S$ and an independent set~$V_T$ that are complete to each other, we disconnect~$G'$ from~$G''$.
Since our aim is to show how to remove the independent set~$V_{1,3}$ from~$G$, it is therefore sufficient to show how to remove it from~$G''$.
In other words, we may assume that if~$V_S$ is a clique then $V_S=\emptyset$.
That is, we may assume that every set~$V_T$ is a (possibly empty) independent set.

By Claim~\ref{clm2:V0V1V12-clique} $V_\emptyset$, $V_1$, $V_2$, $V_3$, $V_4$, $V_5$, $V_{1,2}$, $V_{2,3}$, $V_{3,4}$, $V_{4,5}$ and~$V_{1,5}$ are cliques, so we may assume that they are empty.
By Claim~\ref{clm2:V123-clique} $V_{1,2,3}$, $V_{2,3,4}$, $V_{3,4,5}$, $V_{1,4,5}$ and $V_{1,2,5}$ are cliques, so we may assume that they are empty.

Since~$V_{1,3}$ is a large, by Claim~\ref{clm2:V_S-V_T-not-contain13-empty} if~$V_S$ is large for some $S \subseteq \{1,2,3,4,5\}$ with $S \neq \{1,3\}$ then $S \cap \{2,4\} \neq \emptyset$ and $S \cap \{2,5\} \neq \emptyset$.
It follows that $V_{1,4}$, $V_{3,5}$, $V_{1,3,4}$ and~$V_{1,3,5}$ are empty.

This means that apart from~$V_{1,3}$, only the following sets can be large: $V_{2,4}$, $V_{2,5}$, $V_{1,2,4}$, $V_{2,3,5}$, $V_{1,2,3,4}$, $V_{1,2,3,5}$, $V_{1,3,4,5}$, $V_{2,4,5}$, $V_{1,2,4,5}$, $V_{2,3,4,5}$ and~$V_{1,2,3,4,5}$ and note that they are all (possibly empty) independent sets by assumption (see also \figurename~\ref{fig:V13-indep}).
For two of these sets, if there is an $i \in \{1,2,3,4,5\}$ such that $i \notin S$ and $i \notin T$ then then~$V_S$ is complete to~$V_T$ by Claim~\ref{clm2:indeps-no-i-comp}.

\begin{figure}
\begin{center}
\begin{tikzpicture}[scale=0.65]
\coordinate (v13)    at (360*0/12+90:8) ;
\coordinate (v12345) at (360*1/12+90:8) ;
\coordinate (v1345)  at (360*2/12+90:8) ;
\coordinate (v2345)  at (360*3/12+90:8) ;
\coordinate (v1245)  at (360*4/12+90:8) ;
\coordinate (v1235)  at (360*5/12+90:8) ;
\coordinate (v1234)  at (360*6/12+90:8) ;
\coordinate (v245)   at (360*7/12+90:8) ;
\coordinate (v235)   at (360*8/12+90:8) ;
\coordinate (v124)   at (360*9/12+90:8) ;
\coordinate (v25)    at (360*10/12+90:8) ;
\coordinate (v24)    at (360*11/12+90:8) ;

\draw [fill=black] (v13)    circle (1.5pt) ;
\draw [fill=black] (v12345) circle (1.5pt) ;
\draw [fill=black] (v1345)  circle (1.5pt) ;
\draw [fill=black] (v2345)  circle (1.5pt) ;
\draw [fill=black] (v1245)  circle (1.5pt) ;
\draw [fill=black] (v1235)  circle (1.5pt) ;
\draw [fill=black] (v1234)  circle (1.5pt) ;
\draw [fill=black] (v245)   circle (1.5pt) ;
\draw [fill=black] (v235)   circle (1.5pt) ;
\draw [fill=black] (v124)   circle (1.5pt) ;
\draw [fill=black] (v25)    circle (1.5pt) ;
\draw [fill=black] (v24)    circle (1.5pt) ;

\draw (v13)    node [label={[label distance=-0pt]+90:$V_{1,3}$}] {};
\draw (v12345) node [label={[label distance=-0pt]+90:$V_{1,2,3,4,5}$}] {};
\draw (v1345)  node [label={[label distance=-0pt]+180:$V_{1,3,4,5}$}] {};
\draw (v2345)  node [label={[label distance=-0pt]+180:$V_{2,3,4,5}$}] {};
\draw (v1245)  node [label={[label distance=-0pt]+180:$V_{1,2,4,5}$}] {};
\draw (v1235)  node [label={[label distance=-0pt]+270:$V_{1,2,3,5}$}] {};
\draw (v1234)  node [label={[label distance=-0pt]+270:$V_{1,2,3,4}$}] {};
\draw (v245)   node [label={[label distance=-0pt]+270:$V_{2,4,5}$}] {};
\draw (v235)   node [label={[label distance=-0pt]+0:$V_{2,3,5}$}] {};
\draw (v124)   node [label={[label distance=-0pt]+0:$V_{1,2,4}$}] {};
\draw (v25)    node [label={[label distance=-0pt]+0:$V_{2,5}$}] {};
\draw (v24)    node [label={[label distance=-0pt]+90:$V_{2,4}$}] {};

\draw (v13) -- (v245);
\draw (v13) -- (v1245);
\draw (v13) -- (v2345);
\draw (v13) -- (v12345);

\draw[dashed] (v24) -- (v25);
\draw[dashed] (v24) -- (v124);
\draw[dashed] (v24) -- (v245);
\draw (v24) -- (v1235);
\draw (v24) -- (v1345);
\draw (v24) -- (v12345);

\draw[dashed] (v25) -- (v235);
\draw[dashed] (v25) -- (v245);
\draw (v25) -- (v1234);
\draw (v25) -- (v1345);
\draw (v25) -- (v12345);

\draw (v124) -- (v235);
\draw[dashed] (v124) -- (v1234);
\draw[dashed] (v124) -- (v1235);
\draw[dashed] (v124) -- (v1245);
\draw (v124) -- (v2345);
\draw (v124) -- (v1345);
\draw[dashed] (v124) -- (v12345);

\draw[dashed] (v235) -- (v1234);
\draw[dashed] (v235) -- (v1235);
\draw (v235) -- (v1245);
\draw[dashed] (v235) -- (v2345);
\draw (v235) -- (v1345);
\draw[dashed] (v235) -- (v12345);

\draw (v245) -- (v1234);
\draw (v245) -- (v1235);
\draw[dashed] (v245) -- (v1245);
\draw[dashed] (v245) -- (v2345);
\draw[dashed] (v245) -- (v1345);
\draw[dashed] (v245) -- (v12345);

\draw[dashed] (v1234) -- (v1235);
\draw[dashed] (v1234) -- (v1245);
\draw[dashed] (v1234) -- (v2345);
\draw[dashed] (v1234) -- (v1345);
\draw[dashed] (v1234) -- (v12345);

\draw[dashed] (v1235) -- (v1245);
\draw[dashed] (v1235) -- (v2345);
\draw[dashed] (v1235) -- (v1345);
\draw[dashed] (v1235) -- (v12345);

\draw[dashed] (v1245) -- (v2345);
\draw[dashed] (v1245) -- (v1345);
\draw[dashed] (v1245) -- (v12345);

\draw[dashed] (v2345) -- (v1345);
\draw[dashed] (v2345) -- (v12345);

\draw[dashed] (v1345) -- (v12345);
\end{tikzpicture}
\end{center}
\caption{\label{fig:V13-indep}The set of possible independent sets when~$V_{1,3}$ is an independent set.
Two sets are joined by a line if the edges between them form a co-matching.
Two sets are joined by a dashed line if at most one of them is large.
Two sets are not joined by a line if they are complete to each other.
These properties follow from Claims~\ref{clm2:indep-comatching}, \ref{clm2:indeps-no-i-comp}, \ref{clm2:V_S-V_T-contain12-empty}, and~\ref{clm2:V_S-V_T-not-contain13-empty}.}
\end{figure}
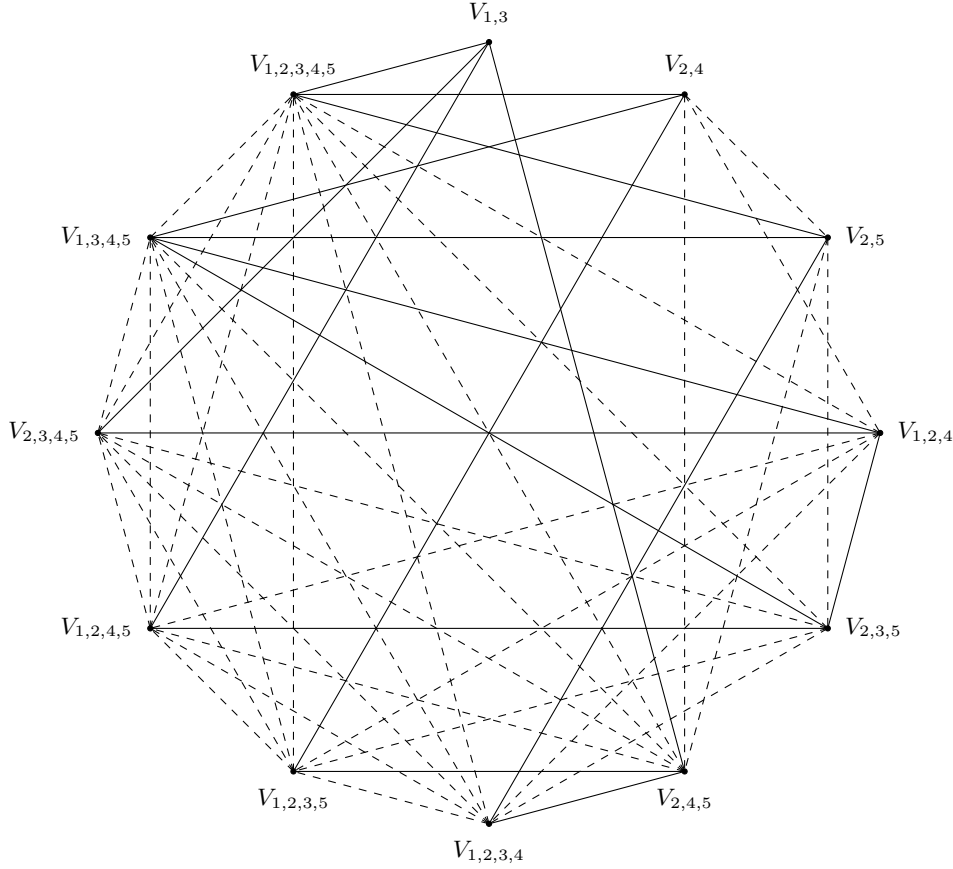

Since $\{4,5\} \subseteq \{1,2,3,4,5\}, \{2,3,4,5\}, \{1,2,4,5\}, \{2,4,5\}$, at most one of the sets $V_{1,2,3,4,5}$, $V_{2,3,4,5}$, $V_{1,2,4,5}$, and~$V_{2,4,5}$ is large by Claim~\ref{clm2:V_S-V_T-contain12-empty}.
We consider several cases.

\thmcase{$V_{1,2,3,4,5}$ is large.}
Since
\begin{itemize}
\item $\{1,2\} \subseteq \{1,2,4\}$,
\item $\{2,3\} \subseteq \{2,3,5\}, \{1,2,3,4\}, \{1,2,3,5\}$ and
\item $\{4,5\} \subseteq \{1,3,4,5\},\allowbreak \{2,4,5\},\allowbreak \{1,2,4,5\}, \{2,3,4,5\}, \{1,2,3,4,5\}$,
\end{itemize}
Claim~\ref{clm2:V_S-V_T-contain12-empty} implies that $V_{1,2,4}$, $V_{2,3,5}$, $V_{1,2,3,4}$, $V_{1,2,3,5}$, $V_{1,3,4,5}$, $V_{2,4,5}$, $V_{1,2,4,5}$, $V_{2,3,4,5}$ and~$V_{1,2,3,4,5}$ are empty.
Now $\{1,3\} \cap (\{2,4\} \cup \{2,5\})=\emptyset$, so Claim~\ref{clm2:V_S-V_T-not-contain13-empty} implies that~either~$V_{2,4}$ or~$V_{2,5}$ is empty.
By symmetry we may assume that~$V_{2,5}$ is empty.
This means that only the sets~$V_{1,3}$, $V_{1,2,3,4,5}$ and~$V_{2,4}$ are large.
By Lemma~\ref{l-victor}, it follows that~$G''$ has bounded clique-width.
This completes the case.

\begin{sloppypar}
\thmcase{$V_{2,3,4,5}$ or~$V_{1,2,4,5}$ is large.}
By symmetry, we may assume that~$V_{2,3,4,5}$ is large.
Since
\begin{itemize}
\item $\{2,3\} \subseteq \{2,3,5\}, \{1,2,3,4\},\allowbreak \{1,2,3,5\}$ and
\item $\{4,5\} \subseteq \{1,3,4,5\}, \{2,4,5\}, \{1,2,4,5\}, \{2,3,4,5\}, \{1,2,3,4,5\}$,
\end{itemize}
Claim~\ref{clm2:V_S-V_T-contain12-empty} implies that $V_{2,3,5}$, $V_{1,2,3,4}$, $V_{1,2,3,5}$, $V_{1,3,4,5}$, $V_{2,4,5}$, $V_{1,2,4,5}$, $V_{2,3,4,5}$ and~$V_{1,2,3,4,5}$ are empty.
This means that apart from~$V_{1,3}$ and~$V_{2,3,4,5}$, only the sets~$V_{2,4}$, $V_{2,5}$ and~$V_{1,2,4}$ can be large.
Now $4 \notin \{2,5\} \cup \{1,3\}$, $1 \notin \{2,5\} \cup \{2,3,4,5\}$, and $3 \notin \{2,5\} \cup \{2,4\} \cup \{1,2,4\}$, so by Claim~\ref{clm2:indeps-no-i-comp}, $V_{2,5}$ is complete to all the other large sets.
By Fact~\ref{fact:bip}, we may apply a bipartite complementation between~$V_{2,5}$ and $V_{1,3} \cup V_{2,3,4,5} \cup V_{2,4} \cup V_{1,2,4}$.
This will disconnect $G[V_{2,5}]$ from the rest of~$G''$.
Since~$V_{2,5}$ is an independent set, $G[V_{2,5}]$ has clique-width at most~$1$.
We may therefore assume that~$V_{2,5}$ is empty.
Since $\{3,5\} \cap (\{2,4\} \cup \{1,2,4\}) = \emptyset$, Claim~\ref{clm2:V_S-V_T-not-contain13-empty} implies that either~$V_{2,4}$ or~$V_{1,2,4}$ is empty.
This means that only at most three sets~$V_S$ are large: $V_{1,3}$, $V_{2,3,4,5}$ and either~$V_{2,4}$ or~$V_{1,2,4}$.
By Lemma~\ref{l-victor}, it follows that~$G''$ has bounded clique-width.
This completes the case.
\end{sloppypar}

\thmcase{$V_{2,4,5}$ is large.}
Since
\begin{itemize}
\item $\{4,5\} \subseteq \{1,3,4,5\}, \{2,4,5\}, \{1,2,4,5\}, \{2,3,4,5\}, \{1,2,3,4,5\}$,
\end{itemize}
Claim~\ref{clm2:V_S-V_T-contain12-empty} implies that $V_{1,3,4,5}$, $V_{2,4,5}$, $V_{1,2,4,5}$, $V_{2,3,4,5}$ and~$V_{1,2,3,4,5}$ are empty.
Since $\{1,3\} \cap (\{2,4\} \cup \{2,5\} \cup \{2,4,5\})=\emptyset$, Claim~\ref{clm2:V_S-V_T-not-contain13-empty} implies that~$V_{2,4}$ and~$V_{2,5}$ are empty.
This means that apart from~$V_{1,3}$ and~$V_{2,4,5}$, only the following sets may be large: $V_{1,2,4}$, $V_{2,3,5}$, $V_{1,2,3,4}$ and $V_{1,2,3,5}$.
Since $\{1,2\} \subseteq \{1,2,4\}, \{1,2,3,4\}, \{1,2,3,5\}$, Claim~\ref{clm2:V_S-V_T-contain12-empty} implies that at most one of $V_{1,2,4}$, $V_{1,2,3,4}$ and $V_{1,2,3,5}$ is large.
Since $\{2,3\} \subseteq \{2,3,5\}, \{1,2,3,4\}, \{1,2,3,5\}$, Claim~\ref{clm2:V_S-V_T-contain12-empty} implies that at most one of $V_{2,3,5}$, $V_{1,2,3,4}$ and $V_{1,2,3,5}$ is large.
Since $5 \notin \{1,2,4\} \cup \{1,3\}$, $3 \notin \{1,2,4\} \cup \{2,4,5\}$, $4 \notin \{2,3,5\} \cup \{1,3\}$, $1 \notin \{2,3,5\} \cup \{2,4,5\}$, Claim~\ref{clm2:indeps-no-i-comp} implies that~$V_{1,3}$ and~$V_{2,4,5}$ are both complete to both~$V_{1,2,4}$ and~$V_{2,3,5}$.
Therefore, if~$V_{1,2,4}$ or~$V_{2,3,5}$ are large, then~$V_{1,2,3,4}$ and~$V_{1,2,3,5}$ are empty and by Fact~\ref{fact:bip} we can apply a bipartite complementation between $V_{1,3} \cup V_{2,4,5}$ and $V_{1,2,4} \cup V_{2,3,5}$.
This will disconnected $G[V_{1,3} \cup V_{2,4,5}]$ from the rest of the graph.
By Claim~\ref{clm2:V13-V245-bdd-cw}, $G[V_{1,3} \cup V_{2,4,5}]$ has bounded clique-width.
We may therefore assume that~$V_{1,2,4}$ and~$V_{2,3,5}$ are empty.
This means that at most three sets~$V_S$ are large: $V_{1,3}$, $V_{2,4,5}$ and either~$V_{1,2,3,4}$ or~$V_{1,2,3,5}$.
By Lemma~\ref{l-victor}, it follows that~$G''$ has bounded clique-width.
This completes the case.

\thmcase{$V_{1,2,3,4,5}$, $V_{2,3,4,5}$, $V_{1,2,4,5}$, and~$V_{2,4,5}$ are empty.}
The only sets apart from~$V_{1,3}$ that can be large are~$V_{2,4}$, $V_{2,5}$, $V_{1,2,4}$, $V_{2,3,5}$, $V_{1,2,3,4}$, $V_{1,2,3,5}$ and $V_{1,3,4,5}$.
Since $4 \notin \{1,3\} \cup \{2,5\}, \{2,3,5\}, \{1,2,3,5\}$, $5 \notin \{1,3\} \cup \{2,4\} \cup \{1,2,4\} \cup \{1,2,3,4\}$ and $2 \notin \{1,3\} \cup \{1,3,4,5\}$, Claim~\ref{clm2:indeps-no-i-comp} implies that~$V_{1,3}$ is complete to all the other large sets.
Applying a bipartite complementation between~$V_{1,3}$ and $V_{2,4} \cup V_{2,5} \cup V_{1,2,4} \cup V_{2,3,5} \cup V_{1,2,3,4} \cup V_{1,2,3,5} \cup V_{1,3,4,5}$ disconnects~$G[V_{1,3}]$ from the rest of~$G''$.
Since~$V_{1,3}$ is an independent set, $G[V_{1,3}]$ has clique-width at most~$1$.
Therefore, by Fact~\ref{fact:bip}, we may delete~$V_{1,3}$ from the graph.
This completes the case

\medskip
\noindent
Since one of the above cases must hold, this completes the proof of the claim when $i=1$.
The claim follows by symmetry.

\clm{For $i \in \{1,2,3,4,5\}$, we may assume~$V_{i,i+1,i+2}$ is empty.}
This follows from Claim~\ref{clm2:V13-empty} by casting to the complement (see also Table~\ref{tbl:WT-VS-corr}).

\medskip
\noindent
We are now ready to complete the proof of the lemma.
First, by Fact~\ref{fact:del-vert}, we may delete the five vertices of the original cycle~$C$.
Let $G' =G[\bigcup V_S \; | \; S \mbox{ is a clique}]$ and let $G'' =G[\bigcup V_S \; | \; S \mbox{ is an independent set}]$.
By Claim~\ref{clm2:clique-indep-trivial}, if~$V_S$ is a clique and~$V_T$ is an independent set, then~$V_S$ is either complete or anti-complete to~$V_T$.
If~$V_S$ is complete to~$V_T$, by Fact~\ref{fact:bip}, we may apply a bipartite complementation between these sets.
Doing so for every pair of a clique~$V_S$ and an independent set~$V_T$ that are complete to each other, we disconnect~$G'$ from~$G''$.
By Fact~\ref{fact:bip} it is sufficient to show that~$G'$ and~$G''$ have bounded clique-width.
In fact, it is sufficient to show that~$G'$ has bounded clique-width, since then we can obtain the same result for~$G''$ by casting to the complement (see also Table~\ref{tbl:WT-VS-corr}) and applying Fact~\ref{fact:comp}.
In the remainder of the proof, we show that~$G'$ has bounded clique-width.

Note that only the following sets~$V_S$ can remain:
$V_\emptyset$, $V_1$, $V_2$, $V_3$, $V_4$, $V_5$, $V_{1,2}$, $V_{2,3}$, $V_{3,4}$, $V_{4,5}$ and~$V_{1,5}$.
Note that all of these sets are cliques by Claim~\ref{clm2:V0V1V12-clique} and by Claim~\ref{clm2:clique-matching} the edges between any two of these sets form a matching.

If $V_\emptyset$ is large then since $\{1,3\} \cap (\emptyset \cup \{4\} \cup \{4,5\}) = \emptyset$ Claim~\ref{clm2:V_S-V_T-not-contain13-empty} implies that~$V_4$ and~$V_{4,5}$ are empty.
Similarly, every set apart from~$V_\emptyset$ is empty, so~$G'$ is a complete graph and therefore has clique-width~$2$.
We may therefore assume that~$V_\emptyset$ is empty.

Suppose that~$V_1$ is large.
Since $\{2,5\} \cap (\{1\} \cup \{3\} \cup \{4\} \cup \{3,4\})=\emptyset$,
$\{2,4\} \cap (\{1\} \cup \{5\} \cup \{1,5\}) = \emptyset$ and
$\{3,5\} \cap (\{1\} \cup \{2\} \cup \{1,2\}) = \emptyset$, Claim~\ref{clm2:V_S-V_T-not-contain13-empty} implies that $V_3$, $V_4$, $V_{3,4}$, $V_5$, $V_{1,5}$, $V_2$ and~$V_{1,2}$ are empty.
Therefore only~$V_1$, $V_{2,3}$ and~$V_{4,5}$ can be large.
Hence by Lemma \ref{l-victor} the graph~$G'$ has bounded clique-width.
We may therefore assume that~$V_1$ is empty.
By symmetry, we may assume that~$V_i$ is empty for all $i \in \{1,2,3,4,5\}$.

Now by Claim~\ref{clm2:cliques-i-anti} if $j \in \{i+1,i-1\}$ then~$V_{i,i+1}$ is anti-complete to~$V_{j,j+1}$.
For every $i \in \{1,2,3,4,5\}$, by Claim~\ref{clm2:clique-matching} the edges between~$V_{i,i+1}$ and~$V_{i+2,i+3}$ form a matching.
By Fact~\ref{fact:comp}, we may apply a complementation to each set~$V_{i,i+1}$.
We obtain a graph of maximum degree at most~$2$, which therefore has clique-width at most~$4$ by Lemma~\ref{lem:atmost-2}.
This completes the proof of the lemma.\qedllncs
\end{proof}

\begin{lemma}\label{lem:2P1+P_3-co-2P_1+P_3-free-c6-non-free}
The class of $(2P_1+\nobreak P_3,\allowbreak \overline{2P_1+P_3})$-free graphs containing an induced~$C_6$ has bounded clique-width.
\end{lemma}

\begin{proof}
\setcounter{ctrclaim}{0}
Suppose~$G$ is a $(2P_1+\nobreak P_3,\allowbreak \overline{2P_1+P_3})$-free graph containing an induced cycle~$C$ on six vertices, say $v_1,\ldots,v_6$ in order.
By Lemma~\ref{lem:2P1+P_3-co-2P_1+P_3-free-c5-non-free}, we may assume that~$G$ is $C_5$-free.
For $S \subseteq \{1,\ldots,6\}$, let~$V_S$ be the set of vertices $x \in V(G) \setminus V(C)$ such that $N(x)\cap V(C)=\{v_i \;|\; i \in S\}$.
We say that a set~$V_S$ is {\em large} if it contains at least two vertices, otherwise it is {\em small}.

To ease notation, in the following claims, subscripts on vertex sets should be interpreted modulo~$6$
and whenever possible we will write~$V_i$ instead of~$V_{\{i\}}$ and~$V_{i,j}$ instead of~$V_{\{i,j\}}$ and so on.

\clm{\label{clm4:V_S-large-or-empty}We may assume that for $S \subseteq \{1,\ldots,6\}$, the set~$V_S$ is either large or empty.}
If a set~$V_S$ is small, but not empty, then by Fact~\ref{fact:del-vert}, we may delete all vertices of this set.
If later in our proof we delete vertices in some set~$V_S$ and in doing so make a large set~$V_S$ become small, we may immediately delete the remaining vertices in~$V_S$.
The above arguments involve deleting a total of at most~$2^6$ vertices.
By Fact~\ref{fact:del-vert}, the claim follows.

\clm{\label{clm4:V0V1-empty}For $S \subseteq \{1,\ldots,6\}$ if $|S| \leq 1$ then $V_S=\emptyset$.}
If $x \in V_\emptyset \cup V_{2}$ then $G[x,v_1,v_3,v_4,v_5]$ is a $2P_1+\nobreak P_3$, a contradiction.
The claim follows by symmetry.

\clm{\label{clm4:V12-clique}For $i \in \{1,\ldots,6\}$, $V_{i,i+1}$ is a clique.}
Suppose that $x,x' \in V_{1,2}$ are non-adjacent.
Then $G[x,x',v_3,v_4,v_5]$ is a $2P_1+\nobreak P_3$, a contradiction.
The claim follows by symmetry.

\clm{\label{clm4:V123-clique}For $i \in \{1,\ldots,6\}$, $V_{i,i+1,i+2}$ is a clique.}
Suppose that $x,x' \in V_{1,2,3}$ are non-adjacent.
Then $G[x,x',v_4,v_5,v_6]$ is a $2P_1+\nobreak P_3$, a contradiction.
The claim follows by symmetry.

\clm{\label{clm4:V13-empty}For $i \in \{1,\ldots,6\}$, $V_{i,i+2}$ is empty.}
If $x \in V_{1,3}$ then $G[x,v_2,v_4,v_5,v_6]$ is a $2P_1+\nobreak P_3$, a contradiction.
The claim follows by symmetry.

\clm{\label{clm4:V14V124V1234-empty}For $i \in \{1,\ldots,6\}$, $V_{i,i+3} \cup V_{i,i+1,i+3} \cup V_{i,i+2,i+3} \cup V_{i,i+1,i+2,i+3}$ is empty.}
If $x \in V_{1,4} \cup V_{1,2,4} \cup V_{1,3,4} \cup V_{1,2,3,4}$ then $G[x,v_4,v_5,v_6,v_1]$ is a $C_5$, a contradiction.
The claim follows by symmetry.

\clm{\label{clm4:V135-P3-free}For $i \in \{1,2\}$, $G[V_{i,i+2,i+4}]$ is $P_3$-free.}
Suppose $x,x',x'' \in V_{1,3,5}$ are such that $G[x,x',x'']$ is a~$P_3$.
Then $G[v_2,v_4,x,x',x'']$ is a $2P_1+\nobreak P_3$, a contradiction.
The claim follows by symmetry.

\clm{\label{clm4:V1235}For $i \in \{1,\ldots,6\}$, $V_{i,i+1,i+2,i+4}$ is empty.}
Suppose for contradiction that~$V_{1,2,3,5}$ is not empty.
By Claim~\ref{clm4:V_S-large-or-empty}, there are two vertices $x,x' \in V_{1,2,3,5}$.
If~$x$ is adjacent to~$x'$ then $G[x,x',v_1,v_5,v_2]$ is a $\overline{2P_1+ P_3}$.
If~$x$ is non-adjacent to~$x'$ then $G[v_4,v_6,x,v_2,x']$ is a $2P_1+\nobreak P_3$.
This contradiction implies that~$V_{1,2,3,5}$ is empty.
The claim follows by symmetry.

\begin{sloppypar}
\clm{\label{clm4:V1245V12345V12456V123456-indep}For $i \in \{1,\ldots,6\}$, $V_{i,i+1,i+3,i+4} \cup V_{i,i+1,i+2,i+3,i+4} \cup V_{i,i+1,i+3,i+4,i+5} \cup V_{i,i+1,i+2,i+3,i+4,i+5}$ is an independent set.}
Suppose $x,x' \in V_{1,2,4,5} \cup V_{1,2,3,4,5} \cup V_{1,2,4,5,6} \cup V_{1,2,3,4,5,6}$ are adjacent.
Then $G[x,x',v_1,v_4,v_2]$ is a $\overline{2P_1+ P_3}$, a contradiction.
The claim follows by symmetry.
\end{sloppypar}

\medskip
\noindent
By Claims~\ref{clm4:V_S-large-or-empty}--\ref{clm4:V1245V12345V12456V123456-indep}, only the following sets can be non-empty:
\begin{itemize}
\item $V_{i,i+1}$ for $i \in \{1,\ldots,6\}$, which are cliques,
\item $V_{i,i+1,i+2}$ for $i \in \{1,\ldots,6\}$, which are cliques,
\item $V_{i,i+2,i+4}$ for $i \in \{1,2\}$, which induce $P_3$-free graphs in~$G$,
\item $V_{i,i+1,i+3,i+4}$ for $i \in \{1,2,3\}$, which are independent sets,
\item $V_{i,i+1,i+2,i+3,i+4}$ for $i \in \{1,\ldots,6\}$, which are independent sets and
\item $V_{1,2,3,4,5,6}$, which is an independent set.
\end{itemize}
In the remainder of the proof, we will prove a number of claims.
First, we will show that we can remove sets of the form~$V_{i,i+1}$ (Claim~\ref{clm4:V12-empty}), and of the form~$V_{i,i+1,i+2}$ (Claim~\ref{clm4:V123-empty}) from the graph.
Then we will show that for $T \subseteq \{1,\ldots,6\}$, with $|T| \geq 4$ we can remove~$V_T$ from the graph (Claims~\ref{clm4:V1245-empty} and~\ref{clm4:V12345-V123456-empty}).
This will leave only the sets~$V_{1,3,5}$ and~$V_{2,4,6}$ and the last stage will be to deal with these sets.

\clm{\label{clm4:clique-indep-trivial}We may assume that for distinct $S,T \subseteq \{1,\ldots,6\}$ if~$V_S$ is an independent set and~$V_T$ is a clique then~$V_S$ is either complete or anti-complete to~$V_T$.}
Let $S,T \subseteq \{1,\ldots,6\}$ be distinct.
If~$V_S$ is an independent set and~$V_T$ is a clique, then by Lemma~\ref{lem:comp-anti}, we may delete at most three vertices from each of these sets, such that in the resulting graph, $V_S$ will be complete or anti-complete to~$V_T$.
Doing this for every pair of an independent set~$V_S$ and a clique~$V_T$ we delete at most $\binom{2^6}{2}\times 2 \times 3$ vertices from~$G$.
The claim follows by Fact~\ref{fact:del-vert}.

\begin{sloppypar}
\clm{\label{clm4:V12-or-V135-empty}For $i \in \{1,2\}$ and $j \in \{1,\ldots,6\}$, if~$V_{i,i+2,i+4}$ is large then~$V_{j,j+1}$ is empty.}
Suppose, for contradiction, that $x \in V_{1,3,5}$ and $y \in V_{1,2}$.
Then $G[v_4,v_6,v_2,y,x]$ or $G[y,v_6,x,v_3,v_4]$ is a $2P_1+\nobreak P_3$ if~$x$ is adjacent or non-adjacent to~$y$, respectively.
This is a contradiction.
The claim follows by symmetry.
\end{sloppypar}

\clm{\label{clm4:V12-V45-empty}For $i \in \{1,2,3\}$, either~$V_{i,i+1}$ or~$V_{i+3,i+4}$ is empty.}
Suppose $x \in V_{1,2}$ and $y \in V_{4,5}$.
If~$x$ is adjacent to~$y$ then $G[x,v_2,v_3,v_4,y]$ is a~$C_5$.
If~$x$ is non-adjacent to~$y$ then $G[y,v_3,x,v_1,v_6]$ is a $2P_1+\nobreak P_3$, a contradiction.
The claim follows by symmetry.

\clm{\label{clm4:V12-Vii+1i+2-anti}For $i,j \in \{1,\ldots,6\}$, $V_{i,i+1,i+2}$ is anti-complete to~$V_{j,j+1}$.}
Let~$i=1$, $j \in \{2,3,4\}$ (the other cases follow by symmetry).
If~$V_{1,2,3}$ and~$V_{j,j+1}$ are not empty then by Claim~\ref{clm4:V_S-large-or-empty} they must be large.
Suppose~$x,x' \in V_{1,2,3}$ and $y \in V_{j,j+1}$ with~$x$ adjacent to~$y$.
By Claim~\ref{clm4:V123-clique}, $x$ must be adjacent to~$x'$.
If $j=2$ then $G[x,v_2,v_3,v_1,y]$ is a $\overline{2P_1+P_3}$.
If $j=3$ then $G[x,x',v_1,y,v_2]$ or $G[v_3,x,x',y,v_2]$ is a $\overline{2P_1+P_3}$ if~$x'$ is adjacent or non-adjacent to~$y$, respectively.
If $j=4$ then $G[y,v_5,v_6,v_1,x]$ is a~$C_5$.
This is a contradiction.
The claim follows by symmetry.

\clm{\label{clm4:V12-empty}We may assume that~$V_{i,i+1}$ is empty for all $i \in \{1,\ldots,6\}$.}
Let~$G'$ be the graph induced by the sets of the form~$V_{i,i+1}$.
We will show how to disconnect~$G'$ from the rest of~$G$ and then show that~$G'$ has bounded clique-width.

We may assume that at least one set of the form~$V_{i,i+1}$ is non-empty (in which case it must be large by Claim~\ref{clm4:V_S-large-or-empty}), otherwise we are done.
By Fact~\ref{fact:bip}, for each $i \in \{1,\ldots,6\}$ we may apply a bipartite complementation between~$V_{i,i+1}$ and $\{v_i,v_{i+1}\}$.
By Claim~\ref{clm4:V12-Vii+1i+2-anti}, for $i,j \in \{1,\ldots,6\}$ there are no edges between~$V_{i,i+1}$ and~$V_{j,j+1,j+2}$.
By Claim~\ref{clm4:V12-or-V135-empty}, $V_{1,3,5}$ and~$V_{2,4,6}$ are empty.
By Claim~\ref{clm4:V12-clique}, all sets of the form~$V_{i,i+1}$ are cliques.
By Claim~\ref{clm4:V1245V12345V12456V123456-indep}, if~$V_T$ is large with $|T|\geq 4$ then~$V_T$ is an independent set.
Therefore, by Claim~\ref{clm4:clique-indep-trivial}, for all $i \in \{1,\ldots,6\}$ and all $T \subseteq \{1,\ldots,6\}$ with $|T|\geq 4$, $V_{i,i+1}$ is either complete or anti-complete to~$V_T$.
If~$V_{i,i+1}$ is complete to~$V_T$ then by Fact~\ref{fact:bip} we may apply a bipartite complementation between these two sets.
This removes all edges from vertices in~$G'$ to vertices outside~$G'$.

It remains to show that~$G'$ has bounded clique-width.
By Claim~\ref{clm4:V12-V45-empty}, $V_{1,2}$ or~$V_{4,5}$ is empty, $V_{2,3}$ or~$V_{5,6}$ is empty and~$V_{3,4}$ or~$V_{1,6}$ is empty.
This means that at most three sets of the form~$V_{i,i+1}$ can be large.
By Claim~\ref{clm4:V12-clique}, every set of the form~$V_{i,i+1}$ induces a clique in~$G$ (and therefore in~$G'$).
By Lemma~\ref{l-victor}, it follows that~$G'$ has bounded clique-width.

We conclude that we can remove all sets of the form~$V_{i,i+1}$ from~$G$, that is, we may assume that these sets are empty.
This completes the proof of the claim.

\clm{\label{clm4:V135-V234-comp}For $i \in \{1,2\}$ and $j \in \{i+1,i+3,i+5\}$, $V_{i,i+2,i+4}$ is complete to~$V_{j,j+1,j+2}$.}
If $x \in V_{1,3,5}$ is non-adjacent to $y \in V_{2,3,4}$ then $G[x,v_6,v_2,y,v_4]$ is a $2P_1+\nobreak P_3$, a contradiction.
The claim follows by symmetry.

\begin{sloppypar}
\clm{\label{clm4:V135-V123-anti}For $i \in \{1,2\}$ and $j \in \{i,i+2,i+4\}$, $V_{i,i+2,i+4}$ is anti-complete to~$V_{j,j+1,j+2}$.}
If $x \in V_{1,3,5}$ is adjacent to $y \in V_{1,2,3}$ then $G[v_4,v_6,v_2,y,x]$ is a $2P_1+\nobreak P_3$, a contradiction.
The claim follows by symmetry.
\end{sloppypar}

\clm{\label{clm4:V123-or-V234-empty}For $i \in \{1,\ldots,6\}$ either~$V_{i,i+1,i+2}$ or~$V_{i+1,i+2,i+3}$ is empty.}
Suppose that~$V_{1,2,3}$ and~$V_{2,3,4}$ are both non-empty.
Then by Claim~\ref{clm4:V_S-large-or-empty} they must both be large and by Claim~\ref{clm4:V123-clique}, they must both be cliques.
If $x \in V_{1,2,3}$ is adjacent to $y \in V_{2,3,4}$ then $G[x,v_2,y,v_1,v_3]$ is a $\overline{2P_1+P_3}$, a contradiction.
Therefore~$V_{1,2,3}$ is anti-complete to~$V_{2,3,4}$.
If $x \in V_{1,2,3}$ and $y,y' \in V_{2,3,4}$ then $G[v_2,v_3,y,x,y']$ is a $\overline{2P_1+P_3}$, a contradiction.
The claim follows by symmetry.

\clm{\label{clm4:V123-empty}We may assume that~$V_{i,i+1,i+2}$ is empty for all $i \in \{1,\ldots,6\}$.}
Let~$G'$ be the graph induced by the sets of the form~$V_{i,i+1,i+2}$.
We will show how to disconnect~$G'$ from the rest of~$G$ and then show that~$G'$ has bounded clique-width.

We may assume that at least one set of the form~$V_{i,i+1,i+2}$ is non-empty (in which case it must be large by Claim~\ref{clm4:V_S-large-or-empty}), otherwise we are done.
By Fact~\ref{fact:bip}, for each $i \in \{1,\ldots,6\}$ we may apply a bipartite complementation between~$V_{i,i+1,i+2}$ and $\{v_i,v_{i+1},v_{i+2}\}$.
By Claims~\ref{clm4:V135-V234-comp} and~\ref{clm4:V135-V123-anti} for $i \in \{1,\ldots,6\}$ and~$j \in \{1,2\}$, $V_{i,i+1,i+2}$ is either complete or anti-complete to~$V_{j,j+2,j+4}$; if it is complete then by Fact~\ref{fact:bip}, we may apply a bipartite complementation between these two sets.
By Claim~\ref{clm4:V123-clique}, all sets of the form~$V_{i,i+1,i+2}$ are cliques.
By Claim~\ref{clm4:V1245V12345V12456V123456-indep}, if~$V_T$ is large with $|T|\geq 4$ then~$V_T$ is an independent set.
Therefore, for all $i \in \{1,\ldots,6\}$ and all $T \subseteq \{1,\ldots,6\}$ with $|T|\geq 4$, $V_{i,i+1,i+2}$ is either complete or anti-complete to~$V_T$.
If~$V_{i,i+1,i+2}$ is complete to~$V_T$ then by Fact~\ref{fact:bip} we may apply a bipartite complementation between these two sets.
This removes all edges from vertices in~$G'$ to vertices outside~$G'$.

It remains to show that~$G'$ has bounded clique-width.
By Claim~\ref{clm4:V123-or-V234-empty}, $V_{1,2,3}$ or~$V_{2,3,4}$ is empty, $V_{3,4,5}$ or~$V_{4,5,6}$ is empty and $V_{1,5,6}$ or~$V_{1,2,6}$ is empty.
This means that at most three sets of the form~$V_{i,i+1,i+2}$ can be large.
By Claim~\ref{clm4:V123-clique}, every set of the form~$V_{i,i+1,i+2}$ induces a clique in~$G$ (and therefore in~$G'$).
By Lemma~\ref{l-victor}, it follows that~$G'$ has bounded clique-width.

We conclude that we can remove all sets of the form~$V_{i,i+1,i+2}$ from~$G$, that is, we may assume that these sets are empty.
This completes the proof of the claim.

\clm{\label{clm4:leq3-V1245-or-V12345-orV123456}Suppose $S,T \subseteq \{1,\ldots,6\}$ are distinct with $|S|\geq 4$ and $|T|\geq 5$.
Then~$V_S$ or~$V_T$ is empty.}
Suppose the claim is false for some~$S$ and~$T$.
By Claim~\ref{clm4:V_S-large-or-empty}, we may assume~$V_S$ and~$V_T$ are large.
Without loss of generality, we may assume that $|T| \geq |S|$.
Without loss of generality, we may assume that $\{1,2,4,5\} \subseteq S$.
The vertices of~$V_T$ are non-adjacent to at most one vertex of the original cycle~$C$, so without loss of generality, we may assume that $\{1,2,4\} \subseteq S \cap T$.
If $x,x' \in V_S \cup V_T$ are adjacent then $G[x,x',v_1,v_4,v_2]$ is a $\overline{2P_1+P_3}$, so $V_S \cup V_T$ is an independent set.
Suppose $x,x \in V_T$, $y,y' \in V_S$ and $i \in T \setminus S$ (which exists since $|T| \geq |S|$ and $S \neq T$).
Then $G[y,y',x,v_i,x']$ is a $2P_1+\nobreak P_3$, a contradiction.
The claim follows by symmetry.

\medskip
\noindent
Claim~\ref{clm4:leq3-V1245-or-V12345-orV123456} implies that, in addition to the sets~$V_{1,3,5}$ and~$V_{2,4,6}$, which may be large, exactly one of the following must hold:
\begin{itemize}
\item Exactly one, two or three sets~$V_S$ with $|S|=4$ are large and all sets~$V_T$ with $|T| \in \{5,6\}$ are empty.
\item Exactly one set~$V_S$ with $|S|=5$ is large and all sets~$V_T$ with $|T| \in \{4,6\}$ are empty.
\item $V_{1,2,3,4,5,6}$ is large and all sets~$V_T$ with $|T| \in \{4,5\}$ are empty.
\item All sets~$V_T$ with $|T| \geq 4$ are empty.
\end{itemize}

\begin{sloppypar}
\clm{\label{clm4:V135-V1234V12345-V135-indep-comp}Let $i \in \{1,2\}$ and $T \subseteq \{1,\ldots,6\}$ with $|T| \geq 4$ such that there is a $j \in \{1,\ldots,6\}$, with $j \notin \{i,i+2,i+4\}$, $j \notin T$.
If~$V_{i,i+2,i+4}$ and~$V_T$ are large then~$V_{i,i+2,i+4}$ is an independent set that is complete to~$V_T$.}
By symmetry, we need only prove the claim in the case where $i=1$ and $T \in \{\{1,2,4,5\},\{1,2,3,4,5\}\}$, in which case $j=6$.
Suppose $x \in V_{1,3,5}$ is non-adjacent to $y \in V_T$.
Then $G[x,v_6,v_2,y,v_4]$ is a $2P_1+\nobreak P_3$, a contradiction.
Therefore~$V_{1,3,5}$ is complete to~$V_T$.
If $x,x' \in V_{1,3,5}$ with~$x$ adjacent to~$x'$ and $y \in V_T$ then $G[v_1,y,x,v_2,x']$ is a $\overline{2P_1+P_3}$, a contradiction.
Therefore~$V_{1,3,5}$ is an independent set.
The claim follows by symmetry.
\end{sloppypar}

\clm{\label{clm4:V135-V12346-or-V123456-V135-clique-plus-indep}For $i \in \{1,2\}$ and $T \subseteq \{1,\ldots,6\}$ with $|T| \geq 4$ and $\{i,i+2,i+4\} \cup T = \{1,\ldots,6\}$.
If~$V_{i,i+2,i+4}$ and~$V_T$ are large then we may assume~$V_{i,i+2,i+4}$ is the disjoint union of a (possibly empty) clique and a (possibly empty) independent set.}
Suppose $i=1$ and $T \in \{\{1,2,3,4,6\},\{1,2,3,4,5,6\}\}$.
Suppose~$V_{1,3,5}$ and~$V_T$ are large.
If $x \in V_T$ and $y,y' \in V_{1,3,5}$ with $x,y,y'$ pairwise non-adjacent then $G[y,y',v_2,x,v_4]$ is a $2P_1+\nobreak P_3$, a contradiction.
Therefore every vertex of~$V_T$ is complete to all but at most one component of~$G[V_{1,3,5}]$.
If $x \in V_T$ and $y,y' \in V_{1,3,5}$ with $x,y,y'$ pairwise adjacent then $G[x,v_1,y,v_2,y']$ is a $\overline{2P_1+P_3}$, a contradiction.
Therefore every vertex of~$V_T$ has at most one neighbour in each component of~$G[V_{1,3,5}]$.
It follows that~$G[V_{1,3,5}]$ has at most one non-trivial component.
The claim follows by symmetry.

\clm{\label{clm4:V1245-empty}For $i \in \{1,2,3\}$, we may assume that~$V_{i,i+1,i+3,i+4}$ is empty.}
Suppose that~$V_{i,i+1,i+3,i+4}$ is not empty for some $i \in \{1,2,3\}$.
Then by Claim~\ref{clm4:V_S-large-or-empty}, this set must be large.
By Claim~\ref{clm4:leq3-V1245-or-V12345-orV123456}, in this case only the following sets can be large:
$V_{1,3,5}$, $V_{2,4,6}$, $V_{1,2,4,5}$, $V_{2,3,5,6}$ and~$V_{3,4,6,1}$.
By Claim~\ref{clm4:V135-V1234V12345-V135-indep-comp}, $V_{1,3,5}$ and~$V_{2,4,6}$ are complete to $V_{1,2,4,5}$, $V_{2,3,5,6}$ and~$V_{3,4,6,1}$.
By Fact~\ref{fact:bip}, we may apply a bipartite complementation between $V_{1,3,5} \cup V_{2,4,6}$ and $V_{1,2,4,5} \cup V_{2,3,5,6} \cup V_{3,4,6,1}$.
By Fact~\ref{fact:bip}, we can also apply bipartite complementations between~$\{v_1,v_2,v_4,v_5\}$ and~$V_{1,2,4,5}$, between~$\{v_2,v_3,v_5,v_6\}$ and~$V_{2,3,5,6}$ and between~$\{v_3,v_4,v_6,v_1\}$ and~$V_{3,4,6,1}$.
This disconnects $G'=G[V_{1,2,4,5} \cup V_{2,3,5,6} \cup V_{3,4,6,1}]$ from the rest of the graph.

It remains to show that~$G'$ has bounded clique-width.
By Claim~\ref{clm4:V1245V12345V12456V123456-indep}, $V_{1,2,4,5}$, $V_{2,3,5,6}$ and~$V_{3,4,6,1}$ are independent sets.
By Lemma~\ref{l-victor}, it follows that~$G'$ has bounded clique-width.
We may therefore assume that $V_{1,2,4,5} \cup V_{2,3,5,6} \cup V_{3,4,6,1} = \emptyset$.
The claim follows.

\clm{\label{clm4:V12345-V123456-empty}For $T \subseteq \{1,\ldots,6\}$ with $|T| \geq 5$ we may assume that~$V_T$ is empty.}
Suppose there is a $T \subseteq \{1,\ldots,6\}$ with $|T| \geq 5$ such that~$V_T$ is large.
By Claim~\ref{clm4:leq3-V1245-or-V12345-orV123456}, only the following sets can be large: $V_{1,3,5}$, $V_{2,4,6}$ and~$V_T$.
By Claims~\ref{clm4:V135-V1234V12345-V135-indep-comp} and~\ref{clm4:V135-V12346-or-V123456-V135-clique-plus-indep}, each of~$V_{1,3,5}$ and~$V_{2,4,6}$ is the union of a (possibly empty) clique and a (possibly empty) independent set.
By Claim~\ref{clm4:V1245V12345V12456V123456-indep}, $V_T$ is an independent set.
By Fact~\ref{fact:del-vert}, we may delete the vertices of the original cycle~$C$.
We obtain a graph whose vertex set can be partitioned into at most two cliques and at most three independent sets, so~$G$ has bounded clique-width by Lemma~\ref{l-victor}.
The claim follows.

\medskip
\noindent
Only two sets~$V_S$ remain that may be non-empty, namely~$V_{1,3,5}$ and~$V_{2,4,6}$.
If one of these sets is empty, then by Fact~\ref{fact:del-vert} we may delete the vertices of the original cycle~$C$.
Claim~\ref{clm4:V135-P3-free} implies that the resulting graph is a disjoint union of cliques, and so has clique-width at most~$2$.
We may therefore assume that both~$V_{1,3,5}$ and~$V_{2,4,6}$ are non-empty, in which case they are both large by Claim~\ref{clm4:V_S-large-or-empty}.

Suppose $x \in V_{1,3,5}$ and $y,y' \in V_{2,4,6}$ and these three vertices are pairwise non-adjacent.
Then $G[y,y',v_1,x,v_3]$ is a $2P_1+\nobreak P_3$, a contradiction.
Therefore each vertex of~$V_{1,3,5}$ is complete to all but at most one component of~$G[V_{2,4,6}]$.
Similarly, each vertex of~$V_{2,4,6}$ is complete to all but at most one component of~$G[V_{1,3,5}]$.

First consider the case where~$G[V_{2,4,6}]$ contains at least three non-trivial components.
Every vertex in~$V_{1,3,5}$ must be complete to at least two of these components.
If $x,x' \in V_{1,3,5}$ are adjacent then they must both be complete to a common non-trivial component~$C'$ of~$G[V_{2,4,6}]$.
Let $y,y' \in V(C')$.
Then $G[x,x',y,v_1,y']$ is a $\overline{2P_1+ P_3}$, a contradiction.
It follows that~$V_{1,3,5}$ is an independent set.
Note that this implies that every vertex of~$V_{2,4,6}$ has at most one non-neighbour in~$V_{1,3,5}$.
By Fact~\ref{fact:bip}, we may apply a bipartite complementation between~$V_{1,3,5}$ and~$V_{2,4,6}$.
Let~$G'$ be the resulting graph.
In~$G'$, every vertex in~$V_{1,3,5}$ has neighbours in at most one component of~$G[V_{2,4,6}]$ and each vertex of~$V_{2,4,6}$ has at most one neighbour in~$V_{1,3,5}$.
This means that every component of~$G[V_{2,4,6}]$ lies in a different component of~$G'$.
It suffices to show that the components of~$G'$ have bounded clique-width.
Let~$C''$ be such a component of~$G'$.
By Fact~\ref{fact:comp}, we may apply a complementation to $V_{2,4,6} \cap V(C'')$.
We obtain a disjoint union of stars (some of which may be isolated vertices).
Since stars have clique-width at most~$2$, we have shown that in this case~$G$ has bounded clique-width.

We may therefore assume that~$V_{2,4,6}$ contains at most two non-trivial components.
By symmetry, we may also assume that~$V_{1,3,5}$ contains at most two non-trivial components.
This means that each of these sets consists of the disjoint union of at most two cliques and at most one independent set.
Let~$K^1$ and~$K^2$ be the two cliques and~$I^1$ be the independent set in~$V_{1,3,5}$.
Let~$K^3$ and~$K^4$ be the two cliques and~$I^2$ be the independent set in~$V_{2,4,6}$.
(We allow the case where some of the sets~$K^i$ or~$I^j$ are empty.)
Also note that in~$G$, $K^1$ is anti-complete to~$K^2$ and~$K^3$ is anti-complete to~$K^4$.
By Fact~\ref{fact:del-vert} and Lemma~\ref{lem:comp-anti}, we may assume that each clique~$K^i$ is either complete or anti-complete to each independent set~$I^j$.
If a clique~$K^i$ is complete to~$I^j$ then by Fact~\ref{fact:bip} we may apply a bipartite complementation between these sets.
This removes all edges between $K^1 \cup K^2 \cup K^3 \cup K^4$ and $I^1 \cup I^2$.
Now by Fact~\ref{fact:del-vert} and Lemma~\ref{lem:matching-comatching} we may assume that the edges between each pair of~$K^1$, $K^2$, $K^3$ and~$K^4$ and the edges between~$I^1$ and~$I^2$ either form a matching or a co-matching.
If the edges between two such sets form a co-matching, by Fact~\ref{fact:bip} we may apply a bipartite complementation between these sets.
Finally, by Fact~\ref{fact:comp}, we may complement each clique~$K^i$.
Let~$G''$ be the resulting graph and note that in this graph it is still the case that $K^1 \cup K^2 \cup K^3 \cup K^4$ is anti-complete to $I^1 \cup I^2$.
In~$G''$ the edges between~$I^1$ and~$I^2$ form a matching, so $G''[I^1 \cup I^2]$ has maximum degree at most~$1$ and thus clique-width at most~$2$.
In~$G''$ the sets~$K^1$, $K^2$, $K^3$ and~$K^4$ are independent and the edges between each pair of these sets forms a matching.
In fact, $K^1$ is anti-complete to~$K^2$ and~$K^3$ is anti-complete to~$K^4$.
Therefore $G''[K^1 \cup K^2 \cup K^3 \cup K^4]$ has maximum degree at most~$2$, and therefore clique-width at most~$4$ by Lemma~\ref{lem:atmost-2}.
It follows that~$G''$ has bounded clique-width and therefore~$G$ also has bounded clique-width.
This completes the proof of the lemma.\qedllncs
\end{proof}

\begin{lemma}\label{lem:2P1P3-co-2P1P3-C6-co-C6-prime-graphs}
Every prime $(2P_1+P_3,\overline{2P_1+P_3}, C_6, \overline{C_6})$-free graph is either $K_7$-free or $\overline{K_7}$-free.
\end{lemma}
\begin{proof}
\setcounter{ctrclaim}{0}
Let~$G$ be a prime $(2P_1+P_3,\overline{2P_1+P_3}, C_6, \overline{C_6})$-free graph.
Suppose, for contradiction, that~$G$ contains an induced~$K_7$ and an induced~$\overline{K_7}$.
We will show that in this case the graph~$G$ is not prime.
Note that any induced~$K_7$ and induced~$\overline{K_7}$ in~$G$ can share at most one vertex.
We may therefore assume that~$G$ contains a clique~$C$ on at least six vertices and a vertex-disjoint independent set~$I$ on at least six vertices.
Furthermore, we may assume that~$C$ is a maximum clique in $G \setminus I$ and~$I$ is a maximum independent set in $G \setminus C$ (if not, then replace~$C$ or~$I$ with a bigger clique or independent set, respectively).

By Lemma~\ref{lem:comp-anti}, there exist sets $R_1 \subset C$ and $R_2 \subset I$ each of size at most~$3$ such that $C' = C \setminus R_1$ is either complete or anti-complete to $I' = I \setminus R_2$.
Without loss of generality, we may assume that~$R_1$ and~$R_2$ are minimal, in the sense that the above property does not hold if we remove any vertex from~$R_1$ or~$R_2$.
Note that the class of prime $(2P_1+P_3,\overline{2P_1+P_3}, C_6, \overline{C_6})$-free graphs containing an induced~$K_7$ and an induced~$\overline{K_7}$ is closed under complementation.
Therefore, complementing~$G$ if necessary (in which case the sets~$I$ and~$C$ will be swapped, and the sets~$R_1$ and~$R_2$ will be swapped), we may assume that~$C'$ is anti-complete to~$I'$.

\clm{\label{clm:R1_R2_leq_1}$|R_1| \leq 1$ and $|R_2| \leq 1$.}
By construction, $R_1$ and~$R_2$ each contain at most three vertices and~$I'$ and~$C'$ each contain at least three vertices.
Since~$R_1$ (resp.~$R_2$) is minimal, every vertex of~$R_1$ (resp.~$R_2$) has at least one neighbour in~$I'$ (resp.~$C'$).

Choose $i_1,i_2 \in I'$ arbitrarily and suppose, for contradiction, that $y \in R_2$ is not complete to~$C'$.
Then~$y$ must have a neighbour $c_1 \in C'$ and a non-neighbour~$c_2 \in C'$, so $G[i_1,i_2,y,c_1,c_2]$ is a $2P_1+\nobreak P_3$, a contradiction.
Therefore~$R_2$ is complete to~$C'$.
If $y,y' \in R_2$ then for arbitrary $c_1 \in C'$, the graph $G[i_1,i_2,y,c_1,y']$ is a $2P_1+\nobreak P_3$, a contradiction.
It follows that $|R_2| \leq 1$.

Choose $c_1,c_2 \in C'$ arbitrarily.
Suppose, for contradiction, that $x \in R_1$ has two non-neighbours $i_1,i_2 \in I'$.
Recall that~$x$ must have a neighbour $i_3 \in I'$, so $G[i_1,i_2,i_3,x,c_1]$ is a $2P_1+\nobreak P_3$, a contradiction.
Therefore every vertex of~$R_1$ has at most one non-neighbour in~$I'$.
Suppose, for contradiction, that $x,x' \in R_1$.
Since~$I'$ contains at least three vertices, there must be a vertex $i_1 \in I'$ that is a common neighbour of~$x$ and~$x'$.
Now $G[x,x',c_1,i_1,c_2]$ is a $\overline{2P_1+P_3}$, a contradiction.
It follows that $|R_1| \leq 1$.
This completes the proof of Claim~\ref{clm:R1_R2_leq_1}.

\medskip
\noindent
Note that Claim~\ref{clm:R1_R2_leq_1} implies that $|C'| \geq 5$ and $|I'| \geq 5$.
Let~$A$ be the set of vertices in $V \setminus (C \cup I)$ that are complete to~$C'$.
If $x \in A$ is adjacent to~$y \in R_1$ then by Claim~\ref{clm:R1_R2_leq_1} $C \cup \{x\}$ is a bigger clique than~$C$, contradicting the maximality of~$C$.
It follows that~$A$ is anti-complete to~$R_1$.
If $x,y \in A$ are adjacent then by Claim~\ref{clm:R1_R2_leq_1}, $(C \cup \{x,y\})\setminus R_1$ is a bigger clique than~$C$, contradicting the maximality of~$C$.
It follows that~$A$ is an independent set.
Furthermore, by the maximality of~$I$ and the definition of~$A$, every vertex in $V \setminus (C \cup I \cup A)$ has a neighbour in~$I$ and non-neighbour in~$C'$.

\begin{sloppypar}
\clm{\label{clm:partition}Let~$x$ be a vertex in $V \setminus (C \cup I \cup A)$.
Then either~$x$ is complete to~$I'$, or~$x$ has exactly one neighbour in~$I$.}
Suppose, for contradiction, that~$x$ has a non-neighbour~$z$ in~$I'$, and two neighbours $y,y' \in\nobreak I$.
Now~$x$ cannot have another non-neighbour $z' \in I\setminus\{z\}$, otherwise $G[z,z',y,x,y']$ would be a $2P_1+\nobreak P_3$.
Therefore~$x$ must be complete to $I \setminus \{z\}$.
In particular, since $|I'| \geq 5$, this means that~$x$ has two neighbours in~$I'$, say~$y_1$ and~$y_2$ (not necessarily distinct from~$y$ and~$y'$).
Recall that~$x$ must have a non-neighbour $c_1 \in C'$.
Now $G[c_1,z,y_1,x,y_2]$ is a $2P_1+\nobreak P_3$.
This contradiction completes the proof of Claim~\ref{clm:partition}.
\end{sloppypar}

\medskip
\noindent
By Claim~\ref{clm:partition} we can partition the vertex set $V \setminus (C \cup I \cup A)$ into subsets~$V_{I'}$ and~$V_x$ for every $x \in I$, where~$V_{I'}$ is the set of vertices that are complete to~$I'$, and~$V_x$ is the set of vertices whose unique neighbour in~$I$ is~$x$.
Let $U_x=V_x \cup \{x\}$.

\clm{\label{clm:anti-complete_to_C}For all $x \in I'$, $U_x$ is anti-complete to~$C'$.}
Suppose $x \in I'$.
Clearly~$x$ is anti-complete to~$C'$.
Suppose, for contradiction, that $y \in U_x\setminus \{x\} = V_x$ has a neighbour $z \in C'$ and choose $u,v \in I' \setminus \{x\}$.
Then $G[u,v,x,y,z]$ is a $2P_1+\nobreak P_3$.
This contradiction completes the proof of Claim~\ref{clm:anti-complete_to_C}.

\clm{\label{clm:Ux_is_clique}For every $x \in I$, the set~$U_x$ is a clique.}
Note that $x \in I$ is adjacent to all other vertices of~$U_x$, by definition.
If $y,z \in V_x$ are non-adjacent then $(I \setminus \{x\}) \cup \{y,z\}$ would be a bigger independent set than~$I$.
This contradiction completes the proof of Claim~\ref{clm:Ux_is_clique}.

\clm{\label{clm:Ux_anti-complete_to_Uy}If $x,y \in I$ are distinct, then~$U_x$ is anti-complete to~$U_y$.}
Clearly~$x$ is anti-complete to~$U_y$ and~$y$ is anti-complete to~$U_x$.
Suppose, for contradiction, that $x' \in U_x \setminus \{x\}$ is adjacent to $y' \in U_y \setminus \{y\}$.
Choose $u,v \in I \setminus \{x,y\}$.
Then $G[u,v,x,x',y']$ is a $2P_1+\nobreak P_3$.
This contradiction completes the proof of Claim~\ref{clm:Ux_anti-complete_to_Uy}.

\clm{\label{clm:Ux_complete_to_V_I}For every $x \in I'$, the set~$U_x$ is complete to~$V_{I'}$.}
Clearly~$x$ is complete to~$V_{I'}$, by definition.
Suppose, for contradiction that $x' \in U_x \setminus \{x\}$ is non-adjacent to $y \in V_{I'}$.
Since $y \notin A$, the vertex~$y$ must have a non-neighbour $c_1 \in C'$ and note that~$x'$ is non-adjacent to~$c_1$ by Claim~\ref{clm:anti-complete_to_C}.
Choose $u,v \in I' \setminus \{x\}$.
Then $G[c_1,x',u,y,v]$ is a $2P_1+\nobreak P_3$.
This contradiction completes the proof of Claim~\ref{clm:Ux_complete_to_V_I}.

\medskip
\noindent
Suppose $x \in I'$.
Claim~\ref{clm:Ux_is_clique} implies that~$U_x$ is a clique, Claim~\ref{clm:anti-complete_to_C} implies that~$U_x$ is anti-complete to~$C'$ and Claim~\ref{clm:Ux_complete_to_V_I} implies that~$U_x$ is complete to~$V_{I'}$.
Furthermore for all $y \in I \setminus \{x\}$, Claim~\ref{clm:Ux_anti-complete_to_Uy} implies that~$U_x$ is anti-complete to~$U_y$.
We conclude that given any two vertices $x,y \in I'$, no vertex in $V \setminus (A \cup R_1 \cup U_x \cup U_y)$ can distinguish the set $U_x \cup U_y$.
In the remainder of the proof, we will show that there exist $x,y \in I'$ such that no vertex of $A \cup R_1$ distinguishes the set $U_x \cup U_y$, meaning that $U_x \cup U_y$ is a non-trivial module, contradicting the assumption that~$G$ is prime.

\clm{\label{clm:y_anti-or-almost-complete_to_Ux}If $u \in A \cup R_1$ then either~$u$ is anti-complete to~$U_x$ for all $x \in I'$ or else~$u$ is complete to~$U_x$ for all but at most one~$x \in I'$.}
Suppose, for contradiction, that the claim does not hold for a vertex $u \in A \cup R_1$.
Then~$u$ must have a neighbour $x' \in U_x$ for some $x \in I'$ and must have non-neighbours $y' \in U_y$ and $z' \in U_z$ for some $y,z \in I'$ with $y \neq z$.
Since $|I'| \geq 5$, we may also assume that $x \notin \{y,z\}$.
Choose $c_1 \in C'$ arbitrarily.
By Claim~\ref{clm:anti-complete_to_C}, $c_1$ is non-adjacent to~$x'$, $y'$ and~$z'$.
It follows that $G[y',z',c_1,u,x']$ is a $2P_1+\nobreak P_3$.
This contradiction completes the proof of Claim~\ref{clm:y_anti-or-almost-complete_to_Ux}.

\medskip
\noindent
Let~$A^*$ denote the set of vertices in $A \cup R_1$ that have a neighbour in~$U_x$ for some $x \in I'$.

\clm{\label{clm:A_cup_R1_complete_to_Ux}The set $A^*$ is complete to all, except possibly two, sets $U_x, x \in I'$.}
Suppose, for contradiction, that there are three different vertices $x,y,z \in I'$ such that~$A^*$ is complete to none of the sets $U_x, U_y$, and~$U_z$.
By Claim~\ref{clm:y_anti-or-almost-complete_to_Ux} and the definition of~$A^*$, every vertex in~$A^*$ is complete to at least two of the sets $U_x, U_y, U_z$.
Therefore there exist three vertices $u,v,w \in A^*$ such that:
\begin{itemize}
\item $u$ is not adjacent to some vertex $x' \in U_x$, but is complete to~$U_y$ and~$U_z$;
\item $v$ is not adjacent to some vertex $y' \in U_y$, but is complete to~$U_x$ and~$U_z$;
\item $w$ is not adjacent to some vertex $z' \in U_z$, but is complete to~$U_x$ and~$U_y$.
\end{itemize}
Therefore $G[u,y',w,x',v,z']$ is a~$C_6$.
This contradiction completes the proof of Claim~\ref{clm:A_cup_R1_complete_to_Ux}.

\medskip
\noindent
Now, since $|I'| \geq 5$, Claims~\ref{clm:y_anti-or-almost-complete_to_Ux} and~\ref{clm:A_cup_R1_complete_to_Ux} imply that there exist two distinct vertices $x,y \in I'$ such that every vertex of $A \cup R_1$ is either complete or anti-complete to $U_x \cup U_y$.
It follows that $U_x \cup U_y$ is a non-trivial module in~$G$, contradicting the fact that~$G$ is prime.
This completes the proof.\qedllncs
\end{proof}

Let~$G$ be a graph.
The \emph{chromatic number}~$\chi(G)$ of~$G$ is the minimum positive integer~$k$ such that~$G$ is $k$-colourable.
The \emph{clique number}~$\omega(G)$ of~$G$ is the size of a largest clique in~$G$.
A class~${\cal C}$ of graphs is {\em $\chi$-bounded} if there is a function~$f$ such that $\chi(G) \leq f(\omega(G))$ for all $G \in {\cal C}$.

\begin{lemma}[\cite{Gy87}]\label{lem:Pk-free-chi-bounded}
For every natural number~$k$ the class of $P_k$-free graphs is $\chi$-bounded.
\end{lemma}

\begin{lemma}\label{lem:2P1P3-co-2P1P3-Kk-free-bdd-cw}
For every fixed~$k\geq 1$, the class of $(K_k,2P_1+\nobreak P_3,\overline{2P_1+P_3})$-free graphs has bounded clique-width.
\end{lemma}

\begin{proof}
Fix a constant~$k\geq 1$ and let~$G$ be a $(K_k,2P_1+\nobreak P_3,\overline{2P_1+P_3})$-free graph.
By Lemma~\ref{lem:2P1+P_3-co-2P_1+P_3-free-c6-non-free}, we may assume that~$G$ is $C_6$-free.
Since~$G$ is $(2P_1+\nobreak P_3)$-free, it is $P_7$-free, so by Lemma~\ref{lem:Pk-free-chi-bounded} it has chromatic number at most~$\ell$ for some constant~$\ell$.
This means that we can partition the vertices of~$G$ into~$\ell$ independent sets $V_1,\ldots,V_\ell$ (some of which may be empty).
By Lemma~\ref{lem:matching-comatching}, deleting finitely many vertices (which we may do by Fact~\ref{fact:del-vert}), we may assume that for all distinct $i,j \in \{1,\ldots,\ell\}$, the edges between~$V_i$ and~$V_j$ form a matching or a co-matching.
Since~$G$ is $C_6$-free, if the vertices between~$V_i$ and~$V_j$ form co-matching, this co-matching can contain at most two non-edges.
Therefore, by deleting finitely many vertices (which we may do by Fact~\ref{fact:del-vert}), we may assume that the edges between~$V_i$ and~$V_j$ form a matching or~$V_i$ and~$V_j$ are complete to each other.
By deleting finitely many vertices (which we may do by Fact~\ref{fact:del-vert}), we may assume that each set~$V_i$ is either empty or contains at least five vertices.

Suppose the edges from~$V_i$ to~$V_j$ and the edges from~$V_i$ to~$V_k$ form a matching and that there is a vertex $x \in V_i$ that has a neighbour $y \in V_j$ and a neighbour $z \in V_k$.
Then~$y$ must be adjacent to~$z$, otherwise for $x',x'' \in V_i \setminus \{x\}$ the graph $G[x',x'',y,x,z]$ would be a $2P_1+\nobreak P_3$, a contradiction.
If~$V_j$ is complete to~$V_k$ then for $y',y'' \in V_j$, $z' \in V_k$ and $x',x'' \in V_i \setminus (N(y') \cup N(y'') \cup N(z'))$ (such vertices exist since each of $y',y''$ and~$z'$ have at most one neighbour in~$V_i$ and~$V_i$ contains at least five vertices) we have $G[x',x'',y',z',y'']$ is a $2P_1+\nobreak P_3$, a contradiction.
Therefore the edges between~$V_j$ and~$V_k$ form a matching.

Now for each $i,j \in \{1,\ldots,\ell\}$ with $i<j$, if~$V_i$ is complete to~$V_j$, we apply a bipartite complementation between~$V_i$ and~$V_j$ (which we may do by Fact~\ref{fact:bip}).
Let~$G'$ be the resulting graph.
The previous paragraph implies that if~$x$ has two neighbours~$y$ and~$z$ in~$G'$ then~$y$ is adjacent to~$z$ in~$G$, so $G'$ is $P_3$-free.
Therefore~$G'$ is a disjoint union of cliques, so it has clique-width at most~$2$.\qedllncs
\end{proof}

We are now ready to prove our main result.

\begin{theorem}\label{thm:2P1P3-co-2P1P3-bounded-cw}
The class of $(2P_1+P_3,\overline{2P_1+P_3})$-free graphs has bounded clique-width.
\end{theorem}
\begin{proof}
Let~$G$ be a $(2P_1+P_3,\overline{2P_1+P_3})$-free graph.
By Lemma~\ref{lem:prime}, we may assume that~$G$ is prime.
If~$G$ contains an induced~$C_6$ then we are done by Lemma~\ref{lem:2P1+P_3-co-2P_1+P_3-free-c6-non-free}.
If~$G$ contains an induced~$\overline{C_6}$ then we are done by Lemma~\ref{lem:2P1+P_3-co-2P_1+P_3-free-c6-non-free} and Fact~\ref{fact:comp}.
We may therefore assume that~$G$ is also $(C_6,\overline{C_6})$-free.
By Lemma~\ref{lem:2P1P3-co-2P1P3-C6-co-C6-prime-graphs}, we may assume that~$G$ is either $K_7$-free or $\overline{K_7}$-free.
By Fact~\ref{fact:comp}, we may assume that~$G$ is $K_7$-free.
Lemma~\ref{lem:2P1P3-co-2P1P3-Kk-free-bdd-cw} completes the proof.\qedllncs
\end{proof}

We combine Theorem~\ref{thm:2P1P3-co-2P1P3-bounded-cw} with Theorem~\ref{thm:classification2} to obtain Theorem~\ref{t-main}, which we restate below, followed by a short proof.

\medskip
\noindent
\faketheorem{Theorem~\ref{t-main} (restated).} 
{\it For a graph~$H$, the class of $(H,\overline{H})$-free graphs has bounded clique-width if and only if~$H$ or~$\overline{H}$ is an induced subgraph of $K_{1,3}$, $P_1+P_4$, $2P_1+P_3$ or~$sP_1$ for some $s\geq 1$.}

\begin{proof}
Let~$H$ be a graph.
It can be readily checked from Theorem~\ref{thm:classification2} that if~$H$ or~$\overline{H}$ is an induced subgraph of $K_{1,3}$, $P_1+P_4$, $2P_1+P_3$ or~$sP_1$ for some $s\geq 1$ then the class of $(H,\overline{H})$-free graphs has bounded clique-width.
Suppose this is not the case.
If $H\notin {\cal S}$ and $\overline{H}\notin {\cal S}$, then the class of $(H,\overline{H})$-free graphs has unbounded clique-width by Theorem~\ref{thm:classification2}.
By Fact~\ref{fact:comp}, we may therefore assume that $H\in {\cal S}$.
By Lemma~\ref{lem:useful}, $H$ contains $K_{1,3}+\nobreak P_1$, $2P_2$, $3P_1+\nobreak P_2$ or~$S_{1,1,2}$ as an induced subgraph.
This means that the class of $(H,\overline{H})$-free graphs contains the class $(K_{1,3},\overline{4P_1})$-free, $(2P_2,\overline{2P_2})$-free, $(4P_1,\overline{3P_1+P_2})$-free or $(2P_1+P_2,\overline{K_{1,3}})$-free graphs, respectively.
In each of these cases we apply Theorem~\ref{thm:classification2}.\qedllncs
\end{proof}

\section{Three New Classes of Unbounded Clique-Width and the Proof of Theorem~\ref{t-main2}}\label{s-main2}

In this section we first identify three new graph classes of unbounded clique-width.
To do so, we will need the notion of a {\em wall}.
We do not formally define this notion, but instead refer to \figurename~\ref{fig:wall}, in which three examples of walls of different height are depicted (see e.g.~\cite{Chuzhoy15} for a formal definition).
The class of walls is well known to have unbounded clique-width; see for example~\cite{KLM09}.
A \emph{$k$-subdivided wall} is the graph obtained from a wall after subdividing each edge exactly~$k$ times for some constant $k\geq 0$, and the following lemma is well known.

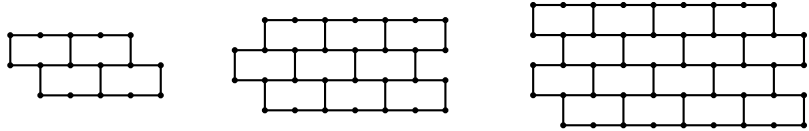
\begin{figure}
\begin{center}
\begin{minipage}{0.2\textwidth}
\centering
\begin{tikzpicture}[scale=0.4, every node/.style={scale=0.3}]
\GraphInit[vstyle=Simple]
\SetVertexSimple[MinSize=6pt]
\Vertex[x=1,y=0]{v10}
\Vertex[x=2,y=0]{v20}
\Vertex[x=3,y=0]{v30}
\Vertex[x=4,y=0]{v40}
\Vertex[x=5,y=0]{v50}

\Vertex[x=0,y=1]{v01}
\Vertex[x=1,y=1]{v11}
\Vertex[x=2,y=1]{v21}
\Vertex[x=3,y=1]{v31}
\Vertex[x=4,y=1]{v41}
\Vertex[x=5,y=1]{v51}

\Vertex[x=0,y=2]{v02}
\Vertex[x=1,y=2]{v12}
\Vertex[x=2,y=2]{v22}
\Vertex[x=3,y=2]{v32}
\Vertex[x=4,y=2]{v42}

\Edges(    v10,v20,v30,v40,v50)
\Edges(v01,v11,v21,v31,v41,v51)
\Edges(v02,v12,v22,v32,v42)

\Edge(v01)(v02)

\Edge(v10)(v11)

\Edge(v21)(v22)

\Edge(v30)(v31)

\Edge(v41)(v42)

\Edge(v50)(v51)

\end{tikzpicture}
\end{minipage}
\begin{minipage}{0.3\textwidth}
\centering
\begin{tikzpicture}[scale=0.4, every node/.style={scale=0.3}]
\GraphInit[vstyle=Simple]
\SetVertexSimple[MinSize=6pt]
\Vertex[x=1,y=0]{v10}
\Vertex[x=2,y=0]{v20}
\Vertex[x=3,y=0]{v30}
\Vertex[x=4,y=0]{v40}
\Vertex[x=5,y=0]{v50}
\Vertex[x=6,y=0]{v60}
\Vertex[x=7,y=0]{v70}

\Vertex[x=0,y=1]{v01}
\Vertex[x=1,y=1]{v11}
\Vertex[x=2,y=1]{v21}
\Vertex[x=3,y=1]{v31}
\Vertex[x=4,y=1]{v41}
\Vertex[x=5,y=1]{v51}
\Vertex[x=6,y=1]{v61}
\Vertex[x=7,y=1]{v71}

\Vertex[x=0,y=2]{v02}
\Vertex[x=1,y=2]{v12}
\Vertex[x=2,y=2]{v22}
\Vertex[x=3,y=2]{v32}
\Vertex[x=4,y=2]{v42}
\Vertex[x=5,y=2]{v52}
\Vertex[x=6,y=2]{v62}
\Vertex[x=7,y=2]{v72}

\Vertex[x=1,y=3]{v13}
\Vertex[x=2,y=3]{v23}
\Vertex[x=3,y=3]{v33}
\Vertex[x=4,y=3]{v43}
\Vertex[x=5,y=3]{v53}
\Vertex[x=6,y=3]{v63}
\Vertex[x=7,y=3]{v73}

\Edges(    v10,v20,v30,v40,v50,v60,v70)
\Edges(v01,v11,v21,v31,v41,v51,v61,v71)
\Edges(v02,v12,v22,v32,v42,v52,v62,v72)
\Edges(    v13,v23,v33,v43,v53,v63,v73)

\Edge(v01)(v02)

\Edge(v10)(v11)
\Edge(v12)(v13)

\Edge(v21)(v22)

\Edge(v30)(v31)
\Edge(v32)(v33)

\Edge(v41)(v42)

\Edge(v50)(v51)
\Edge(v52)(v53)

\Edge(v61)(v62)

\Edge(v70)(v71)
\Edge(v72)(v73)
\end{tikzpicture}
\end{minipage}
\begin{minipage}{0.35\textwidth}
\centering
\begin{tikzpicture}[scale=0.4, every node/.style={scale=0.3}]
\GraphInit[vstyle=Simple]
\SetVertexSimple[MinSize=6pt]
\Vertex[x=1,y=0]{v10}
\Vertex[x=2,y=0]{v20}
\Vertex[x=3,y=0]{v30}
\Vertex[x=4,y=0]{v40}
\Vertex[x=5,y=0]{v50}
\Vertex[x=6,y=0]{v60}
\Vertex[x=7,y=0]{v70}
\Vertex[x=8,y=0]{v80}
\Vertex[x=9,y=0]{v90}

\Vertex[x=0,y=1]{v01}
\Vertex[x=1,y=1]{v11}
\Vertex[x=2,y=1]{v21}
\Vertex[x=3,y=1]{v31}
\Vertex[x=4,y=1]{v41}
\Vertex[x=5,y=1]{v51}
\Vertex[x=6,y=1]{v61}
\Vertex[x=7,y=1]{v71}
\Vertex[x=8,y=1]{v81}
\Vertex[x=9,y=1]{v91}

\Vertex[x=0,y=2]{v02}
\Vertex[x=1,y=2]{v12}
\Vertex[x=2,y=2]{v22}
\Vertex[x=3,y=2]{v32}
\Vertex[x=4,y=2]{v42}
\Vertex[x=5,y=2]{v52}
\Vertex[x=6,y=2]{v62}
\Vertex[x=7,y=2]{v72}
\Vertex[x=8,y=2]{v82}
\Vertex[x=9,y=2]{v92}

\Vertex[x=0,y=3]{v03}
\Vertex[x=1,y=3]{v13}
\Vertex[x=2,y=3]{v23}
\Vertex[x=3,y=3]{v33}
\Vertex[x=4,y=3]{v43}
\Vertex[x=5,y=3]{v53}
\Vertex[x=6,y=3]{v63}
\Vertex[x=7,y=3]{v73}
\Vertex[x=8,y=3]{v83}
\Vertex[x=9,y=3]{v93}

\Vertex[x=0,y=4]{v04}
\Vertex[x=1,y=4]{v14}
\Vertex[x=2,y=4]{v24}
\Vertex[x=3,y=4]{v34}
\Vertex[x=4,y=4]{v44}
\Vertex[x=5,y=4]{v54}
\Vertex[x=6,y=4]{v64}
\Vertex[x=7,y=4]{v74}
\Vertex[x=8,y=4]{v84}

\Edges(    v10,v20,v30,v40,v50,v60,v70,v80,v90)
\Edges(v01,v11,v21,v31,v41,v51,v61,v71,v81,v91)
\Edges(v02,v12,v22,v32,v42,v52,v62,v72,v82,v92)
\Edges(v03,v13,v23,v33,v43,v53,v63,v73,v83,v93)
\Edges(v04,v14,v24,v34,v44,v54,v64,v74,v84)

\Edge(v01)(v02)
\Edge(v03)(v04)

\Edge(v10)(v11)
\Edge(v12)(v13)

\Edge(v21)(v22)
\Edge(v23)(v24)

\Edge(v30)(v31)
\Edge(v32)(v33)

\Edge(v41)(v42)
\Edge(v43)(v44)

\Edge(v50)(v51)
\Edge(v52)(v53)

\Edge(v61)(v62)
\Edge(v63)(v64)

\Edge(v70)(v71)
\Edge(v72)(v73)

\Edge(v81)(v82)
\Edge(v83)(v84)

\Edge(v90)(v91)
\Edge(v92)(v93)
\end{tikzpicture}
\end{minipage}
\caption{\label{fig:wall}Walls of height $2$, $3$ and~$4$, respectively.}
\end{center}
\end{figure}

\begin{lemma}[\cite{LR06}]\label{lem:wall}
For any constant $k\geq 0$, the class of $k$-subdivided walls has unbounded clique-width.
\end{lemma}

\begin{sloppypar}
In~\cite{DGP14}, Dabrowski, Golovach and Paulusma showed that $(4P_1,\overline{3P_1+P_2})$-free graphs have unbounded clique-width.
However, their construction was not $C_5$-free.
We give an alternative construction that neither contains an induced~$C_5$ nor an induced copy of any larger self-complementary graph (see also \figurename~\ref{fig:unbdd-forb}).
\end{sloppypar}

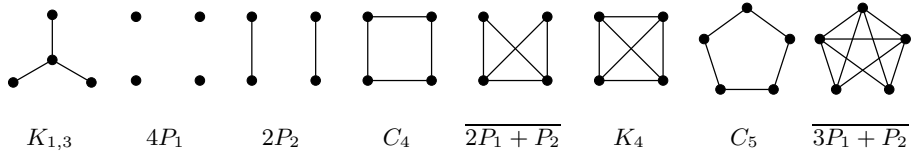
\begin{figure}
\begin{center}
\begin{tabular}{cccccccc}
\begin{minipage}{0.11\textwidth}
\centering
\scalebox{0.6}{
{\begin{tikzpicture}[scale=1,rotate=90]
\GraphInit[vstyle=Simple]
\SetVertexSimple[MinSize=6pt]
\Vertices{circle}{a,b,c}
\Vertex[x=0,y=0]{d}
\Edges(a,d,b)
\Edges(c,d)
\end{tikzpicture}}}
\end{minipage}
&
\begin{minipage}{0.11\textwidth}
\centering
\scalebox{0.6}{
{\begin{tikzpicture}[scale=1,rotate=45]
\GraphInit[vstyle=Simple]
\SetVertexSimple[MinSize=6pt]
\Vertices{circle}{a,b,c,d}
\end{tikzpicture}}}
\end{minipage}
&
\begin{minipage}{0.11\textwidth}
\centering
\scalebox{0.6}{
{\begin{tikzpicture}[scale=1,rotate=45]
\GraphInit[vstyle=Simple]
\SetVertexSimple[MinSize=6pt]
\Vertices{circle}{a,b,c,d}
\Edges(a,d)
\Edges(b,c)
\end{tikzpicture}}}
\end{minipage}
&
\begin{minipage}{0.11\textwidth}
\centering
\scalebox{0.6}{
{\begin{tikzpicture}[scale=1,rotate=45]
\GraphInit[vstyle=Simple]
\SetVertexSimple[MinSize=6pt]
\Vertices{circle}{a,b,c,d}
\Edges(a,b,c,d,a)
\end{tikzpicture}}}
\end{minipage}
&
\begin{minipage}{0.11\textwidth}
\centering
\scalebox{0.6}{
{\begin{tikzpicture}[scale=1,rotate=45]
\GraphInit[vstyle=Simple]
\SetVertexSimple[MinSize=6pt]
\Vertices{circle}{a,b,c,d}
\Edges(b,c,d,a,c)
\Edges(b,d)
\end{tikzpicture}}}
\end{minipage}
&
\begin{minipage}{0.11\textwidth}
\centering
\scalebox{0.6}{
{\begin{tikzpicture}[scale=1,rotate=45]
\GraphInit[vstyle=Simple]
\SetVertexSimple[MinSize=6pt]
\Vertices{circle}{a,b,c,d}
\Edges(a,b,c,d,a,c)
\Edges(b,d)
\end{tikzpicture}}}
\end{minipage}
&
\begin{minipage}{0.11\textwidth}
\centering
\scalebox{0.6}{
{\begin{tikzpicture}[scale=1,rotate=90]
\GraphInit[vstyle=Simple]
\SetVertexSimple[MinSize=6pt]
\Vertices{circle}{a,b,c,d,e}
\Edges(a,b,c,d,e,a)
\end{tikzpicture}}}
\end{minipage}
&
\begin{minipage}{0.11\textwidth}
\centering
\scalebox{0.6}{
{\begin{tikzpicture}[scale=1,rotate=90]
\GraphInit[vstyle=Simple]
\SetVertexSimple[MinSize=6pt]
\Vertices{circle}{a,b,c,d,e}
\Edges(d,e,a,b,c,e,b,d,a,c)
\end{tikzpicture}}}
\end{minipage}\\
\\
$K_{1,3}$ &
$4P_1$ &
$2P_2$ &
$C_4$ &
$\overline{2P_1+\nobreak P_2}$ &
$K_4$ &
$C_5$ &
$\overline{3P_1+P_2}$
\\
\end{tabular}
\end{center}
\caption{\label{fig:unbdd-forb}Forbidden induced subgraphs from Theorems~\ref{thm:4P1-co-3P1+P2-C5-unbdd-cw}, \ref{thm:claw-K4-diamond-C5-unbdd-cw} and~\ref{thm:C4-2P_2-F-free-unbdd-cw}.}
\end{figure}
\begin{theorem}\label{thm:4P1-co-3P1+P2-C5-unbdd-cw}
Let~${\cal F}$ be the set of all self-complementary graphs on at least five vertices that are not equal to the bull.
The class of $(\{4P_1,\overline{3P_1+P_2}\} \cup {\cal F})$-free graphs has unbounded clique-width.
\end{theorem}
\begin{proof}
Consider a wall~$H$ (see also \figurename~\ref{fig:wall}).
Let~$H'$ be the graph obtained from~$H$ by subdividing every edge once, that is, $H'$ is a $1$-subdivided wall.
By Lemma~\ref{lem:wall}, such graphs have unbounded clique-width.
Let~$V_1$ be the set of vertices in~$H'$ that are also present in~$H$.
Let~$V_2$ be the set of vertices obtained from subdividing vertical edges in~$H$, and let~$V_3$ be the set of vertices obtained from subdividing horizontal edges.
Note that~$V_1$, $V_2$ and~$V_3$ are independent sets.
Furthermore, every vertex in~$V_1$ has at most one neighbour in~$V_2$ and at most two neighbours in~$V_3$, while every vertex in $V_2 \cup V_3$ has at most two neighbours, each of which is in~$V_1$.
Let~$H''$ be the graph obtained from~$H'$ by applying complementations on~$V_1$, $V_2$ and~$V_3$.
By Fact~\ref{fact:comp}, such graphs have unbounded clique-width.

It remains to show that~$H''$ is $(\{4P_1,\overline{3P_1+P_2}\} \cup {\cal F})$-free.
Since the vertex set of~$H''$ is the disjoint union of the three cliques $V_1$, $V_2$, $V_3$, we find that~$H''$ is $4P_1$-free.
We now show that~$H''$ is $\overline{3P_1+P_2}$-free.
Suppose, for contradiction, that~$H''$ contains an induced $\overline{3P_1+P_2}$ with vertex set~$S$.
Note that $\overline{3P_1+P_2}$ is the graph obtained from a~$K_5$ after deleting an edge.
Since~$V_1$, $V_2$ and~$V_3$ are cliques in~$H''$, this means that~$S$ must have vertices in at least two of these sets.
As~$V_2$ is anti-complete to~$V_3$, we find that~$S$ contains at least one vertex of~$V_1$.
Suppose~$S$ has only vertices in~$V_1$ and~$V_i$ for some $i\in \{2,3\}$.
Then, since $|S|=5$ and~$S$ contains vertices of both~$V_1$ and~$V_i$, it follows that one of~$V_1$,~$V_i$ contains either three or four vertices.
As each vertex of~$V_1$ has at most two neighbours in~$V_i$ and vice versa, this means that in both cases~$H''[S]$ is missing at least two edges, which is not possible.
Hence~$S$ must have at least one vertex in each of $V_1$, $V_2$, $V_3$.
As~$V_2$ is anti-complete to~$V_3$ and~$H''[S]$ is missing only one edge, $S$ must have exactly one vertex in each of~$V_2$ and~$V_3$, and the remaining three vertices of~$S$ must be in~$V_1$.
We recall that every vertex in~$V_2$ has at most two neighbours in~$V_1$ and no neighbours in~$V_3$.
Hence, the vertex of~$S$ that is in~$V_2$ has degree at most~$2$ in~$H''[S]$.
This is a contradiction, as $\overline{3P_1+P_2}$ has minimum degree~$3$.
We conclude that~$H''$ is $\overline{3P_1+P_2}$-free.

Next, we show that~$H''$ is $X$-free for any self-complementary graph~$X$ on at least five vertices that is not equal to the bull.
First suppose that~$X$ has at most seven vertices.
Then~$X$ must be the~$C_5$.
Since the vertex set of~$H''$ is the disjoint union of the three cliques $V_1$, $V_2$, $V_3$, and $V_2$ is anti-complete to~$V_3$, we find that $H''$ is $C_5$-free.

It remains to show that if~$X$ is a self-complementary graph on at least eight vertices, then~$H''$ is $X$-free.
Suppose, for contradiction, that~$H''$ contains such an induced subgraph~$X$.
Since~$H''$ is $4P_1$-free, $X$ must be~$4P_1$-free and therefore $K_4$-free.
Let $U_i=V_i\cap V(F)$ for $i=1,2,3$.
Since each set~$V_i$ is a clique and~$X$ is $K_4$-free, each set~$U_i$ must be a clique on at most three vertices.
Since~$X$ contains at least eight vertices, two sets of $\{U_1,U_2,U_3\}$ consist of three vertices and the other set consists of either two or three vertices.
This means that at least one of $U_2$, $U_3$ contains three vertices (while the other set may contain two vertices). 
Now~$U_2$ is anti-complete to~$U_3$, so~$X$ contains an induced~$K_3+\nobreak P_2$.
Since~$X$ is self-complementary, $X$ must contain an induced~$K_{2,3}$.
Consider an induced~$K_{1,3}$ in~$H''$.
The three degree-$1$ vertices of the~$K_{1,3}$ must be in different sets~$V_i$.
As~$V_2$ is anti-complete to~$V_3$, the central vertex of the~$K_{1,3}$ must be in~$V_1$.
Since two non-adjacent vertices in a~$K_{2,3}$ are the centres of an induced~$K_{1,3}$, it follows that both these vertices must be contained in~$H''[V_1]$.
This is a contradiction, since~$V_1$ is a clique.
This completes the proof.\qedllncs
\end{proof}

Brandst{\"a}dt et al.~\cite{BELL06} proved that $(C_4,K_{1,3}, K_4, \overline{2P_1+P_2})$-free graphs have unbounded clique-width.
In fact, their construction is also~$C_5$-free (see also \figurename~\ref{fig:unbdd-forb}).
Note that by Lemma~\ref{lem:ramsey-for-self-comp}, any self-complementary graph on at least five vertices that is not equal to the bull contains an induced subgraph isomorphic to~$C_4$, $C_5$ or~$K_4$, so such graphs are automatically excluded from the class specified in the following theorem.
\begin{sloppypar}
\begin{theorem}\label{thm:claw-K4-diamond-C5-unbdd-cw}
The class of $(C_4, C_5, K_{1,3}, K_4, \overline{2P_1+P_2})$-free graphs has unbounded clique-width.
\end{theorem}
\end{sloppypar}
\begin{proof}
Consider a wall~$H$ (see also \figurename~\ref{fig:wall}).
Let~$H'$ be the graph obtained from~$H$ by subdividing every edge once.
Let~$V_1$ be the set of vertices in~$H'$ that are also present in~$H$ and let~$V_2$ be the set of vertices obtained by subdividing edges of~$H$.
Note that in~$H'$, the neighbourhood of every vertex in~$V_1$ is an independent set.
For every vertex in~$V_1$, we add edges between its neighbours; this will cause its neighbourhood to induce a~$P_2$ or a triangle.
Finally, we delete the vertices in~$V_1$ and let~$H''$ be the resulting graph.
As the smallest induced cycle in~$H$ has length~$6$, the smallest induced cycle in~$H'$ has length~$12$. 
Hence, making vertices that are of distance~$2$ from each other in~$H'$ adjacent to each other in~$H''$ does not create any induced~$C_5$, and we find that~$H''$ is $C_5$-free.
Brandst{\"a}dt et al. proved that such graphs are $(C_4, K_{1,3}, K_4, \overline{2P_1+P_2})$-free and have unbounded clique-width~\cite{BELL06}; the latter fact also follows from~\cite[Theorem~3]{DP15}.
This completes the proof.\qedllncs
\end{proof}

Before proving our third unboundedness result, we will first need to introduce some more terminology and two lemmas.
Given natural numbers~$k,\ell$, let~$Rb(k,\ell)$ denote the smallest number such that if every edge of a~$K_{Rb(k,\ell),Rb(k,\ell)}$ is coloured red or blue then it will contain a monochromatic~$K_{k,\ell}$.
Beineke and Schwenk~\cite{BS75} showed that~$Rb(k,\ell)$ always exists and proved the following result.

\begin{lemma}[\cite{BS75}]\label{lem:Rb22=5}
$Rb(2,2)=5$.
\end{lemma}

The next lemma was independently proved by Ringel and Sachs.
\begin{lemma}[\cite{Ri63,Sa62}]\label{lem:odd-self-comp-fixed-vertex}
Let~$G$ be a self-complementary graph on an odd number of vertices and let $f:\nobreak V(G) \to V(G)$ be an isomorphism from~$G$ to~$\overline{G}$.
Then there is a unique vertex $v \in V(G)$ such that $f(v)=v$.
\end{lemma}

Recall that the clique number~$\omega(G)$ of~$G$ is the size of a largest clique in~$G$.
The next lemma was proved by Sridharan and Balaji.

\begin{lemma}[\cite{SB98}]\label{lem:even-self-comp-omega}
Let~$G$ be a self-complementary split graph on~$n$ vertices.
If~$n$ is even then $\omega(G)=\frac{n}{2}$.
\end{lemma}

Let $G=(V,E)$ be a split graph.
The {\em independence number}~$\alpha(G)$ of~$G$ is the size of a largest independent set in~$G$.
By definition, $G$ has a {\em split partition}, that is, a partition of~$V$ into two (possibly empty) sets~$C$ and~$I$, where~$C$ is a clique and~$I$ is an independent set.
A split graph~$G$ may have multiple split partitions.
For self-complementary split graphs we can show the following.

\begin{lemma}\label{lem:split-self-comp-graph-properties}
Let~$G$ be a self-complementary split graph on~$n$ vertices.
\begin{enumerate}[(i)]
\item If~$n$ is even, then~$G$ has a unique split partition and in this partition the clique and independent set are of equal size.
\item If~$n$ is odd, then there exists a vertex~$v$ such that $G \setminus v$ is also a self-complementary split graph.
\end{enumerate}
\end{lemma}

\begin{proof}
First consider the case where~$n$ is even and let $(C,I)$ be a split partition of~$G$.
Then $\omega(G) \geq |C|$ and $\alpha(G) \geq |I|$.
Since~$G$ is self-complementary, it follows that $\omega(G) = \alpha(G)$ and, by Lemma~\ref{lem:even-self-comp-omega}, $\omega(G)=\frac{n}{2}$.
Therefore $n=|C|+|I| \leq \omega(G) + \alpha(G) \leq n$ and so $|I|=|C|=\omega(G) = \alpha(G) = \frac{n}{2}$. 
Suppose, for contradiction, that there is another split partition $(C', I')$ with $|C'| = |I'|$.
Now $I \setminus I' \subseteq C' \cap I$, so $|I \setminus I'|=1$.
Similarly $| I' \setminus I | = 1$.
This implies that $I \setminus I' = C' \setminus C = \{ u \}$ and $I' \setminus I = C \setminus C' = \{ w \}$.
Note that both~$u$ and~$w$ are complete to $C \cap C'$ and anti-complete to $I \cap I'$.
Hence if $u,w$ are adjacent, then $\{ u,w \} \cup (C \cap C')$ is a clique of size $\frac{n}{2}+1$, and if $u,w$ are non-adjacent, then $\{ u,w \} \cup (I \cap I')$ is an independent set of size $\frac{n}{2}+1$.
In both cases we get a contradiction with the fact that $\omega(G) = \alpha(G) = \frac{n}{2}$.
This completes the case where~$n$ is even.

Now suppose that~$n$ is odd and let~$f$ be an isomorphism from~$G$ to~$\overline{G}$.
By Lemma~\ref{lem:odd-self-comp-fixed-vertex}, there is a vertex $v \in V(G)$ such that~$f$ maps~$v$ to~$v$ and maps the vertices of $G \setminus \{v\}$ to the vertices of $\overline{G} \setminus \{v\}$.
Therefore $G \setminus \{v\}$ is self-complementary.\qedllncs
\end{proof}

\begin{theorem}\label{thm:C4-2P_2-F-free-unbdd-cw}
Let~${\cal F}$ be the set of all self-complementary graphs on at least five vertices that are not equal to the bull.
The class of $(\{C_4,2P_2\} \cup {\cal F})$-free graphs has unbounded clique-width.
\end{theorem}
\begin{proof}
First note that the only self-complementary graph on five vertices apart from the bull is the~$C_5$.
Since $C_5 \in {\cal F}$, the class of $(\{C_4,2P_2\} \cup {\cal F})$-free graphs is a subclass of the class of split graphs by Lemma~\ref{lem:split-forb-graphs}.
Hence we may remove all graphs that are not split graphs from~${\cal F}$, apart from~$C_5$; in particular, this means that we remove $X_4,\ldots,X_{10}$ from~${\cal F}$ (see also \figurename~\ref{fig:self-complementary-8-vertices}).
By Lemma~\ref{lem:split-self-comp-graph-properties}, if $G \in {\cal F}$ has an odd number of vertices, but is not equal to~$C_5$, then $G\setminus v \in {\cal F}$ for some vertex $v \in V(G)$.
Let~${\cal F'}$ be the set of self-complementary split graphs on at least eight vertices that have an even number of vertices.
It follows that the class of ${\cal F'}$-free split graphs is equal to the class of $(\{C_4,2P_2\} \cup {\cal F})$-free graphs.

Consider a $2$-subdivided wall~$H$ and note that it is $(C_4,C_8)$-free; recall that $2$-subdivided walls have unbounded clique-width by Lemma~\ref{lem:wall}.
Note that~$H$ is a bipartite graph, and fix a bipartition $(A,B)$ of~$H$.
Let~$H'$ be the graph obtained from~$H$ by applying a complementation to~$A$ and note that~$H'$ is a split graph with split partition $(A,B)$.

Note that in~$H'$, every vertex in~$B$ has a non-neighbour in~$A$ and every vertex in~$A$ has a neighbour in~$B$.
We claim that $(A,B)$ is the unique split partition of~$H'$.
Indeed, suppose $(A',B')$ is a split partition of~$H'$; we will show that $A'=A$ and $B'=B$.
First suppose, for contradiction, that~$B'$ contains a vertex $a\in A$.
There exists a vertex $b\in B$ that is adjacent to~$a$ and therefore $b \notin B'$.
Similarly, every vertex of $(A\setminus \{a\})$ is adjacent to~$a$, so $(A\setminus \{a\}) \cap B' = \emptyset$.
Therefore every vertex of $(A\setminus \{a\})\cup \{b\}$ must lie in~$A'$, which is a clique.
It follows that~$b$ is complete to~$A$, a contradiction.
Therefore~$B'$ cannot contain a vertex of~$A$.
Similarly, $A'$ cannot contain a vertex of~$B$.
We conclude that $(A',B')=(A,B)$, so the split partition of~$H'$ is indeed unique.

Now, by Fact~\ref{fact:comp}, the class of graphs~$H'$ produced in the above way also has unbounded clique-width.
It remains to show that~$H'$ is ${\cal F'}$-free.

First note that~$X_1$ (see also \figurename~\ref{fig:self-complementary-8-vertices}) is the graph obtained from the bipartite graph~$C_8$ by complementing one of the independent sets in the bipartition.
Since~$H$ is $C_8$-free and~$X_1$ has a unique split partition (by Lemma~\ref{lem:split-self-comp-graph-properties}), it follows that~$H'$ is $X_1$-free.

Note that~$H$ is $C_4$-free and so~$H'$ does not contain two vertices $x,x'$ in the clique~$A$ and two vertices $y,y'$ in the independent set~$B$ such that $\{x,x'\}$ is complete to~$\{y,y'\}$.
Now suppose $G \in {\cal F'} \setminus \{X_1\}$.
Recall that by Lemma~\ref{lem:split-self-comp-graph-properties}, $G$ has a unique split partition~$(C,I)$, and this partition has the property that $|C|=|I|$.
Therefore, if we can show that~$G$ contains two vertices $x,x' \in C$ and two vertices $y,y' \in I$ with $\{x,x'\}$ complete to~$\{y,y'\}$ then~$H'$ must be $G$-free and the proof is complete.
This is easy to verify that this is the case if $G \in \{X_2,X_3\}$ (see also \figurename~\ref{fig:self-complementary-8-vertices} and recall that $X_4,\ldots,X_{10} \notin {\cal F'}$).
Otherwise, $G$ has at least ten vertices so $|C|,|I| \geq 5$.
By Lemma~\ref{lem:Rb22=5}, there must be two vertices $x,x' \in C$ and two vertices $y,y' \in I$ with $\{x,x'\}$ either complete or anti-complete to~$\{y,y'\}$.
In the first case we are done.
In the second case we note that complementing~$G$ will swap the sets~$C$ and~$I$ and make $\{x,x'\}$ complete to~$\{y,y'\}$, returning us to the previous case.

We conclude that~$H'$ is indeed ${\cal F'}$-free.
This completes the proof.\qedllncs
\end{proof}

We are now ready to prove Theorem~\ref{t-main2}.
Note that this theorem holds even if~${\cal F}$ is infinite.

\medskip
\noindent
\faketheorem{Theorem~\ref{t-main2} (restated).}
{\it Let~${\cal F}$ be a set of self-complementary graphs on at least five vertices not equal to the bull.
For a graph~$H$, the class of $(\{H,\overline{H}\} \cup {\cal F})$-free graphs has bounded clique-width if and only if~$H$ or~$\overline{H}$ is an induced subgraph of $K_{1,3}$, $P_1+P_4$, $2P_1+P_3$ or~$sP_1$ for some $s\geq 1$.}
\begin{proof}
Let~$H$ be a graph.
By Theorem~\ref{t-main}, if~$H$ or~$\overline{H}$ is an induced subgraph of $K_{1,3}$, $P_1+P_4$, $2P_1+P_3$ or~$sP_1$ for some $s\geq 1$, then the class of $(\{H,\overline{H}\} \cup {\cal F})$-free graphs has bounded clique-width.

Consider a graph $F \in {\cal F}$.
Since~$F$ contains at least five vertices and is not isomorphic to the bull, Lemma~\ref{lem:ramsey-for-self-comp} implies that~$F$ contains an induced subgraph isomorphic to $C_4$, $C_5$ or~$K_4$, and so $F \notin {\cal S}$.
Therefore the class of $(\{H,\overline{H}\} \cup {\cal F})$-free graphs contains the class of $(H,\overline{H},C_4,C_5,K_4)$-free graphs.
If $H\notin {\cal S}$ and $\overline{H}\notin {\cal S}$, then the class of $(H,\overline{H},C_4,C_5,K_4)$-free graphs has unbounded clique-width by Lemma~\ref{lem:classS}.
By Fact~\ref{fact:comp}, we may therefore assume that $H\in {\cal S}$.
By Lemma~\ref{lem:useful}, we may assume~$H$ contains $K_{1,3}+\nobreak P_1$, $2P_2$, $3P_1+\nobreak P_2$ or~$S_{1,1,2}$ as an induced subgraph, otherwise we are done.
In this case, the class of $(\{H,\overline{H}\} \cup {\cal F})$-free graphs contains the class of
$(K_{1,3},K_4,C_4,C_5)$-free, $(\{2P_2,C_4\} \cup {\cal F})$-free, $(\{4P_1,\overline{3P_1+P_2}\} \cup {\cal F})$-free or $(K_{1,3},\overline{2P_1+P_2},C_4,C_5,K_4)$-free graphs, respectively.
These classes have unbounded clique-width by Theorems~\ref{thm:claw-K4-diamond-C5-unbdd-cw}, \ref{thm:C4-2P_2-F-free-unbdd-cw}, \ref{thm:4P1-co-3P1+P2-C5-unbdd-cw} and~\ref{thm:claw-K4-diamond-C5-unbdd-cw}, respectively.
This completes the proof.\qedllncs
\end{proof}

\bibliography{mybib}

\end{document}